\documentclass[11pt, a4paper, oneside]{report}
\pdfoutput=1
\usepackage{preamble}
\begin{document}

\begin{titlepage}
\AddToShipoutPicture*{\put(0,0){\includegraphics*{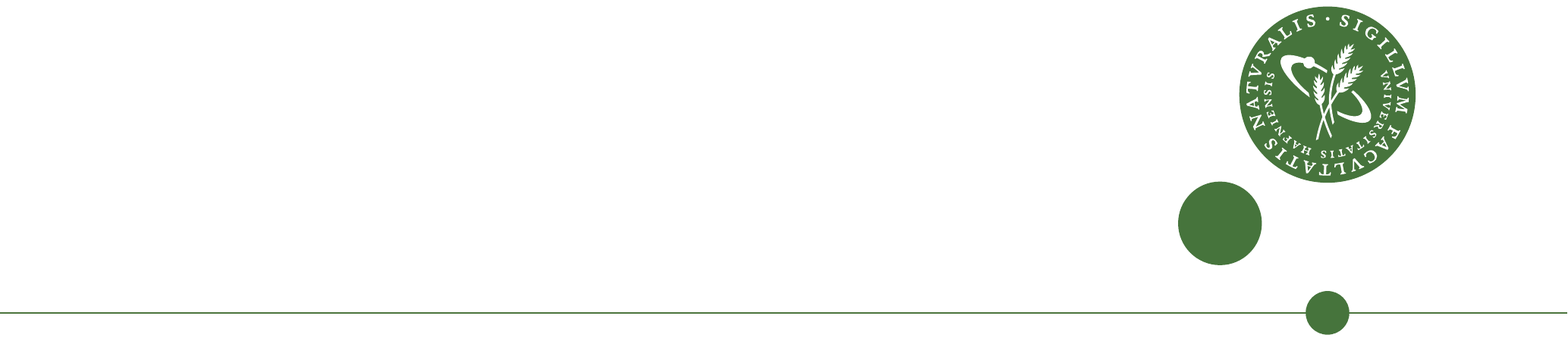}}}
\AddToShipoutPicture*{\put(0,0){\includegraphics*{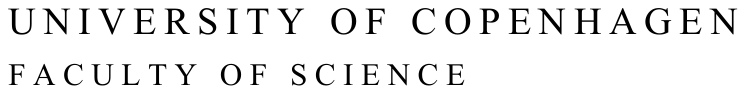}}}

\begin{flushleft}
\vspace*{3cm}
\textbf{\huge{Master's Thesis}}

\vspace*{3mm}
\textbf{Martin Holm Cservenka} \\
Department of Computer Science \\
University of Copenhagen \\
\texttt{djp595@alumni.ku.dk}

\vspace*{4cm}
\textbf{\huge{Design and Implementation of Dynamic Memory Management in a Reversible Object-Oriented Programming Language}}
\vfill
\textbf{Supervisors:} Robert Glück \& Torben Ægidius Mogensen

\textbf{Submitted:} January $25^{th}$, 2018

Revision 1.01
\end{flushleft}
\end{titlepage}

\newpage 

\begin{versionhistory}
  \vhEntry{1.01}{2018-04-11}{Martin}{Fixed incorrect subfigures (d, e, f) in Figure~\ref{fig:deallocation-order-free-list}. Updated Figure~\ref{fig:injective-garbage-in-out} to better emphesize output copying. Updated Figure~\ref{fig:equivalent-free-lists} to better reflect the harmless garbage flow through programs.}

\end{versionhistory} 

\newpage

\chapter*{Abstract}
\markright{Abstract}

The reversible object-oriented programming language (\textsc{Roopl}) was presented in late 2016 and proved that object-oriented programming paradigms works in the reversible setting. The language featured simple statically scoped objects which made non-trivial programs tedious, if not impossible to write using the limited tools provided.
We introduce an extension to \textsc{Roopl} in form the new language \rooplpp, featuring dynamic memory management and fixed-sized arrays for increased language expressiveness. The language is a superset of \textsc{Roopl} and has formally been defined by its language semantics, type system and computational universality. Considerations for reversible memory manager layouts are discussed and ultimately lead to the selection of the Buddy Memory layout. Translations of the extensions added in \rooplpp to the reversible assembly language \textsc{Pisa} are presented to provide garbage-free computations. The dynamic memory management extension successfully increases the expressiveness of \textsc{Roopl} and as a result, shows that non-trivial reversible data structures, such as binary trees and doubly-linked lists, are feasible and do not contradict the reversible computing paradigm.
\newpage 

\chapter*{Preface}
\markright{Preface}

This Master's Thesis is submitted as the last part for the degree of Master of Science in Computer Science at the University of Copenhagen, Department of Computer Science, presenting a 30 ECTS workload.

The thesis consists of \pageref*{LastPage} pages and a ZIP archive containing source code and test programs developed as part of the thesis work.

I would like to thank my two supervisors, Robert Glück and Torben Mogensen, for their invaluable supervision and guidance throughout this project and introduction to the field of reversible computing. A big thanks to my university colleague and friend, Tue Haulund, for allowing me to continue his initial work on \textsc{Roopl} and providing information, sparring and source code material and for being a great ally through our years at the University of Copenhagen. In addition, thanks to my dear aunt Doris, for financially supporting my studies by paying for all my books needed. Finally, a thanks to Jess, for all the love and support throughout the entire span of my thesis process.
 
\newpage

\tableofcontents
\newpage

\listoffigures 
\newpage
 
\chapter{Introduction}
\label{chp:introduction}
In recent years, technologies such as cloud-based services, deep learning, cryptocurrency mining and other services requiring large computational power and availability have been on the rise. Most of these services are hosted on massive server parks, consuming immense amounts of electricity in order to power the machines and the cooling architectures as heat dissipates from the hardware. A recent study showed that the Bitcoin network including its mining processes' currently stands at 0.13\% of the total global electricity consumption, rivaling the usage of a small country like Denmark's~\cite{digiconomist:bitcoin}. With the recent years focus on climate and particularly energy consumption, companies have started to attempt to reduce their power usage in these massive server farms. As an example, Facebook built new server park in the arctic circle in 2013, in an attempt to take advantage of the natural surroundings in the cooling architecture to reduce its power consumption~\cite{bloomberg:facebook}. 

Reversible computing presents a possible solution the problematic power consumption issues revolving around computations. Traditional, irreversible computers dissipates heat during their computation. Landauer's principle states that deletion of information in a system always results by an increase in energy consumption. In reversible computing, all information is preserved throughout the execution, and as such, the energy consumption theoretically should be smaller~\cite{rl:irreversibility}.

Currently, reversible computing is not commercially appealing, as it is an area which still is being actively researched. However, several steps has been taken in the direction of a fully reversible system, which some day might be applicable in a large setting. Reversible machine architectures have been presented such as the Pendulum architecture and its instruction set Pendulum ISA (\textsc{Pisa})~\cite{cv:pendulum, ha:architecture} and the \textsc{BobISA} architecture and instruction set~\cite{mt:bob} and high level languages \textsc{Janus}~\cite{cl:janus, ty:janus, ty:ejanus} and \textsc{R}~\cite{mf:r} exists. 

While cryptocurrency mining and many other computations are not reversible, the area remains interesting in terms of its applications and gains.

\section{Reversible Computing}
\label{sec:reversible-computing}
Reversible computing is a two-directional computational model in which all processes are time-invertible. This means, that at any time during execution, the computation can return to a former state. In order to maintain \textit{reversibility}, the reversible computational model cannot compute \textit{many-to-one} functions, as the models requires an exact inverse $f^{-1}$ of a function $f$ in order to support backwards determinism. Therefore, reversible programs must only consist of \textit{one-to-one} functions, also known as \textit{injective} functions, which result in a garbage-free computation, as garbage-generating functions simply can be unwinded to clean up.

Each step of a reversible program is locally invertible, meaning each step has exactly one inverse step. A reversible program can be inverted simply by computing the inverse of each of its steps, without any knowledge about the overall functionality or requirements of the program. This property immediately yields interesting consequences in terms of software development, as an encryption or compression algorithm implemented in a reversible language immediately yields the decryption or decompression algorithm by running the algorithm backwards.

The reversibility is however not free and comes and the cost of strictness when writing programs. Almost every popular, irreversible programming language features a conditional component in form of \textbf{if}-\textbf{else}-statements. In these languages, we only define the \textit{entry}-condition in the conditional, that is, the condition that determines which branch of the component we continue execution in. In reversible languages, we must also specify an \textit{exit}-condition, such that we can determine which branch we should follow, when executing the program in reverse. In theory, this sounds trivial, but in practice it turns to add a new layer of complexity when writing programs.

\section{Object-Oriented Programming}
\label{sec:object-oriented-programming}
Object-oriented programming (OOP) has for many years been the most widely used programming paradigm as reflected in the popular usage of object-oriented programming languages, such as the \textsc{C}-family languages, \textsc{Java}, \textsc{PHP} and in recent years \textsc{Javascript} and \textsc{Python}. The OOP core concepts such as \textit{inheritance}, \textit{encapsulation} and \textit{polymorphism} allows complex systems to be modeled by breaking the system into smaller parts in form of abstract objects~\cite{jm:concepts}.

\section{Reversible Object-Oriented Programming}
\label{sec:reversible-object-oriented-programming}
The high-level reversible language \textsc{Roopl} (Reversible Object-Oriented Programming Language) was introduced in late 2016~\cite{th:roopl, th:roopl2}. The language extends the design of previously existing reversible imperative languages with object-oriented programming language features such as user-defined data types, class inheritance and subtype-polymorphism. As a first, \textsc{Roopl} successfully integrates the object-oriented programming (OOP) paradigms into the reversible computation setting using a static memory manager to maintain garbage-free computation, but at cost of programmer usability as objects only lives within \textbf{construct} / \textbf{deconstruct} blocks, which needs to be predefined, as the program call stack is required to be reset before program termination.

Conceptualizations and ideas for the \textsc{Joule} language was also published in 2016~\cite{us:joule}. The language, a homonym of \textsc{Janus Object-Oriented Language}, \textsc{Jool}, presented an alternative OOP extension to \textsc{Janus}, differing from \textsc{Roopl}. The language featured heap allocated objects with constructors and multiple object references, as such also addressing the problems with \textsc{Roopl}. The language is still a work in progress, aiming to provide a useful, reversible object oriented-programming language.

\section{Motivation}
\label{sec:motivation}
The block defined objects of \textsc{Roopl} and lack of multiple references are problematic when writing complex, reversible programs using OOP methodologies as they pose severe limitations on the expressiveness. It has therefore been proposed to extend and partially redesign the language with dynamic memory management in mind, such that these shortcomings can be addressed, and ultimately increase the usability of reversible OOP. Work within the field of reversible computing related to heap manipulation~\cite{ha:heap}, reference counting~\cite{tm:refcounting} and garbage collection~\cite{tm:garbage} suggests that a \textsc{Roopl} extension is feasible.

\section{Thesis Statement}
\label{sec:thesis-statement}
An extension of the reversible object-oriented programming language with dynamic memory management is feasible and effective. The resulting expressiveness allows non-trivial reversible programming previously unseen, such as reversible data structures, including linked lists, doubly linked lists and trees.

\section{Outline}
\label{sec:outline}
This Master's thesis consists of four chapters, besides the introductory chapter. The following summary describes the following chapters.
\begin{itemize}
    \item \textbf{Chapter 2} formally defines the \textsc{Roopl} extension exemplified by the new language \rooplpp, a superset of \textsc{Roopl}.
    \item \textbf{Chapter 3} serves as a brief description of dynamic memory management along with a discussion of various reversible, dynamic memory management layouts.
    \item \textbf{Chapter 4} presents the translation techniques utilized in compiling a \rooplpp program to \textsc{Pisa} instructions.
    \item \textbf{Chapter 5} presents the conclusions of the thesis and future work proposals.
\end{itemize}

Besides the five chapters, a number of appendices is supplied, containing \textsc{Pisa} translations of the reversible heap allocation algorithm, the source code of the \rooplpp to \textsc{Pisa} compiler, the \rooplpp source code for the example programs and their translated \textsc{Pisa} versions.
\newpage

\chapter{The \textsc{Roopl\texttt{++}} Language}
\label{chp:rooplpp}
With the design and implementation of the \textsc{Reversible Object-Oriented Programming Language} (\textsc{Roopl}) and the work-in-progress report of \textsc{Joule}, the first steps into the uncharted lands of Object-Oriented Programming (OOP) and reversibility were taken. In this chapter, we present \rooplpp, the natural successor to \textsc{Roopl}, improving the object instantiation of the language by letting objects live outside \textbf{construct}/\textbf{deconstruct} blocks, allowing complex, reversible programs to be written using OOP methodologies. As with its predecessor, \rooplpp is purely reversible and each component of a program written in \rooplpp is locally invertible. This ensures no computation history is required, nor added program size for backwards program execution.

Inspired by other language successors such as \textsc{C\texttt{++}} was to \textsc{C}, \rooplpp is a superset of \textsc{Roopl}, containing all original functionality of its predecessor, extended with new object instantiation methods for increased programming usability and an array type.

\vskip 2em

\begin{figure}[ht]
    \centering
\begin{lstlisting}[language=roopl, style=basic, frame=none, multicols=2]
class Fib
    int[] xs

    method init()
        new int[2] xs
    
    method fib(int n)
        if n = 0 then
            xs[0] ^= 1
            xs[1] ^= 1
        else
            n -= 1
            call fib(n)
            xs[0] += xs[1]
            xs[0] <=> xs[1]
        fi xs[0] = xs[1]
    
    method get(int out)
        out ^= xs[1]

class Program
    int result
    int n
    
    method main()
        n ^= 4

        new Fib f
        call f::init()
        call f::fib(n)
        call f::get(result)
        uncall f::fib(n)
        uncall f::init()
        delete Fib f

\end{lstlisting}
    \caption{Example \rooplpp program implementing the Fibonacci function}    
\end{figure}
\newpage

\section{Syntax}
\label{sec:syntax}
A \rooplpp program consists, analogously to a \textsc{Roopl} program, of one or more class definitions, each with a varying number of fields and class methods. The entry point of the program is a nullary main method, which is defined exactly once and is instantiated during program start-up. Fields of the main object will serve as output of the program, just as in \textsc{Roopl}.
\begin{figure}[h]
    \centering
    \vspace{3mm}
    \textbf{\rooplpp Grammar}
    \begin{align}
    prog		\quad&::= \quad cl^+ \tag{program}\\
    cl			\quad&::=\quad \textbf{class}\ c\ (\textbf{inherits}\ c)^?\ (t\ x)^*\ m^+\tag{class definition}\\
    d           \quad&::=\quad c\ |\ c[e]\ |\ \textbf{int}[e] \tag{class and arrays}\\
    t			\quad&::=\quad \textbf{int}\ |\ c\ |\ \textbf{int}[]\ |\ c[]\tag{data type}\\
    y          \quad&::=\quad x\ |\ x[e] \tag{variable identifiers}\\
    m			\quad&::=\quad \textbf{method}\ q\textbf{\texttt{(}}t\ x,\ \dots,\ t\ x\textbf{\texttt{)}}\ s\tag{method}\\
    s			\quad&::=\quad y\ \odot\textbf{\texttt{=}}\ e\ |\ y\ \textbf{\texttt{<=>}}\ y\tag{assignment}\\\
    			&\ |\ \qquad \textbf{if}\ e\ \textbf{then}\ s\ \textbf{else}\ s\ \textbf{fi}\ e\tag{conditional}\\
    			&\ |\ \qquad \textbf{from}\ e\ \textbf{do}\ s\ \textbf{loop}\ s\ \textbf{until}\ e\tag{loop}\\
                &\ |\ \qquad \textbf{construct}\ c\ x\quad s\quad\textbf{destruct}\ x\tag{object block}\\
                &\ |\ \qquad \textbf{local}\ t\ x\ \texttt{=}\ e\quad s\quad\textbf{delocal}\ t\ x\ \texttt{=}\ e\tag{local variable block}\\
                &\ |\ \qquad \textbf{new}\ d\ y\ |\ \textbf{delete}\ d\ y \tag{object con- and destruction}\\
                &\ |\ \qquad \textbf{copy}\ d\ y\ y\ |\ \textbf{uncopy}\ d\ y\ y \tag{reference con- and destruction}\\
    			&\ |\ \qquad \textbf{call}\ q\textbf{\texttt{(}}x,\ \dots,\ x\textbf{\texttt{)}}\ |\ \textbf{uncall}\ q\textbf{\texttt{(}}x,\ \dots,\ x\textbf{\texttt{)}}\tag{local method invocation}\\
    			&\ |\ \qquad \textbf{call}\ y\textbf{\texttt{::}}q\textbf{\texttt{(}}x,\ \dots,\ x\textbf{\texttt{)}}\ |\ \textbf{uncall}\ y\textbf{\texttt{::}}q\textbf{\texttt{(}}x,\ \dots,\ x\textbf{\texttt{)}}\tag{method invocation}\\
    			&\ |\ \qquad \textbf{skip}\ |\ s\ s\tag{statement sequence}\\
    e			\quad&::=\quad \overline{n}\ |\ x\ |\ x[e]\ |\ \textbf{\texttt{nil}}\ |\ e\ \otimes\ e\tag{expression}\\
    \odot	\quad&::=\quad \textbf{\texttt{+}}\ |\ \textbf{\texttt{-}}\ |\ \textbf{\texttt{\^}}\tag{operator}\\
    \otimes\quad&::=\quad \odot\ |\ \textbf{\texttt{*}}\ |\ \textbf{\texttt{/}}\ |\ \textbf{\texttt{\%}}\ |\ \textbf{\texttt{\&}}\ |\ \textbf{\texttt{|}}\ |\ \textbf{\texttt{\&\&}}\ |\ \textbf{\texttt{||}}\ |\ \textbf{\texttt{<}}\ |\ \textbf{\texttt{\textgreater}}\ |\ \textbf{\texttt{=}}\ |\ \textbf{\texttt{!=}}\ |\ \textbf{\texttt{<=}}\ |\ \textbf{\texttt{\textgreater=}}\tag{operator}
    \end{align}
    \vspace{2mm}
    \textbf{Syntax Domains}
    \begin{align*}
    prog &\in \text{Programs} & s &\in \text{Statements}      & n &\in \text{Constants} \\
    cl   &\in \text{Classes}  & e &\in \text{Expressions}     & x &\in \text{VarIDs}    \\
    t    &\in \text{Types}    & \odot &\in \text{ModOps}      & q &\in \text{MethodIDs} \\
    m    &\in \text{Methods}  & \otimes &\in \text{Operators} & c &\in \text{ClassIDs}
    \end{align*}
    \caption{Syntax domains and EBNF grammar for \rooplpp}
    \label{fig:roopl-grammar}
\end{figure}

The \rooplpp grammar extends the grammar of \textsc{Roopl}  with a new static integer or class array type and a new object lifetime option in form of objects outside of blocks, using the \textbf{new} and \textbf{delete} approach. Furthermore, the local block extension proposed in~\cite{th:roopl} has become a standard part of the language. Class definitions remains unchanged, and consists of a \textbf{class} keyword followed by a class name. Subclasses must be specified using the \textbf{inherits} keyword and a following parent class name. Classes can have any number of fields of any of the data types, including the new Array type. A class definition is required to include at least one method, defined by the \textbf{method} keyword followed by a method name, a comma-separated list of parameters and a body.

Reversible assignments for integer variables and integer array elements uses similar syntax as \textsc{Janus} assignments, by updating a variable through any of the addition (\texttt{+=}), subtraction (\texttt{-=}) or bitwise XOR (\texttt{\^{}=}) operators. As with \textsc{Janus}, when updating a variable $x$ using any of said operators, the right-hand side of the operator argument must be entirely independent of $x$ to maintain reversibility. Usage of these reversible assignment operators for object or array variables is undefined. Variables and array elements of any type can be swapped using the \texttt{<=>} operator as long as the variable is of same type as the array type. If an array is of a base class type, subclass variable values can be swapped in and out of the array, as long as the resulting value in the variable is still of the original subclass type.  

\rooplpp objects can be instantiated in two ways. Either using object blocks known from \textsc{Roopl}, or by using the \textbf{new} statement. The object-blocks have a statically-scoped lifetime, as the object only exists within the \textbf{construct} and \textbf{destruct} segments. Using \textbf{new} allows the object to live until program termination, if the program terminates with a \textbf{delete} call. By design, it is the programmers responsibility to deallocate objects instantiated by the \textbf{new} statement.

Arrays are also instantiated by usage of \textbf{new} and \textbf{delete}. Assignment of array cells depend on the type of the arrays, which is further discussed in section~\ref{sec:array-instantiation}.

The methodologies for argument aliasing and its restrictions on method on invocations from \textsc{Roopl} carries over in \rooplpp and object fields are as such disallowed as arguments to local methods to prevent irreversible updates and non-local method calls to a passed objects are prohibited. The parameter passing scheme remains call-by-reference and the object model of \textsc{Roopl} remains largely unchanged in \rooplpp.

\section{Object Instantiation}
\label{sec:object-instantiation}
Object instantiation through the \textbf{new} statement follows the pattern of the mechanics known from the \textbf{construct}/\textbf{destruct} blocks from \textsc{Roopl}, but providing improved scoping and lifetime options objects. The mechanisms of the statement 

\begin{align*}
\textbf{construct}\ c\ x\ \quad s\ \quad \textbf{destruct}\ x
\end{align*}

are as follows:

\begin{enumerate}
\item Memory for an object of class $c$ is allocated. All fields are automatically zero-initialized by virtue of residing in already zero-cleared memory.
\item The block statement $s$ is executed, with the name $x$ representing a reference to the newly allocated object.
\item The reference $x$ may be modified by swapping its value with that of other references of the same type, but it should be restored to its original value within the statement block $s$, otherwise the meaning of the object block is undefined.
\item Any state that is accumulated within the object should be cleared or uncomputed before the end of the statement is reached, otherwise the meaning of the object block is undefined.
\item The zero-cleared memory is reclaimed by the system.
\end{enumerate}

The statement pair consisting of

\begin{align*}
\textbf{new}\ c\ x\ \quad \dots\ \quad \textbf{delete}\ c\ x\
\end{align*}

could be considered a \textit{dynamic} block, meaning we can have overlapping blocks. Compared to \textbf{construct}/\textbf{destruct} block consisting of a single statement, the \textbf{new}/\textbf{delete} block consist of two separate statements. We can as such initialize an object $x$ of class $c$ and an object $y$ of class $d$ and destroy $x$ before we destroy $y$, a feature that was not possible in \textsc{Roopl}. The mechanisms of the \textbf{new} statement are as follows:

\begin{enumerate}
    \item Memory for an object of class $c$ is allocated. All fields are automatically zero-initialized by virtue of residing in already zero-cleared memory.
    \item The address of the newly allocated block is stored in the previously defined and zero-cleared reference $x$.
\end{enumerate}

and the mechanisms of the \textbf{delete} statement are as follow

\begin{enumerate}
    \item The reference $x$ may be modified by swapping its internal field values with that of other references of the same type, but should be zero-cleared before a \textbf{delete} statement is called on $x$, otherwise the meaning of the object deletion is undefined.
    \item Any state that is accumulated within the object should be cleared or uncomputed before the \textbf{delete} statement is executed, otherwise the meaning of the object block is undefined.
    \item The zero-cleared memory is reclaimed by the system.
\end{enumerate}

The mechanisms of the \textbf{new} and \textbf{delete} statements are, essentially, a split of the mechanisms of the \textbf{construct}/\textbf{destruct} blocks into two separate statements. As with \textsc{Roopl}, fields must be zero-cleared after object deletion, otherwise it is impossible for the system to reclaim the memory reversibly. This is the responsibility of the of the programmer to maintain this, and to ensure that objects are indeed deleted in the first place. A \textbf{new} statement without a corresponding \textbf{delete} statement targeting the same object further ahead in the program is undefined, as is a delete statement without a preceding \textbf{new} statement.

Note that variable scopes are always static, but object scopes can be either static (using \textbf{construct}/\textbf{destruct}) or dynamic (using \textbf{new}/\textbf{delete}).
\newpage

\section{Array Model}
\label{sec:array-model}
Besides asymmetric object lifetimes, \rooplpp also introduces reversible, fixed-sized arrays of either integer or object types. While \textsc{Roopl} only featured integers and custom data types in form of classes, one of its main inspirations, \textsc{Janus}, implemented static, reversible arrays~\cite{ty:janus}. 

While \textsc{Roopl} by design did not include any data storage language constructs, as they are not especially noteworthy nor interesting from an OOP perspective, they do generally improve the expressiveness of the language. Arrays were decided to be part of the core language for this reason, as one of the main goals of \rooplpp is increased expressiveness while implementing reversible programs.

In \rooplpp, arrays expand upon the array model from \textsc{Janus}. Arrays are indexed by integers, starting from 0. In \textsc{Janus}, only integer arrays were allowed, while in \rooplpp arrays of any type can be defined, meaning either integer arrays or custom data types in form of class arrays. They are however, still restricted to one dimension.

Array element accessing is accomplished using the bracket notation known from \textsc{Janus}. Accessing an out-of-bounds index is undefined.
Array instantiation and element assignments, aliasing and circularity is described in detail in the following section.

Arrays can contain elements of different classes sharing a base class, that is, say class $A$ and $B$ both inherit from some class $C$ and array $x$ is of type $C[]$. In this case, the array can hold elements of type $A$, $B$, and $C$. When swapping array elements from a base class array with object references, the programmer must be careful not to swap the values of, say, and $A$ object into a $B$ reference.

\section{Array Instantiation}
\label{sec:array-instantiation}
Array instantiation uses the \textbf{new} and \textbf{delete} keywords to reversibly construct and destruct array types. The mechanisms of the statement
\begin{align*}
\textbf{new}\ \textbf{int}[e]\ x
\end{align*} 
in which we reserve memory for an integer array are as follows

\begin{enumerate}
    \item The expression $e$ is evaluated
    \item Memory equal to the integer value that $e$ evaluates to and an additional small amount of memory for of overhead is reserved for the array.
    \item The address of the newly allocated memory is stored in the previously defined and zero-cleared reference $x$.
\end{enumerate}

In \rooplpp, we only allow instantiation of fixed-sized arrays of a length defined in the given expression $e$. Array elements are assigned dependent on the type of the array. For integer arrays, any of the reversible assignment operators can be used to assign values to cells. For class arrays, we assign cell elements a little differently. We either make use of the \textbf{new} and \textbf{delete} statements, but instead of specifying which variable should hold the newly created/deleted object or array, we specify which array cell it should be stored in or we use the \textbf{swap} statement to swap values in and out of array cells. Usage of the assignment operators on non-integer arrays is undefined.

\begin{lstlisting}[caption={Assignment of array elements}, language=roopl, style=basic,label={lst:array-assignment}]
    new int[5] intArray         // Init new integer array
    new Foo[2] fooArray         // init new Foo array

    intArray[1] += 10           // Legal array integer assignment
    intArray[1] -= 10           // Legal Zero-clearing for integer array cells

    new Foo fooObject
    fooArray[0] <=> fooObject   // Legal object array cell assignment
    new Foo fooArray[2]         // Legal object array cell assignment

    ...                         // Clear all array cells

    delete Foo fooArray[0]      // Legal object array cell zero-clearing
    delete Foo fooArray[1]      // Legal object array cell zero-clearing
\end{lstlisting}
 
As with \rooplpp objects instantiated outside of \textbf{construct}/\textbf{destruct} blocks, arrays must be deleted before program termination to reversibly allow the system to reclaim the memory. Before deletion of an array, all its elements must be zero-cleared such that no garbage data resides in memory after erasure of the array reference.

Consider the statement
\begin{align*}
\textbf{delete}\ \textbf{int}[e]\ x
\end{align*}

with the following mechanics

\begin{enumerate}
\item The reference $x$ may be modified by swapping, assigning cell element values and zero-clearing cell element values, but must be restored to an array of same type with fully zero-cleared cells before the \textbf{delete} statement. Otherwise, the meaning of the statement is undefined.
\item The value of the expression $e$ is evaluated and used to reclaim the allocated memory space. 
\item If the reference $x$ is a fully zero-cleared array upon the \textbf{delete} statement execution, the zero-cleared memory is reclaimed by the system.
\end{enumerate}

With reversible, fixed-sized arrays of varying types, we must be extremely careful when updating and assigning values, to ensure we maintain reversibility and avoid irreversible statements. Therefore, when assigning or updating integer elements with one of the reversible assignment operators, we prohibit the cell value from being reference on the right hand side, meaning the following statement is prohibited
\begin{align*}
x[5]\ \texttt{+=}\ x[5] + 1 
\end{align*}

However, we do allow other initialized, non-zero-cleared array elements from the same array or arrays of same type to be referenced in the right hand side of the statement. As with regular assignment, we still prohibit the left side reference to occur in the ride side, meaning the following statements are also prohibited
\begin{align*}
x\ \texttt{+=}\ y[x]\\
y[x]\ \texttt{+=}\ x
\end{align*}

\newpage
\section{Referencing}
\label{sec:referencing-layout}
Besides the addition of dynamically lifetimed objects and arrays, \rooplpp also increases program flexibility by allowing multiple references to objects and arrays through the usage of the \textbf{copy} statement. Once instantiated through either a \textbf{new} or \textbf{construct/destruct} block, an object or array reference can be copied into another zero-cleared variable. The reference acts as a regular instance and can be modified through methods as per usual. To delete a reference, the logical inverse statement \textbf{uncopy} must be used.

The syntax for referencing consists of the statement
\begin{align*}
    \textbf{copy}\ c\ x\ x'
\end{align*}
which copies a reference of variable $x$, an instance of class or array $c$, and stores the reference in variable $x'$.

For deleting copies, the following statement is used
\begin{align*}
    \textbf{uncopy}\ c\ x\ x'
\end{align*}    
which simply zero-clears variable $x'$, which is a reference to variable $x$, an instance of class or array $c$.

The mechanism of the \textbf{copy} statement is simply as follows
\begin{enumerate}
    \item The memory address stored in variable $x$ is copied into the zero-cleared variable $x'$. If $x'$ is not zero-cleared or $x$ is not a class instance, then \textbf{copy} is undefined.
\end{enumerate}

The mechanism of the \textbf{uncopy} statement is simply as follows
\begin{enumerate}
    \item The memory address stored in variable $x'$ is zero-cleared if it matches the address stored in $x$. If $x'$ is not a copy of $x$ or $x$ has been zero-cleared before the \textbf{uncopy} statement is executed, said statement is undefined.
\end{enumerate}
As references do not require all fields or cells to be zero-cleared (as they are simple pointers to existing objects or arrays), the reversible programmer should carefully ensure that all references are un-copied before deleting said object or array, as copied references to cleared objects or arrays would be pointing to cleared memory, which might be used later by the system. These type of references are also known as \textit{dangling pointers}.

It should be noted, that from a language design perspective, it is the programmer's responsibility to ensure such situations do not occur. From an implementation perspective, such situations are usually checked by the compiler either statically during compilation or during the actual runtime of the program. This is addressed later in sections~\ref{sec:referencing} and~\ref{sec:error-handling}.

\newpage
\section{Local Blocks}
\label{sec:local-blocks}
The local block presented in the extended \textsc{Janus} in~\cite{ty:ejanus} consisted of a local variable allocation, a statement and a local variable deallocation. These local variable blocks add immense programmer usability as the introduce a form of reversible temporary variable. The \textsc{Roopl} compiler features support for local integer blocks, but not object blocks. In \rooplpp, local blocks can be instantiated with all of the languages variable types; integers, arrays and user-defined types in the form of objects.

Local integer blocks works exactly the same as in \textsc{Roopl} and \textsc{Janus}, where the local variable initialized will be set to the evaluated result of a given expression.

Local array and object blocks feature a number of different options. If a local array or object block is initialized with a \textbf{nil} value, the variable must afterwards be initialized using a \text{new} statement before any type-specific functionality is accessible. If the block is initiated with an existing object or array reference, the local variable essentially becomes a reference copy, analogous to a variable initialized from a \textbf{copy} statement. 

\begin{figure}[h]
    \centering
    \begin{equation*}
        \textbf{construct}\ c\ x\ \quad s\ \quad \textbf{destruct}\ x\ \qquad \overset{\textbf{def}}{\scalebox{1.8}{=}} \qquad \begin{array}{ l }
            \textbf{local}\ c\ x\ = \textbf{nil}\\
            \textbf{new}\ c\ x \quad s \quad \textbf{delete}\ c\ x\\
            \textbf{delocal}\ c\ x = \textbf{nil} 
        \end{array}
      \end{equation*}
    \caption{\textbf{construct}/\textbf{destruct}-blocks can be considered syntactic sugar}
    \label{fig:sugar-construct-destruct}
\end{figure}

For objects, the \textbf{construct}/\textbf{destruct}-blocks can be considered syntactic sugar for a local block defined with a \textbf{nil} value, containing a \textbf{new} statement in the beginning of its statement block and a \textbf{delete} statement in the very end, as shown in figure~\ref{fig:sugar-construct-destruct}.

As local array and object blocks allow freedom in terms of their interaction with other statements in the language, it is the programmer's responsibility that the local variable is deallocated using a correct expression at the end of the block definition. The value of the variable is a pointer to an object or an array. Said object or array must have all fields/cells zero-cleared before the pointer is zero-cleared at the end of the local block. If the pointer is at any point exchanged with the pointer of another object or array using the \textbf{swap} statement, the same conditions apply.  

\newpage
\section{\rooplpp Expressiveness}
\label{sec:rooplpp-expressiveness}
By introducing dynamic lifetime objects and by allowing objects to be referenced multiple times, we can express non-trivial reversible programs. To demonstrate the capacities, expressiveness and possibilities of \rooplpp, the following section presents previously unseen reversible data structures, which now are feasible, written in \rooplpp.

\subsection{Linked List}
\label{subsec:linked-list}
\citeauthor{th:roopl} presented a linked list implemented in \textsc{Roopl} in~\cite{th:roopl}. The implementation featured a \textit{ListBuilder} and a \textit{Sum} class, required to determine and retain the sum of a constructed linked list as the statically scoped object blocks of \textsc{Roopl} would deallocate automatically after building the full list. In \rooplpp, we do not face the same challenges and the implementation becomes much more straightforward. Figure~\ref{fig:linked-list-class} implements a \textit{LinkedList} class, which simply has the head of the list and the list length as its internal fields. For demonstration, the class allows extension of the list by either appending or prepending cell elements to the list. In either case, we first check if the \textit{head} field is initialized. If not, the cell we are either appending or prepending simply becomes the new head of the list. If we are appending a cell the Cell-class \textit{append} method is called on the \textit{head} cell with the new cell as its only argument. When prepending, the existing head is simply appended to the new cell and the new cell is set as head of the linked list.

\begin{figure}[ht!]
    \centering
    \begin{lstlisting}[style = basic, language = roopl] 
    class Cell
        Cell next
        int data
    
        method constructor(int value)
            data ^= value     
    
        method append(Cell cell)
            if next = nil & cell != nil then
                next <=> cell           // Store as next cell if current cell is end of list
            else skip
            fi next != nil & cell = nil
    
            if next != nil then
                call next::append(cell) // Recursively search until we reach end of list
            else skip
            fi next != nil
    \end{lstlisting}
    \caption{Linked List cell class}
    \label{fig:linked-list-cell-class}
\end{figure}

Figure~\ref{fig:linked-list-cell-class} shows the \textit{Cell} class of the linked list which has a \textit{next} and a \textit{data} field, a constructor and the \textit{append} method. The append method works by recursively looking through the linked cell nodes until we reach the end of the free list, where the \textit{next} field has not been initialized yet. When we find such a cell, we simply swap the contents of the \textit{next} and \textit{cell} variables, s.t. the cell becomes the new end of the linked list.

\begin{figure}[ht!]
    \centering
    \begin{lstlisting}[style = basic, language = roopl]                
    class LinkedList
        Cell head
        int listLength
    
        method insertHead(Cell cell)
            if head = nil & cell != nil then
                head <=> cell               // Set cell as head of list if list is empty
            else skip
            fi head != nil & cell = nil
    
        method appendCell(Cell cell)
            call insertHead(cell)           // Insert as head if empty list
    
            if head != nil then
                call head::append(cell)     // Iterate until we hit end of list
            else skip
            fi head != nil
    
            listLength += 1                 // Increment length
    
        method prependCell(Cell cell)
            call insertHead(cell)           // Insert as head if empty list
    
            if cell != nil & head != nil then
                call cell::append(head)     // Set cell.next = head. head = nil after execution
            else skip
            fi cell != nil & head = nil
    
            if cell != nil & head = nil then
                cell <=> head               // Set head = cell. Cell is nil after execution
            else skip
            fi cell = nil & head != nil
    
            listLength += 1                 // Increment length
    
        method length(int result)
            result ^= listLength  
    \end{lstlisting}
    \caption{Linked List class}
    \label{fig:linked-list-class}
\end{figure}

An interesting observation is that the \textit{append} method is called an additional time \textit{after} setting the cell as the new end of the linked list. In a non-reversible programming language, we would simply call append in the else-branch of the first conditional. In the reversible setting, this is not an option, as the append call would modify the value of the \textit{next} and \textit{cell} variables and as such, corrupt the control flow as the exit condition would be true after executing both the then- and else-branch of the conditional. To avoid this, we simply call one additional time with a \textbf{nil} value \textit{cell}.
This "wasted" additional call with a \textbf{nil} value is a recurring technique in the following presented reversible data structure implementations. 

\subsection{Binary Tree}
\label{subsec:binary-tree}
Figures~\ref{fig:binary-tree-class},~\ref{fig:binary-tree-node-class} and~\ref{fig:binary-tree-node-class-cont} shows the implementation of a binary tree in form of a  rooted, unbalanced, min-heap. The \textit{Tree} class shown in figure~\ref{fig:binary-tree-class} has a single root node field and the three methods \textit{insertNode, sum} and \textit{mirror}. For insertion, the \textit{insertNode} method is called from the \textit{root}, if it is initialized and if not, the passed node parameter is simply set as the new root of the tree. The \textit{insertNode} method implemented in the \textit{Node} class shown in figure~\ref{fig:binary-tree-node-class} first determines if we need to insert left or right but checking the passed value against the value of the current node. This is done recursively, until an uninitialized node in the correct subtree has been found. Note that as a consequence of reversibility, the value of node we wish to insert must be passed separately in the method call as we otherwise cannot zero-clear it after swapping the node we are inserting with either the right or left child of the current cell.

\begin{figure}[ht!]
    \centering
    \begin{lstlisting}[style = basic, language = roopl] 
    class Tree
        Node root
        
        method insertNode(Node node, int value)
            if root = nil & node != nil then
                root <=> node
            else skip
            fi root != nil & node = nil
    
            if root != nil then
                call root::insertNode(node, value)
            else skip
            fi root != nil
    
        method sum(int result)
            if root != nil then
                call root::getSum(result)
            else skip
            fi root != nil
    
        method mirror()
            if root != nil then
                call root::mirror()
            else skip
            fi root != nil
    \end{lstlisting}
    \caption{Binary Tree class}
    \label{fig:binary-tree-class}
\end{figure}

Summing and mirroring the tree works in a similar fashion by recursively iterating each node of the tree. For summing we simply add the value of the node to the sum and for mirroring we swap the children of the node and then recursively swap the children of the left and right node, if initialized. The sum and mirror methods are implemented in figure~\ref{fig:binary-tree-node-class-cont}.

\begin{figure}[ht!]
    \centering
    \begin{lstlisting}[style = basic, language = roopl] 
    method getSum(int result)
        result += value                  // Add the value of this node to the sum   

        if left != nil then
            call left::getSum(result)   // If we have a left child, follow that path
        else skip                       // Else, skip
        fi left != nil

        if right != nil then
            call right::getSum(result)  // If we have a right child, follow that path
        else skip                       // Else, skip
         fi right != nil

    method mirror()
        left <=> right                  // Swap left and right children

        if left = nil then skip
        else call left::mirror()        // Recursively swap children if left != nil
        fi left = nil

        if right = nil then skip
        else call right::mirror()       // Recursively swap children if right != nil
        fi right = nil 
    \end{lstlisting}
    \caption{Binary Tree node class (cont)}
    \label{fig:binary-tree-node-class-cont}
\end{figure}

\begin{figure}[ht!]
    \centering
    \begin{lstlisting}[style = basic, language = roopl] 
    class Node
        Node left
        Node right
        int value
    
        method setValue(int newValue)
            value ^= newValue 
    
        method insertNode(Node node, int nodeValue)
            // Determine if we insert left or right
            if nodeValue < value then
                if left = nil & node != nil then
                    // If open left node, store here
                    left <=> node
                else skip
                fi left != nil & node = nil
    
                if left != nil then
                    // If current node has left, continue iterating
                    call left::insertNode(node, nodeValue)
                else skip
                fi left != nil
            else
                if right = nil & node != nil then
                    // If open right node spot, store here
                    right <=> node
                else skip
                fi right != nil & node = nil
    
                if right != nil then
                    // If current node has, continue searching
                    call right::insertNode(node, nodeValue)
                else skip
                fi right != nil
            fi nodeValue < value
    \end{lstlisting}
    \caption{Binary Tree node class}
    \label{fig:binary-tree-node-class}
\end{figure}

\subsection{Doubly Linked List}
\label{subsec:doubly-linked-list}
Finally, we present the reversible doubly linked list, shown in figures~\ref{fig:doubly-linked-list-class}-\ref{fig:doubly-linked-list-cell-class-cont}. A \textit{cell} in a doubly linked list contains a reference to itself named \textit{self}, a reference to its left and right neighbours, a data and an index field. As with the linked list and binary tree implementation the \textit{DoubleLinkedList} class has a field referencing the head of the list and its \textit{appendCell} method is identical to the one of the linked list. 

\begin{figure}[ht]
    \centering
    \begin{tikzpicture}
        % Cells
        \draw (0,0) rectangle (2,2) node[midway] {$c_1$};
        \draw (3,0) rectangle (5,2) node[midway] {$c_2$};
        \draw (6,0) rectangle (8,2) node[midway] {$c_3$};

        % Arrows
        \node[circle,fill,inner sep=1pt] at (1.25, 1.75) {};
        \draw[-latex] (1.25, 1.75) to[out=45, in=135, distance=.85cm] (.75, 1.70);

        \node[circle,fill,inner sep=1pt] at (3.25, .5) {};
        \draw[-latex] (3.25, .5) to[out=-135, in=-45] (1.75, .5);

        \node[circle,fill,inner sep=1pt] at (4.25, 1.75) {};
        \draw[-latex, dotted] (4.25, 1.75) to[out=45, in=135, distance=.85cm] (3.75, 1.70);

        \node[circle,fill,inner sep=1pt] at (6.25, .5) {};
        \draw[-latex, dotted] (6.25, .5) to[out=-135, in=-45] (4.75, .5);

        \node[circle,fill,inner sep=1pt] at (1.75, 1.5) {};
        \draw[-latex, dotted] (1.75, 1.5) to[out=45, in=135] (3.25, 1.5);

        \node[circle,fill,inner sep=1pt] at (4.75, 1.5) {};
        \draw[-latex, dashed] (4.75, 1.5) to[out=45, in=135] (6.25, 1.5);

        \node[circle,fill,inner sep=1pt] at (7.25, 1.75) {};
        \draw[-latex, dashed] (7.25, 1.75) to[out=45, in=135, distance=.85cm] (6.75, 1.70);
    \end{tikzpicture}
    \caption{Multiple identical reference are needed for a doubly linked list implementation}
    \label{fig:doubly-linked-list-reference}
\end{figure}
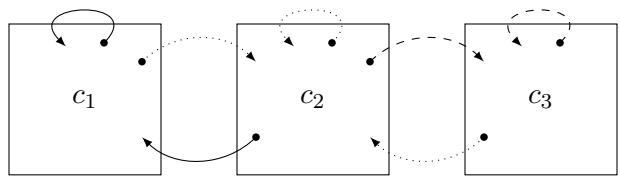

This data structure is particularly interesting, as it, unlike the former two presented structures, cannot be expressed in \textsc{Roopl}, as this requires multiple reference to objects, in order for an object to point to itself and to its left and right neighbours. Figure~\ref{fig:doubly-linked-list-reference} shows the multiple references needed for the doubly linked list implementation denoted by the three different arrow types.

\begin{figure}[ht!]
    \centering
    \begin{lstlisting}[style = basic, language = roopl] 
    class DoublyLinkedList
        Cell head
        int length
    
        method appendCell(Cell cell)
            if head = nil & cell != nil then
                head <=> cell
            else skip
            fi head != nil & cell = nil
    
            if head != nil then 
                call head::append(cell)
            else skip
            fi head != nil
    
            length += 1
    \end{lstlisting}
    \caption{Doubly Linked List class}
    \label{fig:doubly-linked-list-class}
\end{figure}

When we append a cell to the list, we first search recursively through the list until we are at the end. The new cell is then set as \textit{right} of the current cell. A reference to the current self is created using the \textbf{copy} statement, and set as \textit{left} of the new end of the list, thus resulting in the new cell being linked to list and now acting as end of the list.

\begin{figure}[ht!]
    \centering
    \begin{lstlisting}[style = basic, language = roopl] 
        class Cell
        int data
        int index
        Cell left
        Cell right
        Cell self
    
        method setData(int value)
            data ^= value
    
        method setIndex(int i)
            index ^= i    
    
        method setLeft(Cell cell)
            left <=> cell
    
        method setRight(Cell cell)
            right <=> cell
    
        method setSelf(Cell cell)
            self <=> cell
    \end{lstlisting}
    \caption{Doubly Linked List Cell class}
    \label{fig:doubly-linked-list-cell-class}
\end{figure}

\begin{figure}[ht!]
    \centering
    \begin{lstlisting}[style = basic, language = roopl]
    method append(Cell cell)
        if right = nil & cell != nil then   // If current cell does not have a right neighbour
            right <=> cell                  // Set new cell as right neighbour of current cell
        
            local Cell selfCopy = nil      
            copy Cell self selfCopy         // Copy reference to current cell
            call right::setLeft(selfCopy)   // Set current as left of  right neighbour
            delocal Cell selfCopy = nil

            local int cellIndex = index + 1
            call right::setIndex(cellIndex) // Set index in right neighbour of current
            delocal int cellIndex = index + 1
        else skip
        fi right != nil & cell = nil

        if right != nil then
            call right::append(cell)        // Keep searching for empty right neighbour
        else skip
        fi right != nil 
    \end{lstlisting}
    \caption{Doubly Linked List Cell class (cont)}
    \label{fig:doubly-linked-list-cell-class-cont}
\end{figure}

The data structure could relatively easily be extended to work as a dynamic array. Currently each cell contains an \textit{index} field, specifying their position in the list. If, say, we wanted to insert some new data at index $n$, without updating the existing value, but essentially squeezing in a new cell, we could add a method to the \textit{DoublyLinkedList} class taking a data value and an index. When executing this method, we could iterate the list until we reach the cell with index $n$, construct a new \textit{cell} instance, update required \textit{left} and \textit{right} pointers to insert the new cell at the correct position, in such a way that the old cell at index $n$ now is the new right neighbour of the cell and finally recursively iterating the list, incrementing the index of cells to the right of the new cell by one. In reverse, this would remove a cell from the list. If we want to update an existing value at a index, a similar technique could be used, where we iterate through the cells until we find the correct index. If we are given an index that is out of bounds in terms of the current length of the list, we could extend the tail on the list until reach a cell with the wanted index. When we are zero-clearing a value that is the furthest index, the inverse would apply, and a such we would zero-clear the cell, and the deallocate cells until we reach a cell which does not have a zero-cleared \textit{data} field. 

This extended doubly linked list would also allow lists of n-dimensional lists, as the type of the \textit{data} field simply could be changed to, say, a \textit{FooDoublyLinkedList}, resulting in an array of Foo arrays. 
\newpage

\section{Type System}
\label{sec:type-system}
The type system of \rooplpp expands on the type system of \textsc{Roopl} presented by~\citeauthor{th:roopl}~\cite{th:roopl} and is analogously described by syntax-directed inference typing rules in the style of ~\citeauthor{wi:semantics}~\cite{wi:semantics}. As \rooplpp introduces two new types in form of \textit{references} and arrays, a few \textsc{Roopl} typing rules must be modified to accommodate these added types. For completeness all typing rules, including unmodified rules, are included in the following sections. 

\subsection{Preliminaries}
\label{subsec:preliminaries}
The types in \rooplpp are given by the following grammar:
\begin{equation*}
    \tau ::= \textbf{int}\ |\ c \in \text{ClassIDs}\ |\ r \in \text{ReferenceIDs}\ |\ i \in \text{IntegerArrayIDs}\ |\ o \in \text{ClassArrayIDs}
\end{equation*}
The type environment $\Pi$ is a finite map pairing variables to types, which can be applied to an identifier $x$ using the $\Pi(x)$ notation. Notation $\Pi' = \Pi[x \mapsto \tau]$ defines updates and creation of a new type environment $\Pi'$ such that $\Pi'(x) = \tau$ and $\Pi'(y) = \Pi(y)$ if $x \not= y$, for some variable identifier $x$ and $y$. The empty type environment is denoted as $[\ ]$ and the function $vars\ :\ Expressions\ \to \text{VarIDs}$ is described by the following definition
\begin{align*}
    &\text{vars}(\bar{n}) &&= \emptyset\\
    &\text{vars}(\textbf{nil}) &&= \emptyset\\
    &\text{vars}(x) &&= \{\ x\ \}\\
    &\text{vars}(x[e]) &&= \{\ x\ \} \cup \text{vars}(e)\\
    &\text{vars}(e_1 \otimes e_2) &&= \text{vars}(e_1) \cup \text{vars}(e_2).
\end{align*}
The binary subtype relation $c_1 \prec: c_2$ is required for supporting subtype polymorphism and is defined as follows:
\begin{align*}
    c_1 &\prec: c_2 &&\text{if}\ c_1\ \text{inherits from}\ c_2\\
    c   &\prec: c   &&(reflexivity)\\
    c_1 &\prec: c_3 &&\text{if}\ c_1 \prec: c_2\ \text{and}\ c_2 \prec: c_3\ (transitivity)
\end{align*}

Furthermore, we formally define object models in such a way that inherited fields and methods are included, unless overridden by the derived fields. Therefore, we define $\Gamma$ to be the class map of a program $p$, such that $\Gamma$ is a finite map from class identifiers to tuples of methods and fields for the class $p$. Application of a class map $\Gamma$ to some class $cl$ is denoted as $\Gamma(cl)$. Construction of a class map is done through function $gen$, as shown in figure~\ref{fig:class-gen}. Figure~\ref{fig:fields-methods} defines the $fields$ and $methods$ functions to determine these given a class. Set operation $\uplus$ defines method overloading by dropping base class methods if a similarly named method exists in the derived class. The definitions shown in Figure~\ref{fig:class-gen} and~\ref{fig:fields-methods} are originally from~\cite{th:roopl}.
 
\begin{figure}[ht]
    \centering
    $\text{gen}\Big(\overbrace{cl_1,\ ...,\ cl_n}^p\Big) = \overbrace{\Big[ \alpha(cl_1) \mapsto \beta(cl_1),\ ...,\ \alpha(cl_n) \mapsto \beta(cl_n) \Big]}^\Gamma$\\

    \vskip 2em
    
    $\alpha\Big(\textbf{class}\ c\ \cdots\Big) = c$
    \hskip 2em
    $\beta(cl) = \Big(\text{fields}(cl),\ \text{methods}(cl) \Big)$

    \caption{Definition $gen$ for constructing the finite class map $\Gamma$ of a given program $p$, originally from ~\cite{th:roopl}}
    \label{fig:class-gen}
\end{figure}

\begin{figure}[H]
    \centering
    \begin{equation*}
        \text{fields}(cl) = \begin{cases}
            \eta(cl)  & \text{if}\ cl \sim [ \textbf{class}\ c\ \cdots ]\\
            \eta(cl) \cup \text{fields}\Big(\alpha^{-1}(c')\Big) & \text{if}\ cl \sim [ \textbf{class}\ c\ \textbf{inherits}\ c'\ \cdots ]
        \end{cases}
    \end{equation*}

    \vskip 1em

    \begin{equation*}
        \text{methods}(cl) = \begin{cases}
            \delta(cl) & \text{if}\ cl \sim [ \textbf{class}\ c\ \cdots ]\\
            \delta(cl) \uplus \text{methods}\Big(\alpha^{-1}(c')\Big) & \text{if}\ cl \sim [\textbf{class}\ c\ \textbf{inherits}\ c'\ \cdots]
        \end{cases}
    \end{equation*}

    \vskip 1em

    \begin{equation*}
        A\ \uplus B\ \overset{def}{=}\ A \cup \bigg\{ m \in B\ \bigg\vert\ \nexists\  m' \Big( \zeta(m') = \zeta(m) \wedge m' \in A \Big) \bigg\}
    \end{equation*}

    $\zeta\Big( \textbf{method}\ q\ (\cdots)\ s \Big) = q$
    \hskip 2em
    $\eta\Big( \textbf{class}\ c\ \cdots\ \overbrace{t_1 f_1\ \cdots\ t_n f_n}^{fs}\ \cdots \Big) = fs$

    \begin{equation*}
        \delta\Big( \textbf{class}\ c\ \cdots \overbrace{\textbf{method}\ q_1\ (\cdots)\ s_1\ \cdots\ \textbf{method}\ q_n\ (\cdots)\ s_n }^{ms}\ \cdots \Big) = ms
    \end{equation*}
    \caption{Definition of fields and methods, originally from~\cite{th:roopl}}
    \label{fig:fields-methods}
\end{figure}

Finally, we formally define a link between arrays of a given type and other types. The function \textit{arrayType}, defined in figure~\ref{fig:array-type}, is $c$ if the passed array $a$ is an array of class $c$ instances.

\begin{figure}[ht]
    \centering
    \begin{equation*}
        \text{arrayType}(a) = \begin{cases}
            c  & \text{if}\ a \in ClassArrayIDs\ \text{and}\ a\ \text{is a}\ c\ \text{array}\\
            \textbf{int} & \text{if}\ a \in i
        \end{cases}
    \end{equation*}
    \caption{Definition \textit{arrayType} for mapping types of arrays to either class types or the integer type}
    \label{fig:array-type}
\end{figure}

\newpage
\subsection{Expressions}
\label{subsec:expressions}
The type judgment

\begin{prooftree}
	\AXC{}
    \UIC{$\Pi \vdash_{expr}\ e : \tau$} 
\end{prooftree} 

defines the type of expressions. The judgment reads as: under type environment $\Pi$, expression $e$ has type $\tau$.

\begin{figure}[ht]
    \begin{center}
        \AXC{} % Con
        \RL{\textsc{T-Con}}
        \UIC{$\Pi \vdash_{expr}\ n : \textbf{int}$} 
        \DP
        \hskip 2em % Var
        \AXC{$\Pi(x) = \tau$}
        \RL{\textsc{T-Var}}
        \UIC{$\Pi \vdash_{expr}\ x : \tau$} 
        \DP
        \hskip 2em % Nil
        \AXC{$\tau \not= \textbf{int}$}
        \RL{\textsc{T-Nil}}
        \UIC{$\Pi \vdash_{expr}\ \textbf{nil}\ : \tau$}
        \DP
        
        \vskip 2em % BinOpInt
        \AXC{$\Pi \vdash_{expr}\ e_1 : \textbf{int}$}
        \AXC{$\Pi \vdash_{expr}\ e_2 : \textbf{int}$}
        \RL{\textsc{T-BinOpInt}}
        \BIC{$\Pi \vdash_{expr}\ e_1 \otimes e_2 : \textbf{int}$}
        \DP

        \vskip 2em % BinOpObj
        \AXC{$\Pi \vdash_{expr}\ e_1 : c$}
        \AXC{$\Pi \vdash_{expr}\ e_2 : c$}
        \AXC{$\otimes \in \{ \texttt{\textbf{=}}, \texttt{\textbf{!=}} \}$}
        \RL{\textsc{T-BinOpObj}}
        \TIC{$\Pi \vdash_{expr}\ e_1 \otimes e_2 : \textbf{int}$}
        \DP
    \end{center}
    \caption{Typing rules for expressions in \textsc{Roopl}, originally from~\cite{th:roopl}}
    \label{fig:typing-rules-expressions}
\end{figure}

The original expression typing rules from \textsc{Roopl} are shown in figure~\ref{fig:typing-rules-expressions}. The type rules \textsc{T-Con}, \textsc{T-Var} and \textsc{T-Nil} defines typing of the simplest expressions. Numeric literals are of type \textbf{int}, typing of variable expressions depends on the type of the variable in the type environment and the \textbf{nil} literal is a non-integer type. All binary operations are defined for integers, while only equality-operators are defined for objects.

With the addition of the \rooplpp array type, we extend the expression typing rules with rule \textsc{T-ArrElemVar} which defines typing for array element variables, shown in figure~\ref{fig:typing-rules-expression-extension}.

\begin{figure}[ht]
    \begin{center}
        \AXC{$\Pi(x) = \tau[\ ]$}
        \AXC{$\Pi_{expr} \vdash\ e\ :\ \textbf{int}$}
        \RL{\textsc{T-ArrElemVar}}
        \BIC{$\Pi \vdash_{expr}\ x[e] : \tau$} 
        \DP
    \end{center}
    \caption{Typing rule extension for the \textsc{Roopl} typing rules}
    \label{fig:typing-rules-expression-extension}
\end{figure}

\subsection{Statements}
\label{subsec:statements}
The type judgment

\begin{prooftree}
	\AXC{}
    \UIC{$\langle \Pi,\ c\rangle\ \vdash_{stmt}^\Gamma\ s$} 
\end{prooftree} 

defines well-typed statements. The judgment reads as under type environment $\Pi$ within class $c$, statement $s$ is well-typed with class map $\Gamma$.
\begin{figure}[H]
    \begin{center}
        % Assignment
        \AXC{$x \not\in \text{vars}(e)$}
        \AXC{$\Pi \vdash_{expr}\ e : \textbf{int}$}
        \AXC{$\Pi(x) = \textbf{int}$}
        \RL{\textsc{T-AssVar}}
        \TIC{$\langle \Pi, c \rangle \vdash^{\Gamma}_{stmt}\ x\ \texttt{$\odot$=}\ e$}
        \DP

        \vskip 2em
        
        % If
        \AXC{$\Pi \vdash_{expr}\ e_1 : \textbf{int}$}
        \AXC{$\langle \Pi, c \rangle \vdash^{\Gamma}_{stmt}\ s_1$}
        \AXC{$\langle \Pi, c \rangle \vdash^{\Gamma}_{stmt}\ s_2$}
        \AXC{$\Pi \vdash_{expr}\ e_2 : \textbf{int}$}
        \RL{\textsc{T-If}}
        \QIC{$\langle \Pi, c \rangle \vdash^{\Gamma}_{stmt}\ \textbf{if}\ e_1\ \textbf{then}\ s_1\ \textbf{else}\ s_2\ \textbf{fi}\ e_2$}
        \DP

        \vskip 2em

        % Loop
        \AXC{$\Pi \vdash_{expr}\ e_1 : \textbf{int}$}
        \AXC{$\langle \Pi, c \rangle \vdash^{\Gamma}_{stmt}\ s_1$}
        \AXC{$\langle \Pi, c \rangle \vdash^{\Gamma}_{stmt}\ s_2$}
        \AXC{$\Pi \vdash_{expr}\ e_2 : \textbf{int}$}
        \RL{\textsc{T-Loop}}
        \QIC{$\langle \Pi, c \rangle \vdash^{\Gamma}_{stmt}\ \textbf{from}\ e_1\ \textbf{do}\ s_1\ \textbf{loop}\ s_2\ \textbf{until}\ e_2$}
        \DP

        \vskip 2em
        
        % ObjectBlock + Skip
        \AXC{$\langle \Pi[x \mapsto c'], c \rangle \vdash^{\Gamma}_{stmt}\ s$}
        \RL{\textsc{T-ObjBlock}}
        \UIC{$\langle \Pi, c \rangle \vdash^{\Gamma}_{stmt}\ \textbf{construct}\ c'\ x\quad s\quad \textbf{destruct}\ x$} 
        \DP
        \hskip 2em
        \AXC{}
        \RL{\textsc{T-Skip}}
        \UIC{$\langle \Pi, c \rangle \vdash^{\Gamma}_{stmt}\ \textbf{skip}$} 
        \DP

        \vskip 2em

        % Seq + SwapVar
        \AXC{$\langle \Pi, c \rangle \vdash^{\Gamma}_{stmt}\ s_1$}
        \AXC{$\langle \Pi, c \rangle \vdash^{\Gamma}_{stmt}\ s_2$}
        \RL{\textsc{T-Seq}}
        \BIC{$\langle \Pi, c \rangle \vdash^{\Gamma}_{stmt}\ s_1\ s_2$}
        \DP
        \hskip 2em
        \AXC{$\Pi(x_1) = \Pi(x_2)$}
        \RL{\textsc{T-SwpVar}}
        \UIC{$\langle \Pi, c \rangle \vdash^{\Gamma}_{stmt}\ x_1\ \texttt{\textbf{<=>}}\ x_2$}
        \DP

        \vskip 2em

        % T-Call
        \AXC{$\Gamma(\Pi(c)) = \Big( fields,\ methods \Big)$}
        \DP
        \AXC{$\Big( \textbf{method}\ q(t_1\ y_1,\ ...,\ t_n\ y_n)\ s \Big) \in methods$}
        \DP
        \AXC{$\{ x_1,\ ...,\ x_n \} \cap fields = \emptyset$}
        \AXC{$i \not= j \implies x_i \not= x_j$}
        \AXC{$\Pi(x_1) \prec: t_1\ \cdots\ \Pi(x_n) \prec: t_n$}
        \RL{\textsc{T-Call}}
        \TIC{$\langle \Pi, c \rangle \vdash^{\Gamma}_{stmt}\ \textbf{call}\ q(x_1,\ ...,\ x_n)$}
        \DP

        \vskip 2em
        
        % T-CallO
        \AXC{$\Gamma(\Pi(x_0)) = \Big( fields,\ methods \Big)$}
        \DP
        \AXC{$\Big( \textbf{method}\ q(t_1\ y_1,\ ...,\ t_n\ y_n)\ s \Big) \in methods$}
        \DP
        \AXC{$i \not= j \implies x_i \not= x_j$}
        \AXC{$\Pi(x_1) \prec: t_1\ \cdots\ \Pi(x_n) \prec: t_n$}
        \RL{\textsc{T-CallO}}
        \BIC{$\langle \Pi, c \rangle \vdash^{\Gamma}_{stmt}\ \textbf{call}\ x_0::q(x_1,\ ...,\ x_n)$}
        \DP

        \vskip 2em

        % UC and UCO
        \AXC{$\langle \Pi, c \rangle \vdash^{\Gamma}_{stmt}\ \textbf{call}\ q(x_1,\ ...,\ x_n)$}
        \RL{\textsc{T-Uc}}
        \UIC{$\langle \Pi, c \rangle \vdash^{\Gamma}_{stmt}\ \textbf{uncall}\ q(x_1,\ ...,\ x_n)$}
        \DP
        \hskip 2em
        \AXC{$\langle \Pi, c \rangle \vdash^{\Gamma}_{stmt}\ \textbf{call}\ x_0::q(x_1,\ ...,\ x_n)$}
        \RL{\textsc{T-Uco}}
        \UIC{$\langle \Pi, c \rangle \vdash^{\Gamma}_{stmt}\ \textbf{uncall}\ x_0::q(x_1,\ ...,\ x_n)$}
        \DP
    \end{center}
    \caption{Typing rules for statements in \textsc{Roopl}, originally from~\cite{th:roopl}}
    \label{fig:typing-rules-statements}
\end{figure}

Typing rule \textsc{T-AssVar} defines variable assignments for an integer variable and an integer expression result, given that the variable $x$ does not occur in the expression $e$. 

The type rules \textsc{T-If} and \textsc{T-Loop} defines reversible conditionals and loops as known from \textsc{Janus}, where entry and exit conditions are integers and branch and loop statements are well-typed statements. 

The object block, introduced in \textsc{Roopl}, is only well-typed if its body statement is well-typed. 

The \textbf{skip} statement is always well-typed, while a sequence of statements are well-typed if each of the provided statements are. Variable \textbf{swap} statements are well-typed if both operands are of the same type under type environment $\Pi$.

As with \textsc{Roopl}, type correctness of local method invocation is defined in rule \textsc{T-Call} iff:
\begin{itemize}
    \item The number of arguments matches the method arity
    \item No class fields are present in the arguments passed to the method (To prevent irreversible updates)
    \item The argument list contains unique elements
    \item Each argument is a subtype of the type of the equivalent formal parameter. 
\end{itemize}

For foreign method invocations, typing rule \textsc{T-CallO}. A foreign method invocation is well-typed using the same rules as for \textsc{T-Call} besides having no restrictions on class fields parameters in the arguments, but an added rule stating that the callee object $x_0$ must not be passed as an argument.

The typing rules \textsc{T-UC} and \textsc{T-UCO} defines uncalling of methods in terms of their respective inverse counterparts.

\begin{figure}[H]
    \begin{center}
        % Assignment
        \AXC{$\Pi(x) = \textbf{int}[\ ]$}
        \DP
        \AXC{$\Pi \vdash_{expr} e_1\ :\ \textbf{int}$}
        \DP
        \vskip 1em
        \AXC{$\Big( {\ x\ } \cup \text{vars}(e_1) \Big) \cap \text{vars}(e_2) = \emptyset$}
        \AXC{$\Pi \vdash_{expr}\ e_2 : \textbf{int}$}
        \RL{\textsc{T-ArrElemAss}}
        \BIC{$\langle \Pi, c \rangle \vdash^{\Gamma}_{stmt}\ x[e_1]\ \texttt{$\odot$=}\ e_2$}
        \DP 

        \vskip 2em
         % Object new + delete
        \AXC{$\Pi(x) = c'$}
        \RL{\textsc{T-ObjNew}}
        \UIC{$\langle \Pi, c \rangle  \vdash^{\Gamma}_{stmt}\ \textbf{new}\ c'\ x$} 
        \DP
        \hskip 2em
        \AXC{$\Pi(x) = c'$}
        \RL{\textsc{T-ObjDlt}}
        \UIC{$\langle \Pi, c \rangle  \vdash^{\Gamma}_{stmt}\ \textbf{delete}\ c'\ x$} 
        \DP

        \vskip 2em

         % Array new
         \AXC{$\text{arrayType}(a) \in \Big\{ \text{classIDs}, \textbf{int} \Big\}$}
         \AXC{$\Pi \vdash_{expr}\ e = \textbf{int}$}
         \AXC{$\Pi(x) = a[\ ]$}
         \RL{\textsc{T-ArrNew}}
         \TIC{$\langle \Pi, c \rangle  \vdash^{\Gamma}_{stmt}\ \textbf{new}\ a[e]\ x$} 
         \DP

         \vskip 2em

         % Array delete
         \AXC{$\text{arrayType}(a) \in \Big\{ \text{classIDs}, \textbf{int} \Big\}$}
         \AXC{$\Pi \vdash_{expr}\ e = \textbf{int}$}
         \AXC{$\Pi(x) = a[\ ]$}
         \RL{\textsc{T-ArrDlt}}
         \TIC{$\langle \Pi, c \rangle  \vdash^{\Gamma}_{stmt}\ \textbf{delete}\ a[e]\ x$}
         \DP
 
         \vskip 2em

        % Copy + uncopy
        \AXC{$\Pi(x) = c'$}
        \AXC{$\Pi(x') = c'$}
        \RL{\textsc{T-Cp}}
        \BIC{$\langle \Pi, c \rangle  \vdash^{\Gamma}_{stmt}\ \textbf{copy}\ c'\ x\ x'$} 
        \DP
        \hskip 2em
        \AXC{$\Pi(x) = c'$}
        \AXC{$\Pi(x') = c'$}
        \RL{\textsc{T-Ucp}}
        \BIC{$\langle \Pi, c \rangle  \vdash^{\Gamma}_{stmt}\ \textbf{uncopy}\ c'\ x\ x'$}
        \DP

        \vskip 2em

        % Local block
        \AXC{$\langle \Pi, c \rangle \vdash_{expr}\ e_1\ :\ c'$}
        \AXC{$\langle \Pi[x \mapsto c'], c \rangle \vdash^{\Gamma}_{stmt}\ s$}
        \AXC{$\langle \Pi, c \rangle \vdash_{expr}\ e_2\ :\ c'$}
        \RL{\textsc{T-LocalBlock}}
        \TIC{$\langle \Pi, c \rangle \vdash^{\Gamma}_{stmt}\ \textbf{local}\ c'\ x = e_1\quad s\quad \textbf{delocal}\ c'\ x = e_2$} 
        \DP
    \end{center}
    \caption{Typing rules extensions for statements in \rooplpp}
    \label{fig:typing-rules-statements-extensions}
\end{figure}

Figure~\ref{fig:typing-rules-statements-extensions} shows the typing rules for the extensions made to \textsc{Roopl} in \rooplpp, covering the \textbf{new}/\textbf{delete} and \textbf{copy}/\textbf{uncopy} statements for objects and arrays and local blocks.

The typing rule \textsc{T-ArrElemAss} defines assignment to integer array element variables, and is well-typed when the type of array $x$ is \textbf{int}, the variable $x[e_1]$ is not present in the right-hand side of the statement, no variables in $e_1$ exist in $e_2$ and both expressions $e_1$ and $e_2$ evaluates to integers.

The \textsc{T-ObjNew} and \textsc{T-ObjDlt} rules define well-typed \textbf{new} and \textbf{delete} statements for dynamically lifetimed objects. The \textbf{new} statement is well-typed, as long as $c' \in \text{classIDs}$ and the variable $x$ is of type $c'$ under type environment $\Pi$ and \textbf{delete} is also well-typed if the type of $x$ under type environment $\Pi$ is equal to $c'$.

The \textsc{T-ArrNew} and \textsc{T-ArrDlt} rules define well-type \textbf{new} and \textbf{delete} statement for \rooplpp arrays. The \textbf{new} statement is well-typed, if the type of the array either is a classID or \textbf{int}, the length expression evaluates to an integer and $x$ is of of type $a[\ ]$ under the type environment $\Pi$, and \textbf{delete} is well-typed if the type of the array is either a classID or \textbf{int}, the length expression evaluates to an integer and $x$ is equal to the array type $a$.

Typing rules \textsc{T-Cp} and \textsc{T-Ucp} define well-typed reference copy and un-copying statements. A well-typed \textbf{copy} or \textbf{uncopy} statement requires that the types of $x$ and $x'$ both are $c'$ under type environment $\Pi$

The rule \textsc{T-LocalBlock} defines well-typed local blocks. A local block is well-typed if its two expression $e_1$ and $e_2$ are well-typed and its body statement $s$ is well-typed. 

\subsection{Programs}
\label{subsec:programs} 
As with \textsc{Roopl}, a \rooplpp program is well-typed if all of its classes and their respective methods are well-typed and if there exists a nullary main method. Figure~\ref{fig:typing-programs} shows the typing rules for class methods, classes and programs.

\begin{figure}[H]
    \begin{center}
        % Method
        \AXC{$\langle \Pi[x_1 \mapsto t_1,\ ...,\ x_n \mapsto t_n],\ c \rangle \vdash^{\Gamma}_{stmt}\ s$}
        \RL{\textsc{T-Method}}
        \UIC{$\langle \Pi, c \rangle \vdash^{\Gamma}_{meth}\ \textbf{method}\ q(t_1 x_1,\ ...,\ t_n x_n)\ s$}
        \DP

        \vskip 2em

        % Class
        \AXC{$\Gamma(c) = \Big( \overbrace{\big\{ \langle t_1, f_1 \rangle,\ ...,\ \langle t_i, f_i \rangle \big\}}^{fields},\ \overbrace{\{ m_1,\ ...,\ m_n \}}^{methods} \Big)$}
        \DP
        \AXC{$\Pi = \big[ f_1 \mapsto t_1,\ ...,\ f_n \mapsto t_n \big]$}
        \AXC{$\langle \Pi, c \rangle \vdash^{\Gamma}_{meth}\ m_1 \quad \cdots \quad \langle \Pi, c \rangle \vdash^{\Gamma}_{meth}\ m_n$}
        \RL{\textsc{T-Class}}
        \BIC{$\vdash^{\Gamma}_{class}\ c$}
        \DP

        \vskip 2em

        \AXC{$\Big( \textbf{method main ()}\ s \Big) \in \displaystyle\bigcup^{n}_{i = 1} \text{methods}(c_i)$}
        \DP
        \AXC{$\Gamma = \text{gen}(c_1,\ ...,\ c_n)$}
        \AXC{$\vdash^{\Gamma}_{class}\ c_1 \quad \cdots \quad \vdash^{\Gamma}_{class}\ c_n$}
        \RL{\textsc{T-Prog}}
        \BIC{$\vdash_{prog}\ c_1\ \cdots\ c_n$}
        \DP
    \end{center}
    \caption{Typing rules for class methods, classes and programs, originally from~\cite{th:roopl}}
    \label{fig:typing-programs}
\end{figure}

\newpage
\section{Language Semantics}
\label{sec:language-semantics}
The following sections contain the operational semantics of \rooplpp, as specified by syntax-directed inference rules.

\subsection{Preliminaries}
\label{subsec:semantics-preliminaries}
We define $l$ to be a location. We define a location for integer variables to bind to a single location in program memory and a vector of memory locations for object and array variables, where the vector is the size of the object or array. A memory location is in the set of non-negative integers, $\mathbb{N}_0$. An environment $\gamma$ is a partial function mapping variables to memory locations. A store $\mu$ is a partial function mapping memory locations to values. An object is a tuple of a class name and an environment mapping fields to memory locations. A value is either an integer, an object, a location or a vector of locations.

Applications of environments $\gamma$ and stores $\mu$ are analogous to the type environment $\Gamma$, defined in section~\ref{subsec:preliminaries}. 

\begin{figure}[ht]
    \begin{alignat*}{2}
        l      \in\ &\text{Locs}     &&= \mathbb{N}_0\\
        \gamma \in\ &\text{Envs}     &&= \text{VarIDs} \rightharpoonup \text{Locs}\\
        \mu    \in\ &\text{Stores}   &&= \text{Locs} \rightharpoonup \text{Values}\\
        &\text{Objects} &&= \Big\{ \langle c_f,\ \gamma_f \rangle \mid c_f \in \text{ClassIDs}\ \wedge\ \gamma_f \in \text{Envs} \Big\}\\
        v \in\ &\text{Values} &&= \mathbb{Z} \cup \text{Objects}\ \cup \text{Locs} \cup [\text{Locs}]
    \end{alignat*}
    \caption{Semantic values, originally from~\cite{th:roopl}}
    \label{fig:semantic-values}
\end{figure}

\subsection{Expressions}
\label{subsec:semantics-expressions}
The judgment:
\begin{prooftree}
    \AXC{$\langle \gamma, \mu \rangle \vdash_{expr}\ e \Rightarrow v$}
\end{prooftree}
defines the meaning of expressions. We say that under environment $\gamma$ and store $\mu$, expression $e$ evaluates to value $v$.

\begin{figure}[ht]
    \begin{center}
        \AXC{}
        \RL{\textsc{Con}}
        \UIC{$\langle \gamma, \mu \rangle \vdash_{expr}\ n \Rightarrow \bar{n}$}
        \DP
        \hskip 1.5em
        \AXC{}
        \RL{\textsc{Var}}
        \UIC{$\langle \gamma, \mu \rangle \vdash_{expr}\ x \Rightarrow \mu\Big( \gamma(x) \Big)$}
        \DP
        \hskip 1.5em
        \AXC{}
        \RL{\textsc{Nil}}
        \UIC{$\langle \gamma, \mu \rangle \vdash_{expr}\ \texttt{\textbf{nil}} \Rightarrow 0$}
        \DP

        \vskip 1.5em

        \AXC{$\langle \gamma, \mu \rangle \vdash_{expr}\ e_1 \Rightarrow v_1 $}
        \AXC{$\langle \gamma, \mu \rangle \vdash_{expr}\ e_2 \Rightarrow v_2 $}
        \AXC{$\llbracket \otimes \rrbracket(v_1, v_2) = v$}
        \RL{\textsc{BinOp}}
        \TIC{$\langle \gamma, \mu \rangle \vdash_{expr}\ e_1 \otimes e_2 \Rightarrow v$}
        \DP 
    \end{center}
    \caption{Semantic inference rules for expressions, originally from~\cite{th:roopl}}
    \label{fig:semantics-expr-rules}
\end{figure}

As shown in figure~\ref{fig:semantics-expr-rules}, expression evaluation has no effects on the store. Logical values are represented by \textit{truthy} and \textit{falsy} values of any non-zero value and zero respectively. The evaluation of binary operators is presented in figure~\ref{fig:binary-operator-evaluation}.

\begin{figure}[ht]
    \begin{center}
        \AXC{$\langle \gamma, \mu \rangle \vdash_{expr}\ e \Rightarrow v$}
        \AXC{$\gamma(x) = l$}
        \AXC{$\mu(l)[v] = l'$}
        \AXC{$\mu(l') = w$}
        \RL{\textsc{ArrElem}}
        \QIC{$\langle \gamma, \mu \rangle \vdash_{expr}\ x[e] \Rightarrow w$}
        \DP
    \end{center}
    \caption{Extension to the semantic inference rules for expression in \rooplpp}
    \label{sec:semantics-expr-rule-extensions}
\end{figure}

For \rooplpp, we extend the expression ruleset with a single rule for array element variables shown in figure~\ref{sec:semantics-expr-rule-extensions}. As with the expressions inference rules in \textsc{Roopl}, this extension has no effect on the store.

\begin{figure}[ht]
    \begin{align*}
        &\llbracket \texttt{\textbf{+}} \rrbracket (v_1, v_2) &&= v_1 + v_2 &\llbracket \texttt{\textbf{\%}} \rrbracket (v_1, v_2)      &&&= v_1\ \text{mod}\ v_2\\
        &\llbracket \texttt{\textbf{-}} \rrbracket (v_1, v_2) &&= v_1 - v_2 &\llbracket \texttt{\textbf{\&}} \rrbracket (v_1, v_2)       &&&= v_1 \wedge v_2 \quad, bitwise\\
        &\llbracket \texttt{\textbf{*}} \rrbracket (v_1, v_2) &&= v_1 \times v_2 &\llbracket \texttt{\textbf{|}} \rrbracket  (v_1, v_2) &&&= v_1 \vee v_2 \quad, bitwise\\
        &\llbracket \texttt{\textbf{/}} \rrbracket (v_1, v_2) &&= v_1 \div v_2 &\llbracket \texttt{\textbf{\^}} \rrbracket (v_1, v_2)   &&&= v_1 \oplus v_2\\ 
        &\llbracket\texttt{\textbf{\&\&}}\rrbracket (v_1, v_2) &&= \begin{cases} 
            0 & \text{if}\ v_1 = 0 \vee v_2 = 0\\ 
            1 & \text{otherwise}
        \end{cases}        
        &\llbracket\texttt{\textbf{<=}}\rrbracket (v_1, v_2) &&&= \begin{cases} 
            0 & \text{if}\ v_1 \leq v_2\\ 
            1 & \text{otherwise} 
        \end{cases}\\
        &\llbracket\texttt{\textbf{||}}\rrbracket (v_1, v_2) &&= \begin{cases}
            0 & \text{if}\ v_1 = v_2 = 0\\ 
            1 & \text{otherwise}
        \end{cases}
        &\llbracket\texttt{\textbf{>=}}\rrbracket (v_1, v_2) &&&= \begin{cases}
            0 & \text{if}\ v_1 \geq v_2\\
            1 & \text{otherwise}
        \end{cases}\\
        &\llbracket\texttt{\textbf{<}}\rrbracket (v_1, v_2) &&= \begin{cases}
            1 & \text{if}\ v_1 < v_2\\ 
            0 & \text{otherwise}
        \end{cases}
        &\llbracket\texttt{\textbf{=}}\rrbracket (v_1, v_2) &&&= \begin{cases}
            0 & \text{if}\ v_1 = v_2\\
            1 & \text{otherwise}
        \end{cases}\\
        &\llbracket\texttt{\textbf{>}}\rrbracket (v_1, v_2) &&= \begin{cases}
            1 & \text{if}\ v_1 > v_2\\ 
            0 & \text{otherwise}
        \end{cases}
        &\llbracket\texttt{\textbf{!=}}\rrbracket (v_1, v_2) &&&= \begin{cases}
            0 & \text{if}\ v_1 \not= v_2\\
            1 & \text{otherwise}
        \end{cases}
    \end{align*} 
    \caption{Definition of binary expression operator evaluation, originally from~\cite{th:roopl}}
    \label{fig:binary-operator-evaluation}
\end{figure}

\subsection{Statements}
\label{subsec:semantics-statements}
The judgment
\begin{prooftree}
    \AXC{$\gamma \vdash^{\Gamma}_{stmt}\ s : \mu \rightleftharpoons \mu'$}
\end{prooftree}
defines the meaning of statements. We say that under environment $\gamma$, statement $s$ with class map $\Gamma$ reversibly transforms store $\mu$ to store $\mu'$. Figure~\ref{fig:semantic-statements},~\ref{fig:semantic-statements-cont} and~\ref{fig:semantic-statements-extension} defines the operational semantics of \rooplpp.

The following semantic rules have been simplified from the original \textsc{Roopl} semantics~\cite{th:roopl} to better accommodate the extended language. 

\begin{subfigures}
    \begin{figure}[!ht]
        \begin{center}
            % Skip
            \AXC{}
            \RL{\textsc{Skip}}
            \UIC{$\gamma \vdash^{\Gamma}_{stmt}\ \textbf{skip} : \mu \rightleftharpoons \mu$}
            \DP
    
            \vskip 2em
    
            % Seq
            \AXC{$\gamma \vdash^{\Gamma}_{stmt}\ s_1 : \mu \rightleftharpoons \mu'$}
            \AXC{$\gamma \vdash^{\Gamma}_{stmt}\ s_2 : \mu' \rightleftharpoons \mu''$}
            \RL{\textsc{Seq}}
            \BIC{$\gamma \vdash^{\Gamma}_{stmt}\ s_1\ s_2 : \mu \rightleftharpoons \mu''$}
            \DP
    
            \vskip 2em
    
            % AssVar
            \AXC{$\langle \gamma, \mu \rangle \vdash^{\Gamma}_{stmt}\ e \Rightarrow v$}
            \AXC{$\llbracket \odot \rrbracket \Big( \mu \Big( \gamma(x) \Big), v \Big) = v'$}
            \RL{\textsc{AssVar}}
            \BIC{$\gamma \vdash^{\Gamma}_{stmt}\ x \odot= e\ :\ \mu \rightleftharpoons \mu[\gamma(x) \mapsto v']$}
            \DP
    
            \vskip 2em
    
            % SwpVar
            \AXC{$\mu \Big( \gamma(x_1) \Big) = v_1$}
            \AXC{$\mu \Big( \gamma(x_2) \Big) = v_2$}
            \RL{\textsc{SwpVar}}
            \BIC{$\gamma \vdash^{\Gamma}_{stmt}\ x_1\ \texttt{\textbf{<=>}}\ x_2\ :\ \mu \rightleftharpoons \mu[\gamma(x_1) \mapsto v_2,\ \gamma(x_2) \mapsto v_1]$}
            \DP
    
            \vskip 2em
    
            % Loop Main
            \AXC{$\langle \gamma, \mu \rangle \vdash^{\Gamma}_{expr}\ e_1 \not\Rightarrow 0$}
            \AXC{$\gamma \vdash^{\Gamma}_{stmt}\ s_1 : \mu \rightleftharpoons \mu'$}
            \AXC{$\gamma \vdash^{\Gamma}_{loop}\ (e_1, s_1, s_2, e_2) : \mu' \rightleftharpoons \mu''$}
            \RL{\textsc{LoopMain}}
            \TIC{$\gamma \vdash^{\Gamma}_{stmt}\ \textbf{from}\ e_1\ \textbf{do}\ s_1\ \textbf{loop}\ s_2\ \textbf{until}\ e_2\ :\ \mu \rightleftharpoons \mu''$}
            \DP
    
            \vskip 2em
    
            % Loop Base
            \AXC{$\langle \gamma, \mu \rangle \vdash^{\Gamma}_{expr}\ e_2 \not\Rightarrow 0$}
            \RL{\textsc{LoopBase}}
            \UIC{$\gamma \vdash^{\Gamma}_{loop}\ (e_1, s_1, s_2, e_2)\ :\ \mu \rightleftharpoons \mu$}
            \DP
    
            \vskip 2em
    
            % Loop Rec
            \AXC{$\langle \gamma, \mu \rangle \vdash^{\Gamma}_{expr} e_2 \Rightarrow 0$}
            \DP
            \hskip 8em
            \AXC{$\gamma \vdash^{\Gamma}_{stmt} s_2 : \mu \rightleftharpoons \mu'$}
            \DP
            \vskip 0.5em
            \AXC{$\langle \gamma, \mu' \rangle \vdash^{\Gamma}_{expr} e_1 \Rightarrow 0$}
            \AXC{$\gamma \vdash^{\Gamma}_{stmt} s_1 : \mu' \rightleftharpoons \mu''$}
            \AXC{$\gamma \vdash^{\Gamma}_{loop} (e_1, s_1, s_2, e_2) : \mu'' \rightleftharpoons \mu'''$}
            \RL{\textsc{LoopRec}}
            \TIC{$\gamma \vdash^{\Gamma}_{loop}\ (e_1, s_1, s_2, e_2)\ :\ \mu \rightleftharpoons \mu'''$}
            \DP
    
            \vskip 2em
    
            % If True
            \AXC{$\langle \gamma, \mu \rangle \vdash^{\Gamma}_{expr} e_1 \not\Rightarrow 0$}
            \AXC{$\gamma \vdash^{\Gamma}_{stmt} s_1 : \mu \rightleftharpoons \mu'$}
            \AXC{$\langle \gamma, \mu' \rangle \vdash^{\Gamma}_{expr} e_2 \not\Rightarrow 0$}
            \RL{\textsc{IfTrue}}
            \TIC{$\gamma \vdash^{\Gamma}_{stmt}\ \textbf{if}\ e_1\ \textbf{then}\ s_1\ \textbf{else}\ s_2\ \textbf{fi}\ e_2\ :\ \mu \rightleftharpoons \mu'$}
            \DP
    
            \vskip 2em
    
            % If False
            \AXC{$\langle \gamma, \mu \rangle \vdash^{\Gamma}_{expr} e_1 \Rightarrow 0$}
            \AXC{$\gamma \vdash^{\Gamma}_{stmt} s_1 : \mu \rightleftharpoons \mu'$}
            \AXC{$\langle \gamma, \mu' \rangle \vdash^{\Gamma}_{expr} e_2 \Rightarrow 0$}
            \RL{\textsc{IfFalse}}
            \TIC{$\gamma \vdash^{\Gamma}_{stmt}\ \textbf{if}\ e_1\ \textbf{then}\ s_1\ \textbf{else}\ s_2\ \textbf{fi}\ e_2\ :\ \mu \rightleftharpoons \mu'$}
            \DP
        \end{center}
        \caption{Semantic inference rules for statements, modified from~\cite{th:roopl}}
        \label{fig:semantic-statements}
    \end{figure}
    
    \begin{figure}[!ht]
        \begin{center}           
            % Call
            \AXC{$\gamma(this) = l$}
            \DP
            \AXC{$\mu(l) = l' $}
            \DP
            \AXC{$\mu(l') = \Big\langle c, (l_1,\ \dots, l_m) \Big\rangle$}
            \DP
            \AXC{$\gamma(y_i) = l'_i$}
            \DP \vskip .5em
            \AXC{$\Gamma(c) = \Big\langle \overbrace{(x_1,\ \dots, x_m)}^{fields}, \overbrace{(\dots, \textbf{method}\ q(t_1 z_1,\ \dots, t_l z_k)\ s,\ \dots)}^{methods} \Big\rangle$}
            \DP \vskip .5em
            \AXC{$\gamma' = [this \mapsto l, x_1 \mapsto l_1,\ \dots, l_m \mapsto v_m, z_1 \mapsto l'_1, \dots, z_k \mapsto l'_k] $}
            \AXC{$\gamma' \vdash^\Gamma_{stmt}\ s\ : \mu \rightleftharpoons \mu'$}
            \RL{\textsc{Call}}
            \BIC{$\gamma \vdash^{\Gamma}_{stmt}\ \textbf{call}\ q(y_1,\ ...,\ y_n)\ :\ \mu \rightleftharpoons \mu'$}
            \DP
    
            \vskip 2em
    
            % Uncall
            \AXC{$\gamma \vdash^{\Gamma}_{stmt}\ \textbf{call}\ q(y_1,\ ...,\ y_n)\ :\ \mu' \rightleftharpoons \mu$}
            \RL{\textsc{Uncall}}
            \UIC{$\gamma \vdash^{\Gamma}_{stmt}\ \textbf{uncall}\ q(y_1,\ ...,\ y_n)\ :\ \mu \rightleftharpoons \mu'$}
            \DP

            \vskip 2em

            % Object call
            \AXC{$\gamma(x_0) = l $}
            \DP
            \AXC{$\mu(l) = l' $}
            \DP
            \AXC{$\mu(l') = \Big\langle c, (l_1,\ \dots, l_m) \Big\rangle$}
            \DP
            \AXC{$\gamma(y_i) = l'_i$}
            \DP \vskip .5em
            \AXC{$\Gamma(c) = \Big\langle \overbrace{(x_1,\ \dots, x_m)}^{fields}, \overbrace{(\dots, \textbf{method}\ q(t_1 z_1,\ \dots, t_l z_k)\ s,\ \dots)}^{methods} \Big\rangle$}
            \DP \vskip .5em
            \AXC{$\gamma' = [this \mapsto l, x_1 \mapsto l_1,\ \dots, l_m \mapsto v_m, z_1 \mapsto l'_1, \dots, z_k \mapsto l'_k] $}
            \AXC{$\gamma' \vdash^\Gamma_{stmt}\ s\ : \mu \rightleftharpoons \mu'$}
            \RL{\textsc{CallObj}}
            \BIC{$\gamma \vdash^{\Gamma}_{stmt}\ \textbf{call}\ x_0::q(y_1,\ ...,\ y_n)\ :\ \mu \rightleftharpoons \mu'$}
            \DP
    
            \vskip 2em
            
            % Object Uncall
            \AXC{$\gamma \vdash^{\Gamma}_{stmt}\ \textbf{call}\ x_0::q(y_1,\ ...,\ y_n)\ :\ \mu' \rightleftharpoons \mu$}
            \RL{\textsc{ObjUncall}}
            \UIC{$\gamma \vdash^{\Gamma}_{stmt}\ \textbf{uncall}\ x_0::q(y_1,\ ...,\ y_n)\ :\ \mu \rightleftharpoons \mu'$}
            \DP

            \vskip 2em

            % Object block
            \AXC{$\Gamma(c) = \Big\langle \overbrace{(x_1,\ \dots, x_m)}^{fields}, methods \Big\rangle $}
            \DP
            \AXC{$\gamma' = \gamma[x \mapsto l_0]$}
            \DP
            \AXC{$l_0 \notin\ \text{dom}(\mu)\ \dots\ l_m \notin\ \text{dom}(\mu) $}
            \DP
            \AXC{$\mu' = \mu\Bigg[\gamma'(x) \mapsto l_0, l_0 \mapsto \Big\langle c, (l_1,\ \dots, 1_m) \Big\rangle, l_1 \mapsto 0,\ \dots, l_m \mapsto 0\Bigg]$}
            \DP
            \AXC{$\gamma' \vdash^\Gamma_{stmt}\ s\ :\ \mu' \rightleftharpoons \mu''$}
            \DP
            \AXC{$\mu'' = \mu'''\Bigg[\gamma'(x) \mapsto l_0, l_0 \mapsto \Big\langle c, (l_1,\ \dots, 1_m) \Big\rangle, l_1 \mapsto 0,\ \dots, l_m \mapsto 0\Bigg]$}
            \RL{\textsc{ObjBlock}}
            \UIC{$\gamma \vdash^{\Gamma}_{stmt}\ \textbf{construct}\ c\ x \quad s \quad \textbf{destruct}\ x\ :\ \mu \rightleftharpoons \mu'''$}  
            \DP
        \end{center}
        \caption{Semantic inference rules for statements, modified from~\cite{th:roopl} (cont)}
        \label{fig:semantic-statements-cont}
    \end{figure}

    \begin{figure}[!ht]
        \begin{center}   
            % AssArrElemVar
            \AXC{$\langle \gamma, \mu \rangle \vdash^{\Gamma}_{stmt}\ e_1 \Rightarrow v_1$}
            \DP
            \AXC{$\langle \gamma, \mu \rangle \vdash^{\Gamma}_{stmt}\ e_2 \Rightarrow v_2$}
            \DP \vskip 1em
            \AXC{$\gamma(x) = l$}
            \AXC{$\mu(l)[v_1] = l'$}
            \AXC{$\mu(l') = w$}
            \AXC{$\llbracket \odot \rrbracket (w, v_2) = w'$}
            \RL{\textsc{AssArrElemVar}}
            \QIC{$\gamma \vdash^{\Gamma}_{stmt}\ x[e_1]\ \texttt{$\odot$=}\ e_2\ :\ \mu \rightleftharpoons \mu[l' \mapsto w']$}
            \DP
                
            \vskip 2em
                
            % Object New
            \AXC{$\Gamma(c) = \Big\langle \overbrace{(x_1,\ \dots, x_m)}^{fields}, methods \Big\rangle $}
            \AXC{$\gamma(x) = l$}
            \AXC{$l_0 \notin\ \text{dom}(\mu)\ \dots\ l_m \notin\ \text{dom}(\mu) $}
            \RL{\textsc{ObjNew}}
            \TIC{$\gamma \vdash^{\Gamma}_{stmt}\ \textbf{new}\ c\ x\ :\ \mu[l \mapsto 0] \rightleftharpoons \mu\Bigg[l \mapsto l_0, l_0 \mapsto \Big\langle c, (l_1,\ \dots, 1_m) \Big\rangle, l_1 \mapsto 0,\ \dots, l_m \mapsto 0\Bigg]$}  
            \DP
                
            \vskip 2em
                
            % Object delete
            \AXC{$\langle l, \gamma \rangle \vdash^{\Gamma}_{stmt}\ \textbf{new}\ c\ x\ :\ \mu' \rightleftharpoons \mu$}
            \RL{\textsc{ObjDelete}}
            \UIC{$\gamma \vdash^{\Gamma}_{stmt}\ \textbf{delete}\ c\ x\ :\ \mu \rightleftharpoons \mu'$}
            \DP

            \vskip 2em

            % Array New
            \AXC{$\langle \gamma, \mu \rangle \vdash^{\Gamma}_{stmt}\ e \Rightarrow n$}
            \AXC{$\gamma(x) = l$}
            \AXC{$\mu(l) = 0$}
            \AXC{$l' \notin\ \text{dom}(\mu)$}
            \RL{\textsc{ArrNew}}
            \QIC{$\gamma \vdash^{\Gamma}_{stmt}\ \textbf{new}\ a[e]\ x\ :\ \mu \rightleftharpoons \mu[l \mapsto l', l' \mapsto 0^n]$}  
            \DP
                
            \vskip 2em
                
            % Array delete
            \AXC{$\langle l, \gamma \rangle \vdash^{\Gamma}_{stmt}\ \textbf{new}\ a[e]\ x\ :\ \mu' \rightleftharpoons \mu$}
            \RL{\textsc{ArrDelete}}
            \UIC{$\gamma \vdash^{\Gamma}_{stmt}\ \textbf{delete}\ a[e]\ x\ :\ \mu \rightleftharpoons \mu'$}
            \DP

            % Copy
            \AXC{$\gamma(x) = l$}
            \AXC{$\gamma(x') = l'$}
            \AXC{$\mu(l) = v$}
            \RL{\textsc{Copy}}
            \TIC{$\gamma \vdash^{\Gamma}_{stmt}\ \textbf{copy}\ c\ x\ x'\ :\ \mu[l' \mapsto 0] \leftrightharpoons \mu[l' \mapsto v]$} 
            \DP

            \vskip 2em
            
            % Uncopy
            \AXC{$\langle l, \gamma \rangle \vdash^{\Gamma}_{stmt}\ \textbf{copy}\ c\ x\ x'\ :\ \mu' \leftrightharpoons \mu$}
            \RL{\textsc{Uncopy}}
            \UIC{$\gamma \vdash^{\Gamma}_{stmt}\ \textbf{uncopy}\ c\ x\ x'\ :\ \mu \rightleftharpoons \mu'$}
            \DP

            \vskip 2em

            % Local blocks
            \AXC{$\langle \gamma, \mu \rangle \vdash^{\Gamma}_{stmt}\ e_1 \Rightarrow v_1$}
            \DP \hskip 2em
            \AXC{$\langle \gamma, \mu' \rangle \vdash^{\Gamma}_{stmt}\ e_2 \Rightarrow v_2$}
            \DP \vskip 1em
            \AXC{$r \notin\ \text{dom}(\mu)$}
            \AXC{$\gamma[x \mapsto r] \vdash^{\Gamma}_{stmt}\ s : \mu[r \mapsto v_1] \rightleftharpoons \mu'[r \mapsto v_2]$}
            \RL{\textsc{LocalBlock}}
            \BIC{$\gamma \vdash^{\Gamma}_{stmt}\ \textbf{local}\ c\ x = e_1 \quad s \quad \textbf{delocal}\ x = e_2\ :\ \mu \rightleftharpoons \mu'$} 
            \DP
        \end{center}  
        \caption{Extension to the semantic inference rules for statements in \rooplpp}
        \label{fig:semantic-statements-extension}
    \end{figure}
\end{subfigures}

The inference rule \textsc{Skip} defines the operational semantics of \textbf{skip} statements and has no effects on the store $\mu$.

Rule \textsc{Seq} defines statement sequences where the store potentially is updated between each statement execution. 

Rule \textsc{AssVar} defines reversible assignment in which variable identifier $x$ under environment $\gamma$ is mapped to the value $v'$ resulting in an updated store $\mu'$. For variable swapping \textsc{SwpVar} defines how value mappings between two variables are exchanged in the updated store.

For loops and conditionals, Rules \textsc{LoopMain}, \textsc{LoopBase} and \textsc{LoopRec} define the meaning of \text{loop} statements and \text{IfTrue} and \text{IfFalse}, similarly to the operational semantics of Janus, as presented in~\cite{ty:ejanus}. \textsc{LoopMain} is entered if $e_1$ is true and each iteration enters \textsc{LoopRec} until $e_2$ is false, in which case \textsc{LoopBase} is executed. Similarly, if $e_1$ and $e_2$ are true, rule \textsc{IfTrue} is entered, executing the then-branch of the conditional. If $e_1$ and $e_2$ are false, the \textsc{IfFalse} rule is executed and the else-branch is executed.

Rules \textsc{Call}, \textsc{Uncall}, \textsc{CallObj} and \textsc{UncallObj} respectively define local and non-local method invocations. For local methods, method $q$ in current class $c$ should be of arity $n$ matching the number of arguments. The updated store $\mu'$ is obtained after statement body execution in the object environment. As local uncalling is the inverse of local calling, the direction of execution is simply reversed, and as such the input store a \textbf{call} statement serves as the output store of the \textbf{uncall} statement, similarly to techniques presented in~\cite{ty:janus, ty:ejanus}.

The statically scoped object blocks are defined in rule \textsc{ObjBlock}. The operation semantics of these blocks are similar to \textbf{local}-blocks from \textsc{Janus}. We add the reference $x$ to a new environment and afterwards map the location of $x$ to the object tuple at location $l_0$, containing the locations of all object fields, all of which, along with $l_0$, must be unused in $\mu$. The result store $\mu''$ is obtained after executing the body statement $s$ in store $\mu'$ mapping $x$ to object reference at $l_0$, as long as all object fields are zero-cleared in $\mu'''$ afterwards. If any of these conditions fail, the object block statement is undefined.

Figure~\ref{fig:semantic-statements-extension} shows the extensions to the semantics of \textsc{Roopl} with rules for \textbf{new}/\textbf{delete} and \textbf{copy}/\textbf{uncopy} statements, array element assignment and local blocks.

Rule \textsc{AssArrElemVar} defines reversible assignment to array elements. After evaluating expressions $e_1$ to $v_1$ and $e_2$ to $v_2$, the value at the location of variable $x[v_1]$ under environment $\gamma$ is mapped to the value $v_3$ resulting in an updated store $\mu'$.

Dynamic object construction and destruction is defined by rules \textsc{ObjNew} and \textsc{ObjDelete}. For construction, $x$ must be bound to a location $l$. We then make location $l$ point to a new pair consisting of the class name and a vector of $m$ new locations mapping object fields to locations. For destruction, $x$ is still bound to $l$ return $l$ to a null pointer. As with object blocks, it is the program itself responsible for zero-clearing object fields before destruction. If the object fields are not zero-cleared, the \textsc{ObjDelete} statement is undefined.

Array construction and destruction is very similar to object construction and destruction. The major difference is we bind the location to a vector of size equal to the evaluated expression result. For deletion, we return the location of $x$ to a null pointer and remove the binding to the vector from the store.

Object and array referencing is defined by rules \textsc{Copy} and \textsc{Uncopy}. A reference is created and a new store $\mu'$ obtained by mapping $x'$ to the reference $r$ which $x$ current maps to, if $c$ matches the tuple mapped to the location $l$. A reference is removed and a new store $\mu'$ obtained if $x$ and $x'$ maps to the same reference $r$ and $x'$ then is removed from the store.

Local blocks are as previously mentioned, semantically similar to object blocks, where the memory location of variable $x$ is mapped to an unused reference $r$ in the store $\mu$. Before body statement execution, we let $r$ bind to the evaluated value of $e_1$, $v_1$. The result store after body statement execution, $\mu'$ must have $r$ mapped to the expression value of $e_2$, $v_2$. $r$ is then zero-cleared using the value of expression evaluation and becomes unused again.

\subsection{Programs}
\label{subsec:semantics-programs}
The judgment
\begin{prooftree}
    \AXC{$\vdash_{prog} p \Rightarrow \sigma$}
\end{prooftree}
defines the meaning of programs. The class $p$ containing the main method is instantiated and the main function is executed with the partial function $\sigma$ as the result, mapping variable identifiers to values, correlating to the class fields of the main class.

\begin{figure}[ht]
    \begin{center}
        
        \AXC{$\Gamma = \text{gen}(c_1,\ ...,\ c_n)$}
        \DP
        \AXC{$\Gamma(c_1) = \Big( \overbrace{\{ \langle t_1, f_1 \rangle,\ ...,\ \langle t_n, f_n \rangle \}}^{fields},\ methods \Big)$}
        \DP
        \AXC{$\Big( \textbf{method main ()}\ s \Big) \in methods$}
        \DP
        \AXC{$\gamma = [f_1 \mapsto 1,\ ...,\ f_i \mapsto i]$}
        \DP
        \AXC{$\mu = [1 \mapsto 0,\ ...,\ i \mapsto 0,\ this \mapsto i+1, i+1 \mapsto \langle c_1, \gamma \rangle]$}
        \AXC{$\gamma \vdash^{\Gamma}_{stmt}\ s\ :\ \mu \rightleftharpoons \mu'$}
        \RL{\textsc{Main}}
        \BIC{$\vdash_{prog}\ c_1\ \cdots\ c_n \Rightarrow (\gamma, \mu')$}
        \DP
    \end{center}
    \caption{Semantic inference rules for programs, originally from~\cite{th:roopl}}
    \label{fig:semantics-programs}
\end{figure}

As with \textsc{Roopl} programs, the fields of the main method in the main class $c$ are bound in a new environment, starting at memory address $1$, as $0$ is reserved for \textbf{nil}. The fields are zero-initialized in the new store $\mu$ and address $i + 1$ which maps to the new instance of $c$. After body execution, store $\mu'$ is obtained.

\section{Program Inversion}
\label{sec:program-inversion}
In order to truly show that \rooplpp in fact is a reversible language, we must demonstrate and prove local inversion of statements is possible, such that any program written in \rooplpp, regardless of context, can be executed in reverse. \citeauthor{th:roopl} presented a statement inverter for \textsc{Roopl} in~\cite{th:roopl}, which maps statements to their inverse counterparts. Figure~\ref{fig:statement-inverter} shows the statement inverter, extended with the new \rooplpp statements for construction/destruction and referencing copying/copy removal.

\begin{figure}[ht]
    \begin{align*}
        &\mathcal{I}\llbracket \textbf{skip} \rrbracket = \textbf{skip} &&\mathcal{I}\llbracket s_1\ s_2 \rrbracket = \mathcal{I}\llbracket s_2 \rrbracket\ \mathcal{I}\llbracket s_1 \rrbracket\\
        &\mathcal{I}\llbracket x\ \texttt{\textbf{+=}}\ e \rrbracket = x\ \texttt{\textbf{-=}}\ e &&\mathcal{I}\llbracket x\ \texttt{\textbf{-=}}\ e \rrbracket = x\ \texttt{\textbf{+=}}\ e\\
        &\mathcal{I}\llbracket x\ \texttt{\textbf{\textasciicircum=}}\ e \rrbracket = x\ \texttt{\textbf{\textasciicircum=}}\ e &&\mathcal{I}\llbracket x\ \texttt{\textbf{<=>}}\ e \rrbracket = x\ \texttt{\textbf{<=>}}\ e\\
        &\mathcal{I}\llbracket x[e_1]\ \texttt{\textbf{+=}}\ e_2 \rrbracket = x[e_1]\ \texttt{\textbf{-=}}\ e_2 &&\mathcal{I}\llbracket x[e_1]\ \texttt{\textbf{-=}}\ e_2 \rrbracket = x[e_1]\ \texttt{\textbf{+=}}\ e_2\\
        &\mathcal{I}\llbracket x[e_1]\ \texttt{\textbf{\textasciicircum=}}\ e_2 \rrbracket = x[e_1]\ \texttt{\textbf{\textasciicircum=}}\ e_2 &&\mathcal{I}\llbracket x[e_1]\ \texttt{\textbf{<=>}}\ e_2 \rrbracket = x[e_1]\ \texttt{\textbf{<=>}}\ e_2\\
        &\mathcal{I} \llbracket \textbf{new}\ c\ x\rrbracket\ = \textbf{delete}\ c\ x  &&\mathcal{I} \llbracket \textbf{copy}\ c\ x\ x' \rrbracket\ = \textbf{uncopy}\ c\ x\ x'\\
        &\mathcal{I} \llbracket \textbf{delete}\ c\ x \rrbracket\ = \textbf{new}\ c\ x  &&\mathcal{I} \llbracket \textbf{uncopy}\ c\ x\ x' \rrbracket\ = \textbf{copy}\ c\ x\ x'\\
        &\mathcal{I} \llbracket \textbf{call}\ q(\dots) \rrbracket\ = \textbf{uncall}\ q(\dots) &&\mathcal{I} \llbracket \textbf{call}\ x::q(\dots) \rrbracket\ = \textbf{uncall}\ x::q(\dots)\\
        &\mathcal{I} \llbracket \textbf{uncall}\ q(\dots) \rrbracket\ = \textbf{call}\ q(\dots) &&\mathcal{I} \llbracket \textbf{uncall}\ x::q(\dots) \rrbracket\ = \textbf{call}\ x::q(\dots)\\
        &\mathcal{I} \llbracket \textbf{if}\ e_1\ \textbf{then}\ s_1\ \textbf{else}\ s_2\ \textbf{fi}\ e_2 \rrbracket &&= \textbf{if}\ e_1\ \textbf{then}\ \mathcal{I}\llbracket s_1 \rrbracket\ \textbf{else}\ \mathcal{I}\llbracket s_2 \rrbracket\ \textbf{fi}\ e_2\\
        &\mathcal{I} \llbracket \textbf{from}\ e_1\ \textbf{do}\ s_1\ \textbf{loop}\ s_2\ \textbf{until}\ e_2 \rrbracket\ &&=\ \textbf{from}\ e_1\ \textbf{do}\ \mathcal{I}\llbracket s_1 \rrbracket\ \textbf{loop}\ \mathcal{I}\llbracket s_2 \rrbracket\ \textbf{until}\ e_2\\
        &\mathcal{I} \llbracket\textbf{construct}\ c\ x\quad s\quad \textbf{destruct}\ x\rrbracket\ &&=\ \textbf{construct}\ c\ x\quad\mathcal{I}\llbracket s \rrbracket\quad\textbf{destruct}\ x\\
        &\mathcal{I} \llbracket\textbf{local}\ t\ x\ = e\ \quad s\quad \textbf{delocal}\ t\ x\ = e\rrbracket\ &&=\ \textbf{local}\ t\ x = e \quad\mathcal{I}\llbracket s \rrbracket\quad\textbf{delocal}\ t\ x\ = e
    \end{align*}
    \caption{\rooplpp statement inverter, extended from~\cite{th:roopl}}
    \label{fig:statement-inverter}
\end{figure}

Program inversion is conducted by recursive descent over components and statements. A proposed extension to the statement inverter for whole-program inversion is retained in the \rooplpp statement inverter. The extension covers a case that reveals itself during method calling. As a method call is equivalent to an uncall with the inverse method we simply change calls to uncalls during inversion, the inversion of the method body cancels out. The proposed extension, presented in~\cite{ty:janus, th:roopl}, simply avoids inversion of calls and uncalls, as shown in figure~\ref{fig:inverter-extension}.

\begin{figure}[ht]    
    \centering
    \begin{align*}
        &\mathcal{I}' \llbracket \textbf{call}\ q(\dots) \rrbracket\ = \textbf{call}\ q(\dots) &&\mathcal{I}' \llbracket \textbf{call}\ x::q(\dots) \rrbracket\ = \textbf{call}\ x::q(\dots)\\
        &\mathcal{I}' \llbracket \textbf{uncall}\ q(\dots) \rrbracket\ = \textbf{uncall}\ q(\dots) &&\mathcal{I}' \llbracket \textbf{uncall}\ x::q(\dots) \rrbracket\ = \textbf{uncall}\ x::q(\dots)\\ 
        &&\mathcal{I}' \llbracket s \rrbracket = \mathcal{I} \llbracket s \rrbracket
    \end{align*}
    \caption{Modified statement inverter for statements, originally from~\cite{th:roopl}}
    \label{fig:inverter-extension} 
\end{figure}

\subsection{Invertibility of Statements}
\label{subsec:invertibility-of-statements}
While the invertibility of statements remains untouched by the extensions made in \rooplpp, the following proof, originally presented in~\cite{th:roopl}, has been included for completeness.

If execution of a statement $s$ in store $\mu$ yields $\mu'$, then execution of the inverse statement, $\mathcal{I} \llbracket s \rrbracket$ in store $\mu'$ should yield $\mu$. Theorem~\ref{thm:invertibility-of-statements} shows that $\mathcal{I}$ is a statement inverter.\\

\begin{theorem}\label{thm:invertibility-of-statements}(Invertibility of statements, originally from~\cite{th:roopl})
    \begin{equation*}
        \overbrace{\langle l, \gamma \rangle\ \vdash_{stmt}^\Gamma\ s\ :\ \mu \rightleftharpoons \mu'}^{\mathcal{S}}\ \Longleftrightarrow\ \overbrace{\langle l, \gamma \rangle\ \vdash_{stmt}^\Gamma\ \mathcal{I}\llbracket s\rrbracket\ :\ \mu' \rightleftharpoons \mu}^{\mathcal{S}'}
    \end{equation*}
\end{theorem}

\begin{proof}\let\qed\relax
    By structural induction on the semantic derivation of $\mathcal{S}$ (omitted). It suffices to show that $\mathcal{S} \implies \mathcal{S}'$, as this can serve as proof of $\mathcal{S}' \implies \mathcal{S}$, as $\mathcal{I}$ is an involution.
\end{proof}    

\subsection{Type-Safe Statement Inversion}
\label{subsec:type-safe-statement-inversion}
Given a well-typed statement, the statement inverter $\mathcal{I}$ should always produce a well-typed, inverse statement in order to correctly support backwards determinism of injective functions. Theorem~\ref{thm:type-inversion} describes this.

\begin{theorem}\label{thm:type-inversion}(Inversion of well-typed statements, originally from~\cite{th:roopl})
    \begin{equation*}
        \overbrace{\langle \Pi,\ c\rangle\ \vdash_{stmt}^\Gamma\ s}^{\mathcal{T}}\ \implies\ \overbrace{\langle \Pi,\ c\rangle\ \vdash_{stmt}^\Gamma\ \mathcal{I}\llbracket s\rrbracket}^{\mathcal{T}'}
    \end{equation*}
\end{theorem}

\begin{proof}\let\qed\relax
    By structural induction on $\mathcal{T}$. Unmodified \textsc{Roopl} statements retained in \rooplpp has been omitted.

    \begin{itemize}
        \item Case $\mathcal{T} =$\vskip 1em
            \begin{center}
                \AXC{$\overbrace{\Pi(x) = \textbf{int}[\ ]}^{\mathcal{C}_1}$}
                \DP\\
                \AXC{$\overbrace{\Pi \vdash_{expr} e_1\ :\ \textbf{int}}^{\mathcal{E}_1}$}
                \AXC{$\overbrace{\Big( {\ x\ } \cup \text{vars}(e_1) \Big) \cap \text{vars}(e_2) = \emptyset}^{\mathcal{C}_2}$}
                \AXC{$\overbrace{\Pi \vdash_{expr}\ e_2 : \textbf{int}}^{\mathcal{E}_2}$}
                \RL{\textsc{T-ArrElemAss}}
                \TIC{$\langle \Pi, c \rangle \vdash^{\Gamma}_{stmt}\ x[e_1]\ \texttt{$\odot$=}\ e_2$}
                \DP\\
            \end{center}
            \vskip 1em
            In this case, we have $\mathcal{I}\llbracket x\ \texttt{$\odot$=}\ e \rrbracket = x\ \texttt{$\odot'$=}\ e$, for some $\odot'$. Therefore, $\mathcal{T}'$ will also be a derivation of rule \textsc{T-ArrElemAss}, and as such, we can simply reuse the conditions $\mathcal{C}_1, \mathcal{C}_2$ and the expressions $\mathcal{E}_1. \mathcal{E}_2$ in construction of $\mathcal{T}'$
            \begin{gather*}
                \AXC{$\overbrace{\Pi(x) = \textbf{int}[\ ]}^{\mathcal{C}_1}$}
                \DP\\
                \AXC{$\overbrace{\Pi \vdash_{expr} e_1\ :\ \textbf{int}}^{\mathcal{E}_1}$}
                \AXC{$\overbrace{\Big( {\ x\ } \cup \text{vars}(e_1) \Big) \cap \text{vars}(e_2) = \emptyset}^{\mathcal{C}_2}$}
                \AXC{$\overbrace{\Pi \vdash_{expr}\ e_2 : \textbf{int}}^{\mathcal{E}_2}$}
                \LL{$\mathcal{T}' =$}
                \TIC{$\langle \Pi, c \rangle \vdash^{\Gamma}_{stmt}\ x[e_1]\ \texttt{$\odot'$=}\ e_2$}
                \DP
            \end{gather*}

        \item Case $\mathcal{T} =$
            \AXC{$\overbrace{\Pi(x) = c'}^{\mathcal{C}_1}$}
            \RL{\textsc{T-ObjNew}}
            \UIC{$\langle \Pi, c \rangle  \vdash^{\Gamma}_{stmt}\ \textbf{new}\ c'\ x$} 
            \DP\\
            \vskip 1em
            In this case we have $\mathcal{I} \llbracket \textbf{new}\ c\ x\rrbracket\ = \textbf{delete}\ c\ x$, meaning $\mathcal{T}'$ must be of the form:
            \begin{equation*}
                \mathcal{T}' = \AXC{$\overbrace{\Pi(x) = c'}^{\mathcal{C}_2}$}          
                \UIC{$\langle \Pi, c \rangle  \vdash^{\Gamma}_{stmt}\ \textbf{delete}\ c'\ x$} 
                \DP
            \end{equation*} 

        \item Case $\mathcal{T} =$ 
            \AXC{$\overbrace{\Pi(x) = c'}^{\mathcal{C}_1}$}
            \RL{\textsc{T-ObjDlt}}            
            \UIC{$\langle \Pi, c \rangle  \vdash^{\Gamma}_{stmt}\ \textbf{delete}\ c'\ x$} 
            \DP\\
            \vskip 1em
            Inverse of the previous case, we now have $\mathcal{I} \llbracket \textbf{delete}\ c\ x\rrbracket\ = \textbf{new}\ c\ x$, meaning $\mathcal{T}'$ must be of the form: 
            \begin{equation*}
                \mathcal{T}' =  \AXC{$\overbrace{\Pi(x) = c'}^{\mathcal{C}_2}$}
                \UIC{$\langle \Pi, c \rangle  \vdash^{\Gamma}_{stmt}\ \textbf{new}\ c'\ x$} 
                \DP
            \end{equation*}

        \item Case $\mathcal{T} =$ 
            \AXC{$\overbrace{\text{arrayType}(a) \in \Big\{ \text{classIDs}, \textbf{int} \Big\}}^{\mathcal{C}_1}$}
            \AXC{$\overbrace{\Pi \vdash_{expr}\ e = \textbf{int}}^{\mathcal{E}}$}
            \AXC{$\overbrace{\Pi(x) = a[\ ]}^{\mathcal{C}_2}$}
            \RL{\textsc{T-ArrNew}}
            \TIC{$\langle \Pi, c \rangle  \vdash^{\Gamma}_{stmt}\ \textbf{new}\ a[e]\ x$} 
            \DP\\
            \vskip 1em
            In this case we still have $\mathcal{I} \llbracket \textbf{new}\ c\ x\rrbracket\ = \textbf{delete}\ c\ x$. Using $\mathcal{C}_1$ and $\mathcal{E}$, $\mathcal{T}'$ must be of the form:
            \begin{equation*}
                \mathcal{T}' = \AXC{$\overbrace{\text{arrayType}(a) \in \Big\{ \text{classIDs}, \textbf{int} \Big\}}^{\mathcal{C}_1}$}
                \AXC{$\overbrace{\Pi \vdash_{expr}\ e = \textbf{int}}^{\mathcal{E}}$}
                \AXC{$\overbrace{\Pi(x) = a[\ ]}^{\mathcal{C}_3}$}
                \TIC{$\langle \Pi, c \rangle  \vdash^{\Gamma}_{stmt}\ \textbf{delete}\ a[e]\ x$}
                \DP
            \end{equation*}

         \item Case $\mathcal{T} =$
            \AXC{$\overbrace{\text{arrayType}(a) \in \Big\{ \text{classIDs}, \textbf{int} \Big\}}^{\mathcal{C}_1}$}
            \AXC{$\overbrace{\Pi \vdash_{expr}\ e = \textbf{int}}^{\mathcal{E}}$}
            \AXC{$\overbrace{\Pi(x) = a[\ ]}^{\mathcal{C}_2}$}
            \RL{\textsc{T-ArrDlt}}
            \TIC{$\langle \Pi, c \rangle  \vdash^{\Gamma}_{stmt}\ \textbf{delete}\ a[e]\ x$}
            \DP\\
            \vskip 1em
            Similar to the object deletion case, we still have $\mathcal{I} \llbracket \textbf{delete}\ c\ x\rrbracket\ = \textbf{new}\ c\ x$. Using $\mathcal{C}_1$ and $\mathcal{E}$, $\mathcal{T}'$ must be of the form:
            \begin{equation*}
                \mathcal{T}' = \AXC{$\overbrace{\text{arrayType}(a) \in \Big\{ \text{classIDs}, \textbf{int} \Big\}}^{\mathcal{C}_1}$}
                \AXC{$\overbrace{\Pi \vdash_{expr}\ e = \textbf{int}}^{\mathcal{E}}$}
                \AXC{$\overbrace{\Pi(x) = a[\ ]}^{\mathcal{C}_3}$}
                \TIC{$\langle \Pi, c \rangle  \vdash^{\Gamma}_{stmt}\ \textbf{new}\ a[e]\ x$} 
                \DP
            \end{equation*}

         \item Case $\mathcal{T} =$ 
            \AXC{$\overbrace{\Pi(x) = c'}^{\mathcal{C}_1}$}
            \AXC{$\overbrace{\Pi(x') = c'}^{\mathcal{C}_2}$}
            \RL{\textsc{T-Cp}}
            \BIC{$\langle \Pi, c \rangle  \vdash^{\Gamma}_{stmt}   \ \textbf{copy}\ c'\ x\ x'$} 
            \DP\\
            \vskip 1em
            We have $\mathcal{I} \llbracket \textbf{copy}\ c\ x\ x' \rrbracket\ = \textbf{uncopy}\ c\ x\ x'$. Using $\mathcal{C}_1$, $\mathcal{T}'$ must as such be of the form
            \begin{equation*}
                \mathcal{T}' = \AXC{$\overbrace{\Pi(x) = c'}^{\mathcal{C}_1}$}
                \AXC{$\overbrace{\Pi(x') = c'}^{\mathcal{C}_3}$}
                \BIC{$\langle \Pi, c \rangle  \vdash^{\Gamma}_{stmt}\ \textbf{uncopy}\ c'\ x\ x'$}
                \DP
            \end{equation*}

        \item Case $\mathcal{T} =$
            \AXC{$\overbrace{\Pi(x) = c'}^{\mathcal{C}_1}$}
            \AXC{$\overbrace{\Pi(x') = c'}^{\mathcal{C}_2}$}
            \RL{\textsc{T-Ucp}}
            \BIC{$\langle \Pi, c \rangle  \vdash^{\Gamma}_{stmt}\ \textbf{uncopy}\ c'\ x\ x'$}
            \DP\\
            \vskip 1em
            We have $\mathcal{I} \llbracket \textbf{uncopy}\ c\ x\ x' \rrbracket\ = \textbf{copy}\ c\ x\ x'$. Using $\mathcal{C}_1$, $\mathcal{T}'$ must as such be of the form
            \begin{equation*}
                \mathcal{T}' = \AXC{$\overbrace{\Pi(x) = c'}^{\mathcal{C}_1}$}
                \AXC{$\overbrace{\Pi(x') = c'}^{\mathcal{C}_3}$}
                \BIC{$\langle \Pi, c \rangle  \vdash^{\Gamma}_{stmt}   \ \textbf{copy}\ c'\ x\ x'$} 
                \DP
            \end{equation*}
            
        \item Case $\mathcal{T} =$ 
            \AXC{$\overbrace{\langle \Pi, c \rangle \vdash_{expr}^{\Gamma}\ e_1}^{\mathcal{E}_1}$}
            \AXC{$\overbrace{\langle \Pi[x \mapsto c'], c \rangle \vdash^{\Gamma}_{stmt}\ s}^{\mathcal{S}}$}
            \AXC{$\overbrace{\langle \Pi, c \rangle \vdash_{expr}^{\Gamma}\ e_2}^{\mathcal{E}_2}$}
            \RL{\textsc{T-LocalBlock}}
            \TIC{$\langle \Pi, c \rangle \vdash^{\Gamma}_{stmt}\ \textbf{local}\ c'\ x = e_1\quad s\quad \textbf{delocal}\ c'\ x = e_2$} 
            \DP\\
            \vskip 1em
            We have $\mathcal{I} \llbracket\textbf{local}\ t\ x\ = e\ \quad s\quad \textbf{delocal}\ t\ x\ = e\rrbracket\ =\ \textbf{local}\ t\ x = e \quad\mathcal{I}\llbracket s \rrbracket\quad\textbf{delocal}\ t\ x\ = e$.

            By the induction hypothesis on $\mathcal{S}$, we obtain $\mathcal{S'}$ of $\langle \Pi[x \mapsto c'], c \rangle \vdash^{\Gamma}_{stmt}\ \mathcal{I} \llbracket s \rrbracket$. Using $\mathcal{E}_1$, $\mathcal{S'}$ and $\mathcal{E}_2$ we construct $\mathcal{T}'$
            \begin{equation*}
                \mathcal{T}' = \AXC{$\overbrace{\langle \Pi, c \rangle \vdash_{expr}^{\Gamma}\ e_1}^{\mathcal{E}_1}$}
                \AXC{$\overbrace{\langle \Pi[x \mapsto c'], c \rangle \vdash^{\Gamma}_{stmt}\ \mathcal{I} \llbracket s \rrbracket}^{\mathcal{S'}}$}
                \AXC{$\overbrace{\langle \Pi, c \rangle \vdash_{expr}^{\Gamma}\ e_2}^{\mathcal{E}_2}$}
                \TIC{$\langle \Pi, c \rangle \vdash^{\Gamma}_{stmt}\ \textbf{local}\ c'\ x = e_1\quad \mathcal{I} \llbracket s \rrbracket\quad \textbf{delocal}\ c'\ x = e_2$} 
                \DP
            \end{equation*}
    \end{itemize}
\end{proof}

Using these added cases to the original proof provided in~\cite{th:roopl}, Theorem~\ref{thm:type-inversion} shows that well-typedness is preserved over inversion of \rooplpp methods. As methods are well-typed if their body statement is well-typed, inversion of classes and programs also preserve well-typedness, as classes consists of methods and programs of classes, by using the class inverter presented in figure~\ref{fig:inverter-extension}.

\section{Computational Strength}
\label{sec:computational-strength}
Traditional, non-reversible programming languages have their computational strength measured in terms of their abilities to simulate the Turing machine (TM). If any arbitrary Turing machine can be implemented in some programming language, the language is said to be computationally universal or Turing-complete. In essence, Turing-completeness marks when a language can compute all computable functions. Reversible programming languages, like \textsc{Janus}, \textsc{Roopl} and \rooplpp, are not Turing-complete as they only are capable of computing injective, computable functions.

For determining computing strength of reversible programming languages,~\citeauthor{ty:ejanus} suggests that the reversible Turing machine (RTM) could serve as the baseline criterion~\cite{ty:ejanus}. As such, a reversible programming language is reversibly universal or r-Turing complete if it is able to simulate a reversible Turing machine cleanly, i.e. without generating garbage data. If garbage was on the tape, the function simulated by the machine would not be an injective function and as such, no garbage should be left after termination of the simulation.

\subsection{Reversible Turing Machines}
\label{subsec:reversible-turing-machine}
Before we show that \rooplpp in fact is r-Turing complete, we present the formalized reversible Turing machine definition, as defined in~\cite{ty:ejanus}.
\vspace{4mm}
\begin{definition}
    \label{def:quadruple-tm}(Quadruple Turing Machine)\vspace{4mm}\\
    \noindent A TM T is a tuple $(Q,\ \Gamma,\ b,\ \delta,\ q_s,\ q_f)$ where
    \begin{itemize}[label = {}, itemsep = 1pt]
        \item $Q$ is the finite non-empty set of states
        \item $\Gamma$ is the finite non-empty set of tape alphabet symbols
        \item $b\ \in\ \Gamma$ is the blank symbol
        \item $\delta\ :\ (Q\ \times\ \Gamma\ \times\ \Gamma\ \times\ Q)\ \cup\ (Q\ \times\ \{/\}\ \times\ \{L,\ R\}\ \times\ Q)$ is the partial function representing the transitions
        \item $q_s\ \in\ Q$ is the starting state
        \item $q_f\ \in\ Q$ is the final state
    \end{itemize}
    The symbols $L$ and $R$ represent the tape head shift-directions left and right. A quadruple is either a symbol rule of the form $(q_1,\ s_1,\ s_2,\ q_2)$ or a shift rule of the form $(q_1,\ /,\ d,\ q_2)$ where $q_1 \in Q$, $q_2 \in Q$, $s_1 \in \Gamma$, $s_2 \in \Gamma$ and $d$ being either $L$ or $R$.
    
    A symbol rule $(q_1,\ s_1,\ s_2,\ q_2)$ means that in state $q_1$, when reading $s_1$ from the tape, write $s_2$ to the tape and change to state $q_2$. A shift rule $(q_1,\ /,\ d,\ q_2)$ means that in state $q_1$, move the tape head in direction $d$ and change to state $q_2$.
\end{definition}
\vspace{4mm}
\begin{definition}
    \label{def:reversible-tm}(Reversible Turing Machine)\vspace{4mm}\\
    \noindent A TM T is a reversible TM iff, for any distinct pair of quadruples $(q_1,\ s_1,\ s_2,\ q_2)\ \in\ \delta_T$ and $(q'_1,\ s'_1,\ s'_2,\ q'_2)\ \in\ \delta_T$, we have
    \begin{itemize}[label = {}, itemsep = 1pt]
        \item $q_1\ =\ q'_1\ \implies\ (t_1\ \neq\ / \quad \wedge \quad t'_1\ \neq\ / \quad \wedge \quad t_1\ \neq\ t'_1)$ (forward determinism)
        \item $q_2\ =\ q'_2\ \implies\ (t_1\ \neq\ / \quad \wedge \quad t'_1\ \neq\ / \quad \wedge \quad t_2\ \neq\ t'_2)$ (backward determinism)
    \end{itemize}
\end{definition}

A RTM simulation implemented in \textsc{Roopl} by representing the set of states $\{q_1,\ \dots,\ q_n\}$ and the tape alphabet $\Gamma$ as integers and the rule $/$ and direction symbols $L$ and $R$ as the uppercase integer literals \inst{SLASH}, \inst{LEFT} and \inst{RIGHT} was presented in~\cite{th:roopl}. As \textsc{Roopl} contains no array or stack primitives, the transition table $\delta$ was suggested to be represented as a linked list of objects containing four integers ${\textbf{q1}}$, ${\textbf{s1}}$, ${\textbf{s2}}$ and ${\textbf{q2}}$ each, where ${\textbf{s1}}$ equals \inst{SLASH} for shift rules. In \rooplpp, we do, however, have an array primitive and as such, we can simply simulate transitions by having rules ${\textbf{q1}}$, ${\textbf{s1}}$, ${\textbf{s2}}$ and ${\textbf{q2}}$ represented as arrays, where the number of cells in each array is \inst{PC\_MAX}, in a similar fashion as shown in~\cite{ty:ejanus}.

\subsection{Tape Representation}
\label{subsec:tape-representation}
As with regular Turing machines, the Reversible Turing machines also have tapes of infinite length. Therefore, we must simulate tape growth in either direction.
\citeauthor{ty:ejanus} represented the tape using two stack primitives in the Janus RTM interpreter and \citeauthor{th:roopl} used list of objects. In \rooplpp, we could implement a stack, as objects are not statically scoped as in \textsc{Roopl}. However, in terms of easy of use, a doubly linked list implementation similar to the one presented in section~\ref{subsec:doubly-linked-list}, of simple cell objects containing \textit{value}, \textit{left}, \textit{right} and \textit{self} fields, is more intuitive.

As such, the tape head finds a tape cell by inspecting a specific element of the doubly linked list tape representation. When we move in either direction, we simply set the neighbour element as the new tape head and allocate a new neighbour for the new tape head cell, if we are at the end of the list, to simulate the infinitely-length tape. Reversibly, this means that when we move in the opposite direction, blank cells are deallocated if we are moving the tape head away from the cell currently neighbouring either end of the tape.

\begin{figure}[ht]
    \centering
\begin{lstlisting}[style = basic, language = roopl]
method moveRight(int symbol, Cell tapeHead)
    local Cell right = nil
    local Cell tmp = nil
    uncall tapeHead::getSymbol(symbol)   // Put symbol back in current cell
    call tapeHead::getRight(right)       // Get right neighbour

    if right = nil && symbol = BLANK then
        symbol ^= BLANK                  // Zero clear symbol
        new Cell right                   // Init new neighbour
        copy Cell right tmp              // Copy reference to self
        uncall right::getSelf(tmp)       // Store self reference
        uncall right::getLeft(tapeHead)  // Set tape head as left of new cell
        right <=> tapeHead 
    else          
        call right::getLeft(tmp)         // Get copy of tape head reference
        uncopy Cell tmp tapeHead         // Clear reference to tape head

        if tapeHead = nil && symbol = BLANK 
            call tmp::getSelf(tapeHead) // rev: set self pointer
            uncopy Cell tmp tapeHead    // rev: new self pointer
            delete Cell tmp             // rev: new left neighbour
            symbol ^= BLANK             
        else skip                       // In reverse:
        fi tmp = nil                    // Allocate new left if current is nil

        uncall right::getLeft(tmp)       // Put tape head reference back
        tapeHead <=> right
        call tapeHead::getRight(right)   // Get right of new tape head
        call tapeHead::getSymbol(symbol) // Get symbol of new tape head
    fi right = nil
    uncall tapeHead::getRight(right)    // Set right neighbour
    delocal Cell right = nil
    delocal Cell tmp = nil
\end{lstlisting}
    \caption{Method for moving the tape head in the RTM simulation}
    \label{fig:rtm-move-tape-head-method}
\end{figure}

Figure~\ref{fig:rtm-move-tape-head-method} shows the \textit{moveRight} method for moving the tape head right. If the current tape head has no instantiated right neighbour we construct one using the \textbf{new} statement. Uncalling this method will move the tape head left. If the tape head is empty after moving left, we simply allocate a new cell, thus allowing tape growth in both directions. 
\newpage

\subsection{Reversible Turing Machine Simulation}
\label{subsec:rtm-simulation}
Figure~\ref{fig:rtm-instruction-method} shows the modified method \textit{inst} from~\cite{ty:ejanus}, which executes a single instruction given the tape head, the current state, symbol, program counter and the four arrays representing the transition rules. As described above, we \textbf{call} \textit{moveRight} to move the tape head right and \textbf{uncall} to move the tape head left.

\begin{figure}[ht]
    \centering
\begin{lstlisting}[style = basic, language = roopl]
method inst(int state, int symbol, int[] q1, int[] s1, 
            int[] s2, int[] q2, int pc, Cell tapeHead)
    if state = q1[pc] && symbol = s1[pc] then   // Symbol rule:
        state += q2[pc]-q1[pc]                  // set state to q2[pc]
        symbol += s2[pc]-s1[pc]                 // set symbol to s2[pc]
    fi state = q2[pc] && symbol = s2[pc] 
    if state = q1[pc] && s1[pc] = SLASH then    // Move rule:
        state += q2[pc]-q1[pc]                  // set state to q2[pc] 
        if s2[pc] = RIGHT then  
            call moveRight(symbol, tapeHead)    // Move tape head right 
        fi s2[pc] = RIGHT   
        if s2[pc] = LEFT then   
            uncall moveRight(symbol, tapeHead)  // Move tape head left
        fi s2[pc] = LEFT
    fi state = q2[pc] && s1[pc] = SLASH
\end{lstlisting}
    \caption{Method for executing a single TM transition}
    \label{fig:rtm-instruction-method}
\end{figure}

Figure~\ref{fig:rtm-simulation-method} shows the simulate method which is the main method responsible for running the RTM simulation. The tape is extended in either direction when needed, and the program counter is incremented.

\begin{figure}[ht]
    \centering
\begin{lstlisting}[style = basic, language = roopl]
method simulate(Cell tapeHead, int state, int[] q1, int[] s1, int[] s2, int[] q2, int pc) 
    from state = Qs do
        pc += 1                               // Increment pc local int symbol = 0
        call tapeHead::getSymbol(symbol)      // Fetch current symbol 
        call inst(state, symbol, q1, s1, s2, q2, pc, tapeHead)
        uncall tapeHead::getSymbol(symbol)    // Zero-clear symbol delocal symbol = 0
        if pc = PC_MAX then                   // Reset pc 
            pc ^= PC_MAX
        else skip
        fi pc = 0 
    loop skip
    until state = Qf
\end{lstlisting}
    \caption{Main RTM simulation method}
    \label{fig:rtm-simulation-method}
\end{figure} 

Unlike the \textsc{Roopl} simulation, \rooplpp is not limited by stack allocated, statically-scoped objects. Due to this limitation, the \textsc{Roopl} RTM simulator cannot finish with the TM tape as its program output when the RTM halts, as the call stack of the simulation must unwind before termination. As objects in \rooplpp are not bound by this limitation, the TM tape will exist as the program output when the RTM halts.\footnote{We are here breaking the rule that a \textbf{new} statement must eventually be followed by a \textbf{delete} statement to free the data.}

Instantiating a RTM simulation consists of initializing an initial tape head cell, as well as the transition rule arrays. After initialization, the \textit{simulate} method is simply called and the simulation begins.

\newpage

\chapter{Dynamic Memory Management}
\label{chp:dynamic-memory-management}
In order to allow objects to live outside of static scopes, we need to utilize a different memory management technique, such that objects are not allocated on the stack. Dynamic memory management presents a method of storing objects in different memory structures, most commonly, a memory heap. Most irreversible, modern programming languages uses dynamic memory management in some form for allocating space for objects in memory. 

However, reversible, native support for complex data structures is a non-trivial matter to implement. Variable-sized records and frames need to be stored efficiently in a structured heap, while avoiding garbage build-up to maintain reversibility. A reversible heap manager layout has been proposed for a simplified version of the reversible functional language \textsc{RFun} and later expanded to allow references to avoid deep copying values~\cite{ha:heap, ty:rfun, tm:refcounting}.

This chapter presents a brief introduction to fragmentation, garbage and linearity and how these respectively are handled reversibly, and a discussion of various heap manager layouts considered for \rooplpp, along with their advantages and disadvantages in terms of implementation difficulty, garbage build-up and the OOP paradigm.

\section{Fragmentation}
\label{sec:fragmentation}
Efficient memory usage is an important matter to consider when designing a heap layout for a dynamic memory manager. In a stack allocating memory layout, the stack discipline is in effect, meaning only the most recently allocated data can be freed. This is not the case with heap allocation, where data can be freed regardless of allocation order. A potential side effect of this freedom, comes as a consequence of memory fragmentation. We distinguish different types of fragmentation as internal or external fragmentation. 

Internal fragmentation refers to unused space inside a memory block used to store an object, if, say, the object is smaller than the block it has been allocated to. External fragmentation occurs as blocks freed throughout execution are spread across the memory heap, resulting in \textit{fragmented} free space~\cite{tm:languages}.

\subsection{Internal Fragmentation}
\label{subsec:internal-fragmentation}
Internal fragmentation occurs in the memory heap when part of an allocated memory block is unused. This type of fragmentation can arise from a number of different scenarios, but mostly it originates from cases of \textit{over-allocation}, which occurs when the memory manager delegates memory larger than required to fit an object, due to e.g. fixed-block sizing. 

For an example, consider a scenario, in which we allocate memory for an object of size $m$ onto a simple, fixed-sized block heap. The fixed block size is $n$ and $m \neq n$. If $n > m$, internal fragmentation would occur of size $n-m$ for every object of size $m$ allocated in said heap. If $n < m$, numerous blocks would be required for allocation to fit our object. In this case the internal fragmentation would be of size $n - m\ mod\ n$ per allocated object of size $m$.

\begin{subfigures}
  \begin{figure}[ht]
    \centering
    \begin{tikzpicture}
      % Background
      \fill[fill = grey] (0, 0) rectangle (2, 1) node[midway] {};
  
      % Frame
      \draw (0,0) rectangle (3,1);
      \draw (3,0) rectangle (6,1);
      \draw[dashed] (2, 0) -- (2, 1); 
      
      \draw (1, 0.5) node{Object};
      \draw (4.5, 0.5) node{Free space};
  
      %braces
      \draw[decoration={brace,mirror,raise=5pt},decorate] (0,0) -- node[below=6pt] {$n$} (3,0);
      \draw[decoration={brace,raise=5pt},decorate] (0,1) -- node[above=6pt] {$m$} (2,1);
      \draw[decoration={brace,raise=5pt},decorate] (2,1) -- node[above=6pt] {$n-m$} (3,1); 
    \end{tikzpicture}
    \caption{Creation of internal fragmentation of size $n-m$ due to \textit{over-allocation}}
    \label{fig:internal-frag-example}
  \end{figure}
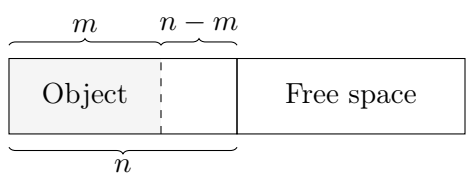

  \begin{figure}[ht]
    \centering
    \begin{tikzpicture}
      % Background
      \fill[fill = grey] (0, 0) rectangle (4, 1) node[midway] {};
  
      % Frame
      \draw (0,0) rectangle (6,1);
      \draw (6,0) rectangle (9,1);
      \draw[dashed] (4, 0) -- (4, 1);
      
      \draw (2.1, 0.5) node{Object};
      \draw (7.5, 0.5) node{Free space};
  
      %braces
      \draw[decoration={brace,mirror,raise=5pt},decorate] (0,0) -- node[below=6pt] {$n$} (3,0);
      \draw[decoration={brace,raise=5pt},decorate] (0,1) -- node[above=6pt] {$m$} (4,1);
      \draw[decoration={brace,mirror,raise=5pt},decorate] (4,0) -- node[below=6pt] {$n-m\ mod\ n$} (6,0); 
    \end{tikzpicture}
    \caption{Creation of internal fragmentation of size $n-m\ mod\ n$ due to \textit{over-allocation}}
    \label{fig:internal-frag-example-cont}
  \end{figure}
\end{subfigures}

Figure~\ref{fig:internal-frag-example} and~\ref{fig:internal-frag-example-cont} visualize the examples of internal fragmentation build-up from \textit{over-allocating} memory. 

It is difficult for the memory manager to reclaim wasted memory caused by internal fragmentation, as it usually originates from a design choice.
Intuitively, internal fragmentation can best be prevented by ensuring that the size of block(s) being used for allocating space for an object of size $m$ either match or sums to this exact size, when designing the layout. 

\subsection{External Fragmentation}
\label{subsec:external-fragmentation}
External fragmentation materializes in the memory heap when a freed block becomes partly or completely unusable for future allocation if, say, it is surrounded by allocated blocks but the size of the freed block is too small to contain objects on its own.

This type of fragmentation is generally a more substantial cause of problems than internal fragmentation, as the amount of wasted memory typically is larger and less predictable in external fragmentation blocks than in internal fragmentation blocks. Depending on the heap implementation, i.e. a layout using variable-sized blocks of, say, size $2^n$, the internal fragment size becomes considerable for large values of $n$. 

Non-allocatable external fragments become a problem when it is impossible to allocate space for a large object as a result of too many non-consecutive blocks scattered around the heap, caused by the external fragmentation. Physically, there is enough space to store the object, but not in the current heap state. In this scenario we would need to relocate blocks in such a manner that the fragmentation disperses, which is not possible to do reversibly.

Allocation and deallocation order is important in order to combat external fragmentation. For example, if we have a class $A$, which fit on one memory block of size $n$, and we have a class $B$, which fit on two memory blocks of size $n$ and limited memory space, we can easily reach a situation, where we cannot fit more $B$ objects due to external fragmentation.

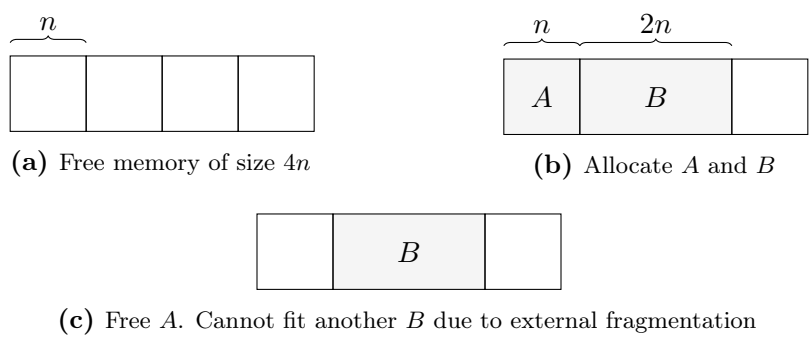
\begin{figure}[ht]
  \centering
  \begin{subfigure}{.4\textwidth}
    \centering
    \begin{tikzpicture}
      % Frame
      \draw (0, 0) rectangle (1, 1);
      \draw (1, 0) rectangle (2, 1);
      \draw (2, 0) rectangle (3, 1);
      \draw (3, 0) rectangle (4, 1);

      % Brace
      \draw[decoration={brace,raise=5pt},decorate] (0,1) -- node[above=6pt] {$n$} (1,1);
    \end{tikzpicture}
    \caption{\footnotesize Free memory of size $4n$}
  \end{subfigure}
  \begin{subfigure}{.4\textwidth}
    \centering
    \begin{tikzpicture}
      % Background
      \fill[fill = grey] (0, 0) rectangle (3, 1) node[midway] {};
  
      % Frame
      \draw (0, 0) rectangle (1, 1) node[pos=.5] {$A$};
      \draw (1, 0) rectangle (3, 1) node[pos=.5] {$B$};
      \draw (3, 0) rectangle (4, 1);

      % Braces
      \draw[decoration={brace,raise=5pt},decorate] (0,1) -- node[above=6pt] {$n$} (1,1);
      \draw[decoration={brace,raise=5pt},decorate] (1,1) -- node[above=6pt] {$2n$} (3,1);
    \end{tikzpicture}
    \caption{\footnotesize Allocate $A$ and $B$}
  \end{subfigure}
  \vskip 1em
  \begin{subfigure}{.6\textwidth}
    \centering
    \begin{tikzpicture}
      % Background
      \fill[fill = grey] (1, 0) rectangle (3, 1) node[midway] {};
  
      % Frame
      \draw (0, 0) rectangle (1, 1);
      \draw (1, 0) rectangle (3, 1) node[pos=.5] {$B$};
      \draw (3, 0) rectangle (4, 1);
    \end{tikzpicture}
    \caption{\footnotesize Free $A$. Cannot fit another $B$ due to external fragmentation}
  \end{subfigure}
  \caption{Example of external fragmentation caused for allocation and deallocation order}
  \label{fig:external-frag-example}
\end{figure}  

Figure~\ref{fig:external-frag-example} shows this example, where the allocation and deallocation order causes a situation, in which we cannot allocate any more $B$ objects, even though we physically have the required amount of free space in memory. 

\begin{figure}[ht]
  \centering
  \begin{subfigure}{.4\textwidth}
    \centering
    \begin{tikzpicture}
      % Frame
      \draw (0, 0) rectangle (1, 1);
      \draw (1, 0) rectangle (2, 1);
      \draw (2, 0) rectangle (3, 1);
      \draw (3, 0) rectangle (4, 1);

      % Brace
      \draw[decoration={brace,raise=5pt},decorate] (0,1) -- node[above=6pt] {$n$} (1,1);
    \end{tikzpicture}
    \caption{\footnotesize Free memory of size $4n$}
  \end{subfigure}
  \begin{subfigure}{.4\textwidth}
    \centering
    \begin{tikzpicture}
      % Background
      \fill[fill = grey] (0, 0) rectangle (3, 1) node[midway] {};
  
      % Frame
      \draw (0, 0) rectangle (2, 1) node[pos=.5] {$B$};
      \draw (2, 0) rectangle (3, 1) node[pos=.5] {$A$};
      \draw (3, 0) rectangle (4, 1);

      % Braces
      \draw[decoration={brace,raise=5pt},decorate] (0,1) -- node[above=6pt] {$2n$} (2,1);
      \draw[decoration={brace,raise=5pt},decorate] (2,1) -- node[above=6pt] {$n$} (3,1);
    \end{tikzpicture}
    \caption{\footnotesize Allocate $B$ and $A$}
  \end{subfigure}
  \vskip 1em
  \begin{subfigure}{.4\textwidth}
    \centering
    \begin{tikzpicture}
      % Background
      \fill[fill = grey] (0, 0) rectangle (2, 1) node[midway] {};
  
      % Frame
      \draw (0, 0) rectangle (2, 1) node[pos=.5] {$B$};
      \draw (2, 0) rectangle (3, 1);
      \draw (3, 0) rectangle (4, 1);
    \end{tikzpicture}
    \caption{\footnotesize Free $A$}
  \end{subfigure}
  \begin{subfigure}{.4\textwidth}
    \centering
    \begin{tikzpicture}
      % Background
      \fill[fill = grey] (0, 0) rectangle (4, 1) node[midway] {};
  
      % Frame
      \draw (0, 0) rectangle (2, 1) node[pos=.5] {$B$};
      \draw (2, 0) rectangle (4, 1) node[pos=.5] {$B$};
    \end{tikzpicture}
    \caption{\footnotesize Allocate another $B$}
  \end{subfigure}
  \caption{Example of avoiding external fragmentation using allocation and deallocation order}
  \label{fig:external-frag-example-cont}
\end{figure}
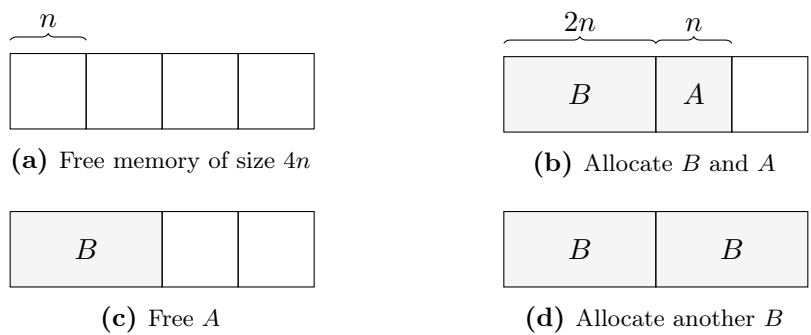  

Figure~\ref{fig:external-frag-example-cont} shows how changing allocation and deallocation order can combat external fragmentation.

\section{Memory Garbage}
\label{sec:memory-garbage}
A reversible computation should be garbage-free and as such it should be our goal to return the memory to its original state after program termination.

Traditionally, in non-reversible programming languages, freed memory blocks are simply re-added to the free list during deallocation and no modification of the actual data stored in the block is performed, as it simply is overwritten when the block is used later on. In the reversible setting we must return the memory block to its original state after the block has been freed (e.g. zero-cleared), to uphold the time-invertible and two-directional computational model. Figure~\ref{fig:injective-garbage-in-out} illustrates how the output data (or garbage) of an injective function $f$ is the input to its inverse function $f^{-1}$.

In heap allocation layouts, we maintain one or more free lists to keep track of free blocks during program execution, which are stored in memory, besides the heap representation itself. These free lists can essentially be considered garbage and as such, they must also be returned to their original state after execution. Furthermore, the heap itself can also be considered garbage and if it grows during execution, it should also be returned to its original size.

\begin{figure}[ht]
  \centering
    \begin{tikzpicture}
      % lines
      \draw[-] (-1,1.75) node[left]{$in$} -- (9,1.75) node[right] {$in$};
      \draw[-] (2,1.25) -- (6,1.25);
      \draw[-] (2,0.75) -- (6,0.75);
      \draw[-] (2,0.25) -- (6,0.25);
      \draw[-] (-1,-0.5) node[left]{$0$} -- (9,-0.5) node[right] {$out$};
      \draw[-] (4,0.24) node[circle,fill,inner sep=1pt] {} -- (4,-0.5) node {$\oplus$};
      
      % boxes
      \draw[fill = white] (0, 0) rectangle (2, 2) node[pos=0.5] {\Large $f$};
      \draw[fill = white] (6, 0) rectangle (8, 2) node[pos=0.5] {\Large $f^{-1}$};
      \node[diamond, fill=white, draw] at (4,1.25) {\scriptsize $garbage$};
    \end{tikzpicture}
    \caption{The "garbage" output of an injective function $f$ is the input to its inverse function $f^{-1}$}
    \label{fig:injective-garbage-in-out}
\end{figure}
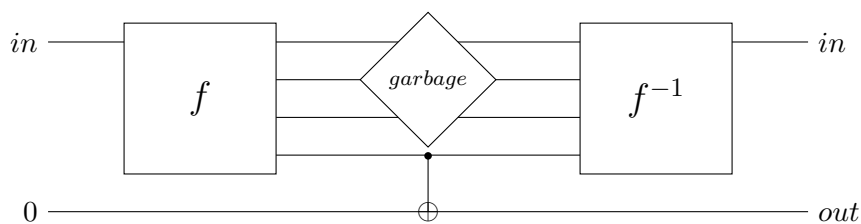

Returning the free list(s) to their original states is a non-trivial matter, which is highly dependent on the heap layout and free list design.~\citeauthor{ha:heap} introduced a dynamic memory manager which allowed heap allocation and deallocation, but without restoring the free list to its original state in~\cite{ha:heap}.~\citeauthor{ha:heap} argue that an unrestored free list can be considered harmless garbage in the sense that the free list residing in memory after termination is equivalent to a restored free list, as it contains the same blocks, but linked in a different order, depending on the order of allocation and deallocation operations performed during program execution. Figure~\ref{fig:equivalent-free-lists} illustrates how an inverse, injective function $f^{-1}$, whose non-inverse function $f$ computes something which modifies a given free lists, does not require the \textit{exact} output free list of $f$, but \textit{any} free list of same layout as input for the inverse function $f^{-1}$. The output free list of $f^{-1}$ will naturally be a further modified free list.

\begin{figure}[ht]
  \centering
  \begin{tikzpicture}
    \draw[-] (-1,1.75) node[left]{$free\ list$} -- (4.5,1.75) node[above]{$free\ list'$} -- (10.5,1.75) node[right] {$free\ list''$};
    \draw[-] (-1,0.25) node[left]{$input_f$} -- (3,0.25) node[right] {$output_f$};
    \draw[-] (6,0.25) node[left]{$input_g$} -- (10,0.25) node[right] {$output_g$};
          
    % boxes
    \draw[fill = white] (0, 0) rectangle (2, 2) node[pos=0.5] {\Large $f$};
    \draw[fill = white] (7, 0) rectangle (9, 2) node[pos=0.5] {\Large $g$};
  \end{tikzpicture}
  \caption{All free lists are considered equivalent "garbage" in terms of injective functions}
  \label{fig:equivalent-free-lists}
\end{figure}
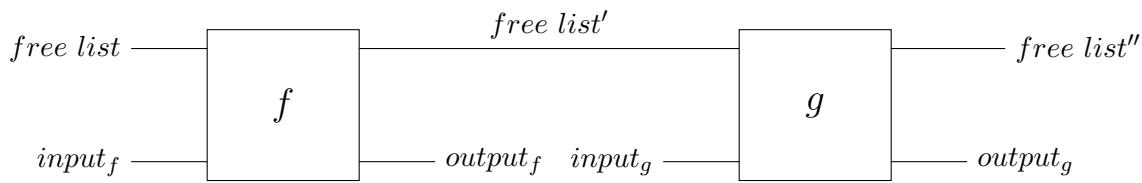

This intuitively leads to the question of garbage classification. In the reversible setting all functions are injective. Thus, given some $input_f$, in a reversible computation using heap allocation, the injective function $f$ produces some $output_f$ and some modified free list $free\ list'$, obtained after storing or freeing data in the heap during the execution of $f$ with an input $free\ list$. A future injective function in the program, function $g$, must thus take any modification of the original free list in addition to its input to produce its output $output_g$ and a potentially further modified free list, $free\ list''$. However, in the context of reversible heaps, we must consider all free lists as of equivalent, harmless garbage class and thus freely substitutable with each other, as injective functions still can drastically change the block layout, free list order, etc. during its execution in either direction. Figure~\ref{fig:equivalent-free-lists} shows how any free list can be passed between a function $f$ and a function $g$ further in the program.

\section{Linearity and Reference Counting}
\label{sec:referencing}
Programming languages use different approaches for storing and synchronizing variables and objects in memory. Typing \textit{linearity} is a distinction, which can reduce storage management and synchronization costs~\cite{hb:linearity}.

Reversible programming languages such as \textsc{Janus} and \textsc{Roopl} are linear in the sense that object and variable pointers cannot be copied and are only deleted during deallocation. Pointer copying greatly increases the flexibility of programming, especially in a reversible setting where zero-clearing is critical, at the cost of increased management in form of reference counting for e.g. objects. For variables, pointer copying is not particular interesting, nor would it add much flexibility as the values of a variable simply can be copied into statically-scoped local blocks. For objects however, tedious amounts of boilerplate work must be done if object $A$ and $B$ need to work on the same object $C$ and only one reference to each object is allowed. Copying is not an option as field modification in one copy does not affect the other copies.  

\citeauthor{tm:refcounting} presented the reversible functional language \textsc{Rcfun} which use reference counting to allow multiple pointers to the same memory nodes as well as a translation from \textsc{Rcfun} into \textsc{Janus} in \cite{tm:refcounting}. In \textsc{Rcfun}, reference counting is used to manage and trace the number of pointer copies made by respectively incrementing and decrementing a \textit{reference count} stored in the memory node, whenever the original node pointer is copied or a copy pointer is deleted. For the presented heap manage, deletion of object nodes was only allowed when no references to a node remained.

In non-reversible languages, reference counting is also used in garbage collection by automatically deallocating unreachable objects and variables which contains no referencing.

\section{Heap Manager Layouts}
\label{sec:heap-manager-layout}
Heap managers can be implemented in numerous ways. Different layouts yield different advantages when allocating memory, finding a free block or when collecting garbage. As our goal is to construct a garbage-free heap manager, our finalized design should emphasize and reflect this objective in particular. Furthermore, we should attempt to allocate and deallocate memory as efficiently as possible, as merging and splitting of blocks is a non-trivial problem in a reversible setting and to avoid problematic fragmentation.

For the sake of simplicity, we will not consider the issue of retrieving memory pages reversibly. A reversible operating system is a long-term dream of the reversible researcher and as reversible programming language designers, we assume that \rooplpp will be running in an environment, in which an operating system will be supplying memory pages and their mappings. As such, the following heap memory designs reflect this preliminary assumption, that we can always query the operating system for more memory. 

Historically, most object-oriented programming languages utilize a dynamic memory manager during program execution. In lower-level languages such as \textsc{C} or \textsc{C\texttt{++}}, memory management is manual and allocation has to be stated explicitly and with the requested size through the \textbf{malloc} statement and deallocated using the \textbf{free} statement. Modern languages, such as \textsc{Java} and \textsc{Python}, \textit{automagically} allocates and frees space for objects and variable-sized arrays by utilizing their dynamic memory manager and garbage collector to dispatch \textbf{malloc}- and \textbf{free}-like operations to the operating system and managing the obtained memory blocks in private heap(s)~\cite{wh:cpp_memory, bv:jvm, py:memory}. The heap layout of these managers vary from language to language and compiler to compiler.

Previous work on reversible heap manipulation has been done for reversible functional languages in~\cite{ha:heap, jsk:translation, tm:garbage}.

\citeauthor{ha:heap} presented a static heap structure consisting of \textsc{Lisp}-inspired constructor cells of fixed size and a single free list for the reversible function language \textsc{Rfun} in~\cite{ha:heap}. \citeauthor{tm:refcounting} presented an implementation in \textsc{Janus} of reversible reference counting under the assumption of \citeauthor{ha:heap}'s heap manager in~\cite{tm:refcounting}. Building on the previous work, \citeauthor{tm:garbage} later presented a reversible intermediate language \textsc{Ril} and an implementation in \textsc{Ril} of a reversible heap manager, which uses reference counting and hash-consing to achieve garbage collection in~\cite{tm:garbage}.

We do not consider reference counting or garbage collection in the layouts presented in the following sections, but we later show how the selected layout for \rooplpp is extended with reference counting in section~\ref{sec:referencing-compilation}.

\subsection{Memory Pools}
\label{subsec:memory-pools}
The simplest heap layout we can design uses fixed-sized blocks. This design is also known as memory pools, as memory is allocated from "pools" of fixed-sized blocks regardless of the record size.
To model these pools of fixed-sized blocks, we simply use a linked list of identically sized free block cells, which we maintain over execution.
While the fixed-block layout is simple and relatively easy in terms of implementation it is also largely uninteresting as it provides little to no options, besides sizing of the fixed-blocks, to combat fragmentation.

This layout comes with a few options in terms of the actual heap layout. If we only allow allocation of consecutive, adjacent free blocks, we should keep the free list sorted. If the free list is not sorted, and we have to allocate an object which requires $n$ blocks, we have to iterate the free list $n^2$ times in the worst case to find a chain of consecutive blocks large enough to fit the object. The sorting part itself is non-trivial matter. Furthermore, we need some overhead storage inside the object to contains the references of the blocks occupied by the object, or some other structure which can be used when deallocating the object and returning all the blocks to the free list. If we allow allocation of non-consecutive blocks, larger amounts of bookkeeping is required as we need to store knowledge of when and where the object is split.

Figures~\ref{fig:external-frag-example} and~\ref{fig:external-frag-example-cont} from earlier in this chapter, in section~\ref{subsec:external-fragmentation} on page~\pageref{fig:external-frag-example} illustrates examples with consecutive, fixed-sized block allocation.

\subsection{One Heap Per Record Size}
\label{subsec:one-heap-per-record-size}
Instead of allocating space for objects from a single free list and heap, we could design an approach which uses one heap per record size, known as a multi-heap layout. The respective classes and their sizes are easily identified during compile time from which the amount of heaps and free list will be initialized. This means the layout is very dynamic and potentially can change drastically in terms of the amount of heaps utilized depending on the input program. 

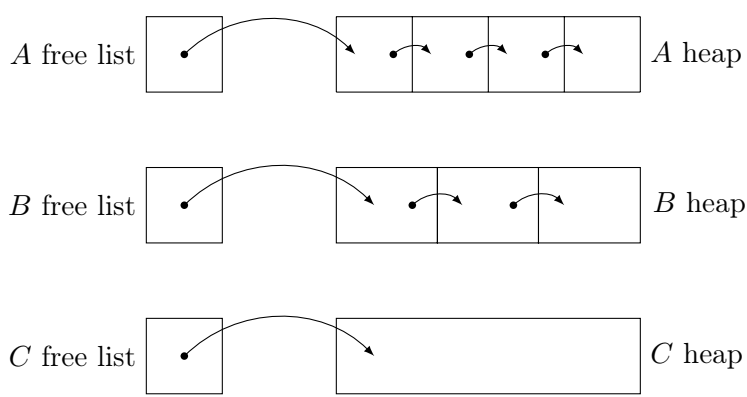
\begin{figure}[ht]
  \centering
  \begin{tikzpicture}
    % Heaps
    \draw[step=1cm] (0,4) grid (4,5);
    \node[right] at (4,4.5) {$A$ heap}; 
    \draw (0,2) rectangle (1.33, 3) rectangle (2.66,2) rectangle node[right, xshift=.7cm] {$B$ heap} (4,3);
    \draw (0,0) rectangle node[right, xshift=2cm] {$C$ heap} (4,1);

    % Free lists
    \draw (-2.5, 4) rectangle node[left, xshift=-0.5cm] {$A$ free list} (-1.5, 5);
    \draw (-2.5, 2) rectangle node[left, xshift=-0.5cm] {$B$ free list} (-1.5, 3);
    \draw (-2.5, 0) rectangle node[left, xshift=-0.5cm] {$C$ free list} (-1.5, 1);
    
    % Arrow for 1st heap pair
    \node[circle,fill,inner sep=1pt] at (-2, 4.5) {};
    \draw[-latex] (-2, 4.5) to[out=45, in=135] (0.25, 4.5);
    \node[circle,fill,inner sep=1pt] at (0.75, 4.5) {};
    \draw[-latex] (0.75, 4.5) to[out=45, in=135] (1.25, 4.5);
    \node[circle,fill,inner sep=1pt] at (1.75, 4.5) {};
    \draw[-latex] (1.75, 4.5) to[out=45, in=135] (2.25, 4.5);
    \node[circle,fill,inner sep=1pt] at (2.75, 4.5) {};
    \draw[-latex] (2.75, 4.5) to[out=45, in=135] (3.25, 4.5);

    % Arrow for 2nd heap pair
    \node[circle,fill,inner sep=1pt] at (-2, 2.5) {};
    \draw[-latex] (-2, 2.5) to[out=45, in=135] (0.5, 2.5);
    \node[circle,fill,inner sep=1pt] at (1, 2.5) {};
    \draw[-latex] (1, 2.5) to[out=45, in=135] (1.66, 2.5);

    \node[circle,fill,inner sep=1pt] at (2.33, 2.5) {};
    \draw[-latex] (2.33, 2.5) to[out=45, in=135] (3, 2.5);

    % Arrow for 3rd heap pair
    \node[circle,fill,inner sep=1pt] at (-2, 0.5) {};
    \draw[-latex] (-2, 0.5) to[out=45, in=135] (0.5, 0.5);
  \end{tikzpicture}
  \caption{Memory layout using one heap per record size}
  \label{fig:one-heap-per-record-size}
\end{figure}

Figure~\ref{fig:one-heap-per-record-size} illustrates three heaps with respective free lists for three classes $A$, $B$ and $C$ of size $n$, $2n$ and $4n$. Each heap is represented as a simple linked list with the free list simply being a pointer to the first free block in the heap. 

The advantage of this approach would be effective elimination of internal and external fragmentation, as each heap fits their targeted record perfectly, making each allocation and deallocation tailored to the size of the record obtained from a static analysis during compilation, resulting in no over-allocation and no unusable chunks of freed memory appearing during varying deallocation order. Implementation-wise, allocation of an object of a given class simply becomes the task of popping the head of the respective free list, which can easily be determined at compile time. The deallocation is simply adding a new head to the free list.

Listing~\ref{lst:one-heap-per-record-size} outlines the allocation algorithm for this layout written in extended \textsc{Janus} from~\cite{ty:ejanus}. We assume that the heads of the free lists are stored in a single array primitive, such that the free list for records of size $n$ are indexed at $n-2$ and $n > 2$ (as every record needs some overhead) and that we have heaps for continuous size range with no gaps.

The algorithm consists of an entry point named \textbf{malloc} and a recursion body named \textbf{malloc1}. Given a zero-cleared pointer $p$, the size of the object we are allocating $o_{size}$ and the array of free lists primitive, the recursion body is called after initializing a $counter$, which is an index into the free lists array and a counter size, $c_{size}$, which is the block size of the current free list the $counter$ is indexed in. The recursion body first updates the free list index until we find a free list with a size greater or equal to the size of the object we are allocating. Once such a free list has been found, the head of the free list is simply popped and the next block is set as the new head.\\

\begin{lstlisting}[caption={Allocation algorithm for one heap per record size implemented in extended Janus}, language=janus, style=basic, label={lst:one-heap-per-record-size}]
  procedure malloc(int p, int osize, int freelists[])
    local int counter = 0
    local int csize = 2
    call malloc1(p, osize, freelists, counter, csize)
    delocal int csize = 2
    delocal int counter = 0

  procedure malloc1(int p, int osize, int freelists[], int counter, int csize)
    if (csize < osize) then
      counter += 1
      csize += 1
      call malloc1(p, osize, freelists, counter, csize) 
      csize -= 1
      counter -= 1
    else
      p += freelists[counter]
      freelists[counter] -= p

      // Swap head of free list with next block of p 
      freelists[counter] ^= M(p)
      M(p) ^= freelists[counter]
      freelists[counter] ^= M(p)
    fi csize < osize   
\end{lstlisting}

The obvious disadvantage to this layout is the amount of bookkeeping and workload associated with growing and shrinking a heap and its neighbours, in case the program requests additional memory from the operating system. In real world object-oriented programming, most classes feature a small number of fields, very rarely more than 16. 

Additionally, helper classes of other sizes would spawn additional heaps and bookkeeping work, making the encapsulation concept of OOP rather unattractive, for the optimization-oriented reversible programmer. 

Finally, while internal and external fragmentation is effectively eliminated, we are left with additional and considerable amounts of garbage in forms of all the heaps and free lists initialized in memory. If two record types only differ one word in size, two heaps would be initialized. Each heap intuitively need to be initialized with a chunk of memory from the underlying operating system such that objects can be allocated on their respective heaps, regardless of the number of times the heap is used during program execution. This is an obvious space requirement increase over the previously presented layout, and on average, the amount of required memory for a program compiled using this approach would probably be larger, than some of the following layouts, due to unoptimized heap utilization and sharing. Heap of some sizes may be mostly empty when another is full, resulting in wasted memory.

\subsection{One Heap Per Power-Of-Two}
\label{subsec:one-heap-per-power-of-two}
To address the issues of the previous heap manager layout, we can optimize the amounts of heaps required by introducing a relatively small amount of internal fragmentation. Instead of having a heap per record size, we could have a heap per power-of-two. Records would be stored in the heap closest to their respective size and, as such, we reduce the number of heaps needed, as many different records can be stored in the same heap. Records of size $5$, $6$, $7$ and $8$ would in the former layout be stored in four different heaps, where they would be stored in a single heap using this layout. Figure~\ref{fig:one-heap-per-power-of-two} illustrates the free lists and heaps up to $2^n$. 

\begin{figure}[ht]
  \centering
  \begin{tikzpicture}
    % Heaps
    \draw[step=1cm] (0,4) grid (4,5);
    \node[right] at (4,4.5) {$1$ heap}; 
    \draw (0,2) rectangle (2,3) rectangle node[right, xshift=1cm] {$2$ heap} (4,2);
    \draw (0,0) rectangle node[right, xshift=2cm] {$3$ heap} (4,1);

    \draw (0,-3) -- (0, -2);
    \draw (0,-3) -- (1.5, -3);
    \draw (0,-2) -- (1.5, -2);
    \draw (2.5,-3) -- (4, -3);
    \draw (2.5,-2) -- (4, -2);
    \draw (4,-3) -- (4, -2);
    \node[right] at (4,-2.5) {$2^n$ heap};  

    % Free lists
    \draw (-2.5, 4) rectangle node[left, xshift=-0.5cm] {$1$ free list} (-1.5, 5);
    \draw (-2.5, 2) rectangle node[left, xshift=-0.5cm] {$2$ free list} (-1.5, 3);
    \draw (-2.5, 0) rectangle node[left, xshift=-0.5cm] {$4$ free list} (-1.5, 1);
    \draw (-2.5, -3) rectangle node[left, xshift=-0.5cm] {$2^n$ free list} (-1.5, -2);
    
    % Arrow for 1st heap pair
    \node[circle,fill,inner sep=1pt] at (-2, 4.5) {};
    \draw[-latex] (-2, 4.5) to[out=45, in=135] (0.25, 4.5);
    \node[circle,fill,inner sep=1pt] at (0.75, 4.5) {};
    \draw[-latex] (0.75, 4.5) to[out=45, in=135] (1.25, 4.5);
    \node[circle,fill,inner sep=1pt] at (1.75, 4.5) {};
    \draw[-latex] (1.75, 4.5) to[out=45, in=135] (2.25, 4.5);
    \node[circle,fill,inner sep=1pt] at (2.75, 4.5) {};
    \draw[-latex] (2.75, 4.5) to[out=45, in=135] (3.25, 4.5);

    % Arrow for 2nd heap pair
    \node[circle,fill,inner sep=1pt] at (-2, 2.5) {};
    \draw[-latex] (-2, 2.5) to[out=45, in=135] (0.5, 2.5);
    \node[circle,fill,inner sep=1pt] at (1.5, 2.5) {};
    \draw[-latex] (1.5, 2.5) to[out=45, in=135] (2.5, 2.5);

    % Arrow for 3rd heap pair
    \node[circle,fill,inner sep=1pt] at (-2, 0.5) {};
    \draw[-latex] (-2, 0.5) to[out=45, in=135] (0.5, 0.5);

    % Arrow for 4th heap pair
    \node[circle,fill,inner sep=1pt] at (-2, -2.5) {};
    \draw[-latex] (-2, -2.5) to[out=45, in=135] (0.5, -2.5);

    % Braces
    \draw[decoration={brace,raise=5pt},decorate] (0,5) -- node[above=6pt] {$1$} (1,5);
    \draw[decoration={brace,raise=5pt},decorate] (0,3) -- node[above=6pt] {$2$} (2,3);
    \draw[decoration={brace,raise=5pt},decorate] (0,1) -- node[above=6pt] {$4$} (4,1);
    \draw[decoration={brace,raise=5pt},decorate] (0,-2) -- node[above=6pt] {$2^n$} (4,-2);

    % Dots
    \node[above] at (-2, -1) {$\vdots$};
    \node[above] at (2, -1) {$\vdots$};
    \node[] at (2, -2) {$\dots$};
    \node[] at (2, -3) {$\dots$};
  \end{tikzpicture}
  \caption{Memory layout using one heap per power-of-two}
  \label{fig:one-heap-per-power-of-two}
\end{figure}
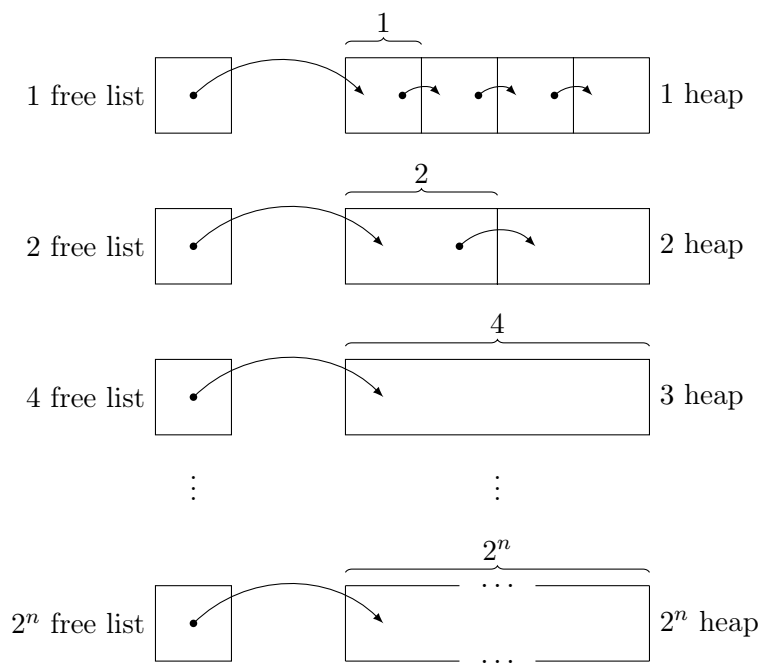

Internal fragmentation does become a problem for very large records, as blocks are only of size $2^n$. An object of size $65$ would fit in a $128$ sized block, resulting in considerable amounts of wasted memory space in form of internal fragmentation. However, in the real world, most records are small and allocation of records causing this much amount of fragmentation is an unlikely scenario. To avoid large amounts of internal fragmentation building up when allocating large records, we could allocate space for large objects using smaller blocks. If a record exceeds some limit, which has been determined the cutoff point, one kilobyte for an example, we could split it into $\sqrt{n}$ sized chunks and use blocks of that size instead. This would reduce the amount of internal fragmentation at the cost of increased bookkeeping.
For smaller records, very minimal amounts of internal fragmentation occur. 

The number of heaps needed for a computation can be determined at compile time by finding the smallest and largest record sizes and ensuring we have heaps to fit these effectively. The allocation process consists of determining the closest $2^n$ to the size of the record we are allocating and then simply popping the head of the respective free list.

Listing~\ref{lst:one-heap-per-power-of-two} shows a modified \textbf{malloc1} recursion body for the power-of-two approach. Once again, we assume our array of free lists contains the head of each free list, such that index $n$ is the head of the free list of size $2^{n+1}$. Instead of incrementing the counter size by one, as in the former layout algorithm, we double it, using the shown \textbf{double} procedure. Besides this change, the algorithm remains unchanged and still assumes each heap has been initialized along with the free lists.\\

\begin{lstlisting}[caption={Allocation algorithm for one heap per power-of-two implemented in extended Janus}, language=janus, style=basic, label={lst:one-heap-per-power-of-two}]
  procedure double(int target)
    local int current = target
    target += current
    delocal int current = target / 2

  procedure malloc1(int p, int osize, int freelists[], int counter, int csize)
    if (csize < osize) then
        counter += 1
        call double(csize)
        call malloc1(p, osize, freelists, counter, csize) 
        uncall double(csize)
        counter -= 1
    else
      if freelists[counter] != 0 then
        p += freelists[counter]
        freelists[counter] -= p

        // Swap head of free list with next block of p
        freelists[counter] ^= M(p)
        M(p) ^= freelists[counter]
        freelists[counter] ^= M(p)
      else
        counter += 1
        call double(csize)
        call malloc1(p, osize, freelists, counter, csize)
        uncall double(csize)
        counter -= 1
      fi freelists[counter] = 0 || p != freelists[counter]
    fi csize < osize   
\end{lstlisting}

\subsection{Shared Heap, Record Size-Specific Free Lists}
\label{subsec:shared-heap}
A natural proposal, considering the disadvantages of the previously presented designs, would be using a shared heap instead of record-specific heaps. 
This way, we ensure minimal fragmentation when allocating and freeing as the different free lists ensure that allocation of an object wastes as little memory as possible. By only keeping one heap, we eliminate the growth/shrinking issues of the multiple heap layout. 

There is, however, still a considerable amount of bookkeeping involved in maintaining multiple free lists. Having mixed-size blocks in a single heap is also a task which might prove difficult to accomplish reversibly. How initialization and destruction of said heap should work is not clear. As with the multiple heap version of this layout, we are still left with the issues surrounding two records which only differs one word in size. In the former layout, two heaps were required to store records of these types. In this layout, we need to store two block sizes in our heap to allocate these records, with no internal fragmentation. We could allow these objects to be allocated on similarly-sized blocks, if we round the calculated class sizes up to, say, a power-of-two. We would essentially have a shared heap, power-of-two-specific free lists layout.

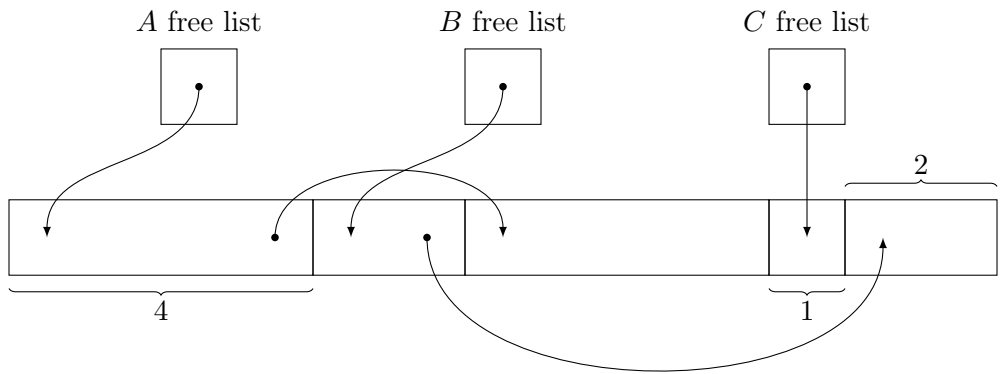
\begin{figure}[ht]
  \centering
  \begin{tikzpicture}
    % Heap
    \draw (0, 0) rectangle (4, 1);
    \draw (4, 0) rectangle (6, 1);
    \draw (6, 0) rectangle (10, 1);
    \draw (10, 0) rectangle (11, 1);
    \draw (11, 0) rectangle (13, 1);

    % Free lists
    \draw (2, 2) rectangle (3, 3) node[midway, above, yshift=0.6cm] {$A$ free list}; 
    \draw (6, 2) rectangle (7, 3) node[midway, above, yshift=0.6cm] {$B$ free list};
    \draw (10, 2) rectangle (11, 3) node[midway, above, yshift=0.6cm] {$C$ free list};

    % Arrows 1st heap
    \node[circle,fill,inner sep=1pt] at (2.5, 2.5) {};
    \draw[-latex] (2.5, 2.5) to[out=270, in=90] (0.5, 0.5);

    \node[circle,fill,inner sep=1pt] at (3.5, 0.5) {};
    \draw[-latex] (3.5, 0.5) to[out=90, in=90] (6.5, 0.5);

    % Arrows 2nd heap
    \node[circle,fill,inner sep=1pt] at (6.5, 2.5) {};
    \draw[-latex] (6.5, 2.5) to[out=270, in=90] (4.5, 0.5);

    \node[circle,fill,inner sep=1pt] at (5.5, 0.5) {};
    \draw[-latex] (5.5, 0.5) to[out=-90, in=-90] (11.5, 0.5);

    % Arrows 3rd heap
    \node[circle,fill,inner sep=1pt] at (10.5, 2.5) {};
    \draw[-latex] (10.5, 2.5) to[out=270, in=90] (10.5, 0.5);

    % Braces
    \draw[decoration={brace, mirror, raise=5pt},decorate] (0,0) -- node[below=6pt] {$4$} (4,0);
    \draw[decoration={brace, mirror, raise=5pt},decorate] (10,0) -- node[below=6pt] {$1$} (11,0);
    \draw[decoration={brace, raise=5pt},decorate] (11,1) -- node[above=6pt] {$2$} (13,1);
  \end{tikzpicture}
  \caption{Record size-specific free lists on a shared heap (powers of two)}
  \label{fig:shared-heap}
\end{figure}

As the only change in this design are the heaps themselves, the allocation process remains unchanged from the one presented in listing~\ref{lst:one-heap-per-record-size} or listing~\ref{lst:one-heap-per-power-of-two} if we use the power-of-two approach. Figure~\ref{fig:shared-heap} visualizes the shared heap and the free lists of this layout.

\subsection{Buddy Memory}
\label{subsec:buddy-memory}
The Buddy Memory layout utilizes blocks of variable-sizes of the power-of-two, typically with one free list per power-of-two using a shared heap. When allocating an object of size $m$, we simply check the free lists for a free block of size $n$, where $n \geq m$. Is such a block found and if $n > m$, we split the block into two halves recursively, until we obtain the smallest block capable of storing $m$. When deallocating a block of size $m$, we do the action described above in reverse, thus merging the blocks again, where possible~\cite{dk:taocp}.

\begin{figure}[ht]
  \centering
  \begin{subfigure}{.5\textwidth}
    \centering
    \begin{tikzpicture}[scale=0.75]
      % Boxes
      \draw[step=1] (2,2) grid (6,3);
      \draw (0,0) rectangle (8, 1);

      % Labels
      \node[above] at (4, 3.2) {Free lists}; 
      \node[above] at (2.5, 2.2) {$2^1$}; 
      \node[above] at (3.5, 2.2) {$2^2$}; 
      \node[above] at (4.5, 2.2) {$2^3$}; 
      \node[above] at (5.5, 2.2) {$2^4$}; 
      
      % Arrows
      \node[circle,fill,inner sep=1pt] at (5.5, 2.2) {};
      \draw[-latex] (5.5, 2.2) to[out=270, in=90] (0.5, 0.5);
    \end{tikzpicture}
    \caption{\footnotesize Initial memory block}
  \end{subfigure}%
  \begin{subfigure}{.5\textwidth}
    \centering
    \begin{tikzpicture}[scale=0.75]
      % Fills
      \draw[fill=grey] (7,0) rectangle (8,1);

      % Boxes
      \draw[step=1] (2,2) grid (6,3);
      \draw (0,0) rectangle (8, 1);

      % Labels
      \node[above] at (4, 3.2) {Free lists}; 
      \node[above] at (2.5, 2.2) {$2^1$}; 
      \node[above] at (3.5, 2.2) {$2^2$}; 
      \node[above] at (4.5, 2.2) {$2^3$}; 
      \node[above] at (5.5, 2.2) {$2^4$}; 

      % Lines
      \draw (4,0) -- (4, 1);
      \draw (6,0) -- (6, 1);
      \draw (7,0) -- (7, 1);
      
      % Arrows
      \node[circle,fill,inner sep=1pt] at (4.5, 2.2) {};
      \draw[-latex] (4.5, 2.2) to[out=270, in=90] (0.5, 0.5);

      \node[circle,fill,inner sep=1pt] at (3.5, 2.2) {};
      \draw[-latex] (3.5, 2.2) to[out=270, in=90] (4.5, 0.5);

      \node[circle,fill,inner sep=1pt] at (2.5, 2.2) {};
      \draw[-latex] (2.5, 2.2) to[out=270, in=90] (6.5, 0.5);
    \end{tikzpicture}
    \caption{\footnotesize Allocate an object of size $2^1$}
  \end{subfigure}%
  \vskip 1em
  \begin{subfigure}{.5\textwidth}
    \centering
    \begin{tikzpicture}[scale=0.75]
      % Fills
      \draw[fill=grey] (7,0) rectangle (8,1);
      \draw[fill=grey] (0,0) rectangle (4,1);

      % Boxes
      \draw[step=1] (2,2) grid (6,3);
      \draw (0,0) rectangle (8, 1);

      % Labels
      \node[above] at (4, 3.2) {Free lists}; 
      \node[above] at (2.5, 2.2) {$2^1$}; 
      \node[above] at (3.5, 2.2) {$2^2$}; 
      \node[above] at (4.5, 2.2) {$2^3$}; 
      \node[above] at (5.5, 2.2) {$2^4$}; 

      % Lines
      \draw (4,0) -- (4, 1);
      \draw (6,0) -- (6, 1);
      \draw (7,0) -- (7, 1);
      
      % Arrows
      \node[circle,fill,inner sep=1pt] at (3.5, 2.2) {};
      \draw[-latex] (3.5, 2.2) to[out=270, in=90] (4.5, 0.5);

      \node[circle,fill,inner sep=1pt] at (2.5, 2.2) {};
      \draw[-latex] (2.5, 2.2) to[out=270, in=90] (6.5, 0.5);
    \end{tikzpicture}
    \caption{\footnotesize Allocate an object of size $2^3$}
  \end{subfigure}%
  \begin{subfigure}{.5\textwidth}
    \centering
    \begin{tikzpicture}[scale=0.75]
      % Fills
      \draw[fill=grey] (7,0) rectangle (8,1);
      \draw[fill=grey] (0,0) rectangle (4,1);
      \draw[fill=grey] (4,0) rectangle (6,1);

      % Boxes
      \draw[step=1] (2,2) grid (6,3);
      \draw (0,0) rectangle (8, 1);

      % Labels
      \node[above] at (4, 3.2) {Free lists}; 
      \node[above] at (2.5, 2.2) {$2^1$}; 
      \node[above] at (3.5, 2.2) {$2^2$}; 
      \node[above] at (4.5, 2.2) {$2^3$}; 
      \node[above] at (5.5, 2.2) {$2^4$}; 

      % Lines
      \draw (4,0) -- (4, 1);
      \draw (6,0) -- (6, 1);
      \draw (7,0) -- (7, 1);
      
      % Arrows
      \node[circle,fill,inner sep=1pt] at (2.5, 2.2) {};
      \draw[-latex] (2.5, 2.2) to[out=270, in=90] (6.5, 0.5);
    \end{tikzpicture}
    \caption{\footnotesize Allocate an object of size $2^2$}
  \end{subfigure}%
  \caption{Buddy Memory block allocation example}
  \label{fig:buddy-memory-block-splitting}
\end{figure}
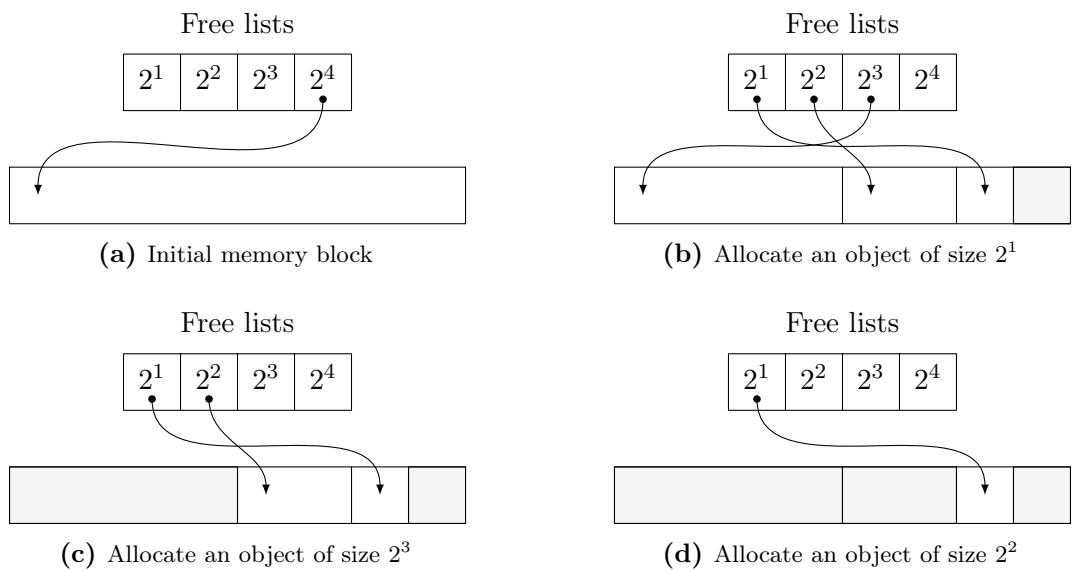

Figure~\ref{fig:buddy-memory-block-splitting} illustrates an example of block splitting during allocation in the buddy system. Originally, one block of free memory is available. When allocating a record three factors smaller than the original block, three splits occurs. 

This layout is somewhat of a middle ground between the previous three designs, addressing a number of problems found in these. The Buddy Memory layout uses a single heap for all record-types, thus eliminating the problems related to moving adjacent heaps reversibly in a multi-heap layout. To optimize the problems around initializing a usable amount of variable-sized blocks in a shared heap, we simply initialize one large block in the buddy system, which we will split into smaller parts during execution and merge where possible when freed.

The main drawback from this layout is the amount of internal fragmentation. As we only allocate blocks of a power-of-two size, substantial internal fragmentation follows when allocating large records, i.e. allocating a block of size 128 for a record of size 65. However, as most real world programs uses much smaller sized records, we do not consider this a very frequent scenario. As discussed in section~\ref{subsec:one-heap-per-power-of-two}, we would split large records into chunks of $\sqrt{n}$ at the cost of additional bookkeeping.

Implementation-wise, this design would require doubling and halving of numbers related to the power-of-two. This action translates well into the reversible setting, as a simply bit-shifting directly gives us the desired result.\\

\begin{lstlisting}[caption={The Buddy Memory algorithm implemented in extended Janus}, language=janus, style=basic, label={lst:buddy-memory}]
  procedure malloc1(int p, int osize, int freelists[], int counter, int csize)
    if (csize < osize) then
        counter += 1
        call double(csize)
        call malloc1(p, osize, freelists, counter, csize) 
        uncall double(csize)
        counter -= 1
    else
        if freelists[counter] != 0 then
            p += freelists[counter]
            freelists[counter] -= p

            // Swap head of free list with next block of p
            freelists[counter] ^= M(p)
            M(p) ^= freelists[counter]
            freelists[counter] ^= M(p)
        else
            counter += 1
            call double(csize)
            call malloc1(p, osize, freelists, counter, csize)
            uncall double(csize)
            counter -= 1
            freelists[counter] += p
            p += csize
        fi freelists[counter] = 0 || p - csize != freelists[counter]
    fi csize < osize   
\end{lstlisting}

Listing~\ref{lst:buddy-memory} shows the Buddy Memory algorithm implemented in the extended Janus variant with local blocks from~\cite{ty:ejanus}. For simplification, object sizes are rounded to the nearest power-of-two during compile-time. The algorithm extends on the one heap per power-of-two algorithm presented in listing~\ref{lst:one-heap-per-power-of-two}, page~\pageref{lst:one-heap-per-power-of-two}.
The body of the allocation function is still executed recursively until a free list for a $2^n$ larger than the size of the object has been found. Once found, we continue searching until we have found a non-empty free list. If the non-empty free list for a $2^n$ larger than the object is found, the head of the list is popped and the popped block is split recursively, until a block the desired size is obtained. Throughout the splitting process, empty free lists are updated when a larger free block is split into a block which fits into those lists.

Since a split block is always added as two blocks to an empty free list, we can only merge adjacent blocks if they are the only two blocks in a free list.

\newpage

\chapter{Compilation}
\label{chp:compilation}
The following chapter presents the considerations and translation schemas used in the process of translating \rooplpp to the reversible low-level machine language \textsc{Pisa}. As \rooplpp is an extension of \textsc{Roopl}, many techniques are carried directly over, and have as such been left out.

Before presenting the \rooplpp compiler, a brief overview of the memory layout and modeling of the \textsc{Roopl} compiler, which the \rooplpp compiler is a continuation of, is provided. 

\section{The \textsc{Roopl} to \textsc{Pisa} Compiler}
\label{sec:roopl-to-pisa-compiler}
\citeauthor{th:roopl} presented a proof-of-concept compiler along with the design for \textsc{Roopl}. The compiler translates well-typed \textsc{Roopl} programs into the reversible machine language \textsc{Pisa} in~\cite{th:roopl}. The \textsc{Roopl} compiler (\textsc{RooplC}) is written in \textsc{Haskell} and hosted at \url{https://github.com/TueHaulund/ROOPLC}.

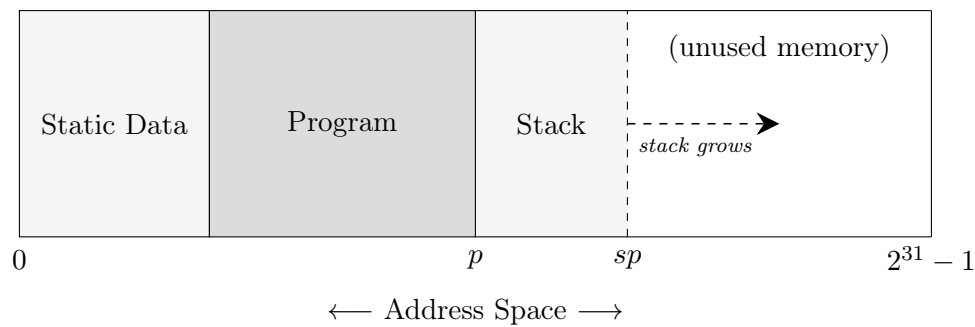
\begin{figure}[ht]
    \centering
    \begin{tikzpicture}
        \fill[fill = grey] (6, 0) rectangle (8, 3) node[midway] {Stack};
        \draw (6, 0) -- (12, 0);
        \draw (6, 3) -- (12, 3);
        \draw (12, 0) -- (12, 3);
        \draw[dashed] (8, 0) -- (8, 3);
    
        \filldraw[fill = grey, draw = black] (0, 0) rectangle (2.5, 3) node[midway] {Static Data};
        \filldraw[fill = darkgrey, draw = black] (2.5, 0) rectangle (6, 3) node[midway] {Program};
        
        \node at (10, 2.5) {(unused memory)};
        \node at (0, -.3) {$0$};
        \node at (6, -.3) {$p$};
        \node at (8, -.3) {$sp$};
        \node at (12, -.3) {$2^{31} - 1$};
        \node at (6, -1) {$\longleftarrow$ Address Space $\longrightarrow$};
        \draw[arrow, dashed] (8, 1.5) -- (10, 1.5);
        \node at (8.9, 1.2) {\scriptsize{\textit{stack grows}}};
    \end{tikzpicture}
    \caption{Memory layout of a \textsc{Roopl} program, originally from~\cite{th:roopl}}
    \label{fig:roopl-memory-layout}
\end{figure}

Figure~\ref{fig:roopl-memory-layout} shows the memory layout of a compiled \textsc{Roopl} program. The layout consists of a static storage segment, the program segment and the stack. 

\begin{figure}[ht]
    \centering
    \begin{subfigure}[t]{.32\textwidth}
        \vskip 0pt
        \centering
        \begin{tikzpicture}
            \draw[dashed] (0, 1.5) -- (0, 2);
            \draw[dashed] (3, 1.5) -- (3, 2);
            \filldraw[fill = grey, draw = black] (0, 1) rectangle (3, 1.5) node[midway] {addr(vtable)};
            \filldraw[fill = darkgrey, draw = black] (0, .5) rectangle (3, 1) node[midway] {x};
            \filldraw[fill = grey, draw = black] (0, 0) rectangle (3, .5) node[midway] {y};
            \draw[dashed] (0, 0) -- (0, -.5);

            \draw[dashed] (3, 0) -- (3, -.5);
    
            \node at (-.3, 1.25) {\texttt{+}$0$};
            \node at (-.3, .75) {\texttt{+}$1$};
            \node at (-.3, .25) {\texttt{+}$2$};
            \draw[->] (3.5, 1.25) -- (3.1, 1.25);
            \node[rotate = 270] at (3.7, 1.25) {$r_{shape}$};
            
            \node at (1.5, 2.5) {\textbf{Shape}};
        \end{tikzpicture}
    \end{subfigure}
    \begin{subfigure}[t]{.32\textwidth}
        \vskip 0pt
        \centering
        \begin{tikzpicture}
            \draw[dashed] (0, 1.5) -- (0, 2);
            \draw[dashed] (3, 1.5) -- (3, 2);
            \filldraw[fill = grey, draw = black] (0, 1) rectangle (3, 1.5) node[midway] {addr(vtable)};
            \filldraw[fill = darkgrey, draw = black] (0, .5) rectangle (3, 1) node[midway] {x};
            \filldraw[fill = grey, draw = black] (0, 0) rectangle (3, .5) node[midway] {y};
            \filldraw[fill = darkgrey, draw = black] (0, -.5) rectangle (3, 0) node[midway] {radius};
            \draw[dashed] (0, -.5) -- (0, -1);
            \draw[dashed] (3, -.5) -- (3, -1);
    
            \node at (-.3, 1.25) {\texttt{+}$0$};
            \node at (-.3, .75) {\texttt{+}$1$};
            \node at (-.3, .25) {\texttt{+}$2$};
            \node at (-.3, -.25) {\texttt{+}$3$};
            \draw[->] (3.5, 1.25) -- (3.1, 1.25);
            \node[rotate = 270] at (3.7, 1.25) {$r_{circ}$};
            
            \node at (1.5, 2.5) {\textbf{Circle}};
        \end{tikzpicture}
    \end{subfigure}
    \begin{subfigure}[t]{.32\textwidth}
        \vskip 0pt
        \centering
        \begin{tikzpicture}
            \draw[dashed] (0, 1.5) -- (0, 2);
            \draw[dashed] (3, 1.5) -- (3, 2);
            \filldraw[fill = grey, draw = black] (0, 1) rectangle (3, 1.5) node[midway] {addr(vtable)};
            \filldraw[fill = darkgrey, draw = black] (0, .5) rectangle (3, 1) node[midway] {x};
            \filldraw[fill = grey, draw = black] (0, 0) rectangle (3, .5) node[midway] {y};
            \filldraw[fill = darkgrey, draw = black] (0, -.5) rectangle (3, 0) node[midway] {a};
            \filldraw[fill = grey, draw = black] (0, -1) rectangle (3, -.5) node[midway] {b};
            \draw[dashed] (0, -1) -- (0, -1.5);
            \draw[dashed] (3, -1) -- (3, -1.5);
    
            \node at (-.3, 1.25) {\texttt{+}$0$};
            \node at (-.3, .75) {\texttt{+}$1$};
            \node at (-.3, .25) {\texttt{+}$2$};
            \node at (-.3, -.25) {\texttt{+}$3$};
            \node at (-.3, -.75) {\texttt{+}$4$};
            \draw[->] (3.5, 1.25) -- (3.1, 1.25);
            \node[rotate = 270] at (3.7, 1.25) {$r_{rect}$};
            
            \node at (1.5, 2.5) {\textbf{Rectangle}};
        \end{tikzpicture}
    \end{subfigure}
    
    \caption[Illustration of object memory layout]{Illustration of prefixing in the memory layout of 3 \textsc{Roopl} objects, originally from~\cite{th:roopl}}
    \label{fig:roopl-object-layout}
\end{figure}
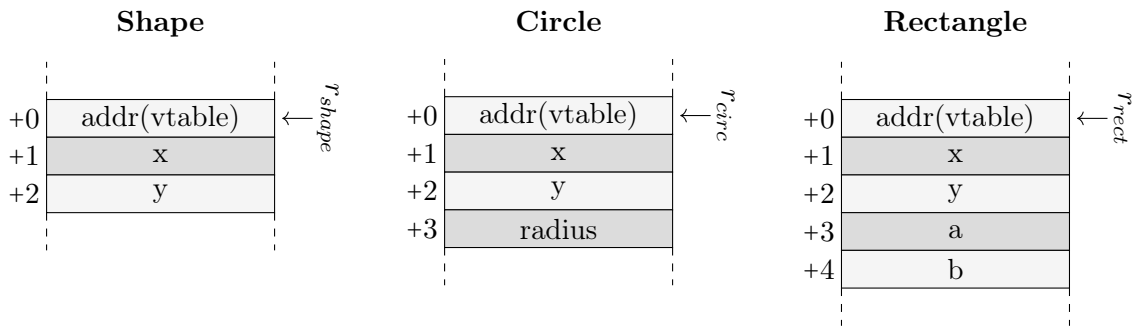

The object model is simple and only features one additional word for storing the address of the virtual table for the object class. Figure~\ref{fig:roopl-object-layout} shows the prefixing for three simple classes modeling geometric shapes.

\section{\rooplpp Memory Layout}
\label{sec:rooplpp-memory-layout}
\rooplpp builds upon the memory layout of its predecessor  with dynamic memory management. The reversible Buddy Memory heap layout presented in section~\ref{subsec:buddy-memory} is utilized in \rooplpp as it is an interesting layout, addressing a number of disadvantages found in other considered layouts, naturally translates into a reversible setting with one simple restriction (i.e only blocks which are heads of their respectable free lists are allocatable) and since its only drawback is dismissible in most real world scenarios.

\begin{figure}[ht]
    \centering
    \begin{tikzpicture}
        \fill[fill = grey] (0, 0) rectangle (2.5, 3) node[midway] {Static Data};
        \fill[fill = darkgrey] (2.5, 0) rectangle (5, 3) node[midway] {Program};
        \fill[fill = grey] (5, 0) rectangle (7, 3) node[midway] {Free lists};
        \fill[fill = grey] (7, 0) rectangle (8.5, 3) node[midway] {Heap};
        \fill[fill = grey] (12.5, 0) rectangle (14, 3) node[midway] {Stack}; 
        
        \draw (0, 0) -- (15, 0);
        \draw (0, 3) -- (15, 3);
        \draw (0, 0) -- (0, 3);
        \draw (15, 0) -- (15, 3);
        \draw (2.5, 0) -- (2.5, 3);
        \draw (5, 0) -- (5, 3); 
        \draw (7, 0) -- (7, 3);
        \draw (14, 0) -- (14, 3); 
        \draw[dashed] (8.5, 0) -- (8.5, 3); 
        \draw[dashed] (12.5, 0) -- (12.5, 3);   
    
        \node at (10.5, 2.5) {(unused memory)};
        \node at (0, -.3) {$0$};
        \node at (5, -.3) {$flp$};
        \node at (7, -.3) {$hp$};
        \node at (14, -.3) {$p$};
        \node at (12.5, -.3) {$sp$};
        \node at (15, -.3) {$2^{31} - 1$};
        \node at (7.5, -1) {$\longleftarrow$ Address Space $\longrightarrow$};

        \draw[arrow, dashed] (8.5, 1.5) -- (10.1, 1.5);
        \node at (9.4, 1.2) {\scriptsize{\textit{heap grows}}};

        \draw[arrow, dashed] (12.5, 1.5) -- (10.9, 1.5); 
        \node at (11.6, 1.2) {\scriptsize{\textit{stack grows}}};
    \end{tikzpicture}
    \caption{Memory layout of a \rooplpp program}
    \label{fig:memory-layout}
\end{figure}
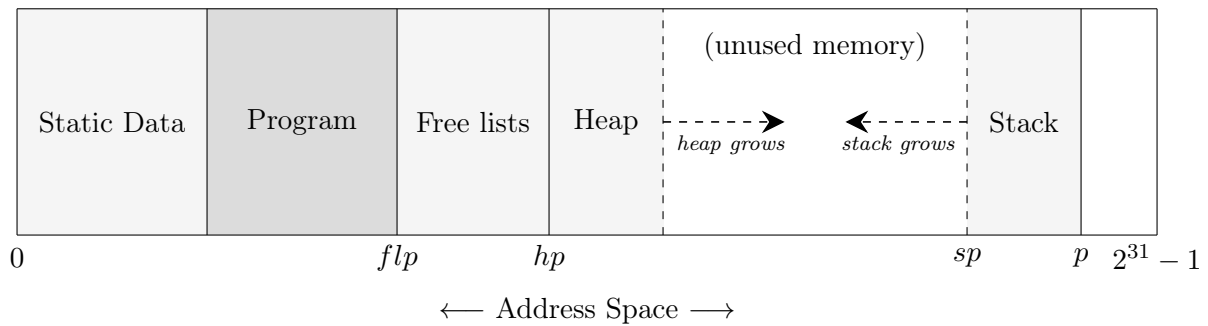

Figure~\ref{fig:memory-layout} shows the full layout of a \rooplpp program stored in memory.

\begin{itemize}
    \item As with \textsc{Roopl}, the static storage segment contains load-time labelled \inst{data} pseudo-instructions, initialized with virtual function tables and other static data needed by the translated program.

    \item The program segment is stored right after the static storage and contains the translated \rooplpp program instructions.

    \item The free lists maintained by the Buddy Memory heap layout is placed right after the program segment, with the \textit{free list pointer} $flp$ pointing at the first free list. The free lists are simply the address pointing to the first block of its respective size. The free lists are stored such that the free list at address $flp + i$ corresponds to the free list of size $2^{i+1}$.   

    \item The heap begins directly following the free lists. Its beginning is marked by the \textit{heap pointer} $(hp)$. 

    \item Unlike in \textsc{Roopl}, where the stack grows upwards, the \rooplpp stack grows downwards and begins at address $p$. The stack remains a LIFO structure, analogously to \textsc{Roopl}.
\end{itemize}

As mentioned in the previous chapter, we assume an underlying reversible operating system providing us with additional memory when needed. With no real way of simulating this, the \rooplpp compiler places the stack at a fixed address $p$ and sets one free block in the largest $2^n$ free list initially. The number of free lists and the address $p$ is configurable in the source code, but defaults to $10$ free lists, meaning initially one block of size $1024$ is available and the stack is placed at address $1024$ words after the heap.

In traditional compilers, the heap pointer usually points to the end of the heap. For reasons stated above, we never grow the heap as we start with a heap of fixed size. As such, the heap pointer simply points to the beginning of the heap.

The heap can simply be expanded by adding another block of the largest possible size and storing the address of the respective free list.

In the following sections of this chapter, we will present various translation techniques. In these translations, we will make use of a number of \textsc{Pisa} pseudo-instructions to subtract integer values from registers and pushing/popping to the program stack. The pseudo-instructions are shown in figure~\ref{fig:pseudo-instructions} and are modified from~\cite{th:roopl}, as the direction of the program stack is flipped in \rooplpp.

\begin{figure}[ht]
    \centering
    \begin{alignat*}{2}
        \inst{subi}\quad r\quad i \qquad &\overset{\textbf{def}}{\scalebox{1.8}{=}} \qquad&&\ \inst{addi}\quad r\quad -i\\[1.6ex]
        \inst{push}\quad r \qquad &\overset{\textbf{def}}{\scalebox{1.8}{=}} \qquad&&\Big[\inst{exch}\quad r\quad r_{sp}\ , \quad \inst{subi}\quad r_{sp}\quad 1\Big]\\[1.6ex]
        \inst{pop}\quad r \qquad &\overset{\textbf{def}}{\scalebox{1.8}{=}} \qquad&&\Big[\inst{addi}\quad r_{sp}\quad 1\ , \quad \inst{exch}\quad r\quad r_{sp}\Big]
    \end{alignat*}
    \caption{Definition of pseudoinstructions \inst{subi}, \inst{push} and \inst{pop}, modified from~\cite{th:roopl}}
    \label{fig:pseudo-instructions}
\end{figure}

\section{Inherited \textsc{Roopl} features}
\label{sec:inherited-features}
As mentioned, a number of features from \textsc{Roopl} carries over to \rooplpp.

The dynamic dispatching mechanism presented in~\cite{th:roopl} is inherited. As such, the invocation of a method implementation is based on the type of the object at run time. Virtual function tables are still the implementation strategy used in the dynamic dispatching implementation.

Evaluation of expressions and control flow remains unchanged. 

For completeness, object blocks are included and still stack allocated as their life time is limited to the scope of their block and the dynamic allocation process is quite expensive in terms of register pressure and number of instructions compared to the stack allocated method implemented in the \textsc{Roopl} compiler.

\section{Program Structure}
\label{sec:program-structure}
The program structure of a translated \rooplpp is analogous to the program structure of a \textsc{Roopl} program with the addition of free lists and heap initialization. The full structure is shown in figure~\ref{fig:pisa-program-layout}. 

\begin{figure}[ht]
    \centering
    \resizebox{.8\linewidth}{!}{
        \begin{minipage}{\linewidth}
            \begin{alignat*}{6}
                &\textbf{(1)}\quad&& &&\cdots\cdots && && &&\text{; Static data declarations}\\
                &\textbf{(2)}\quad&& &&\cdots\cdots && && &&\text{; Code for program class methods}\\
                &\textbf{(3)}\quad&&start\ \texttt{:}\quad&&\inst{start}\quad&& && &&\text{; Program starting point}\\
                &\textbf{(4)}\quad&& &&\inst{addi}\quad &&r_{flps}\quad &&p&&\text{; Initialize free lists pointer}\\
                &\textbf{(5)}\quad&& &&\inst{xor}\quad &&r_{hp}\quad &&r_{flps}\qquad &&\text{; Initialize heap pointer}\\
                &\textbf{(6)}\quad&& &&\inst{addi}\quad &&r_{hp}\quad &&size_{fls}&&\text{; Initialize heap pointer}\\
                &\textbf{(7)}\quad&& &&\inst{xor}\quad &&r_{b}\quad &&r_{hp}\qquad &&\text{; Store address of initial free memory block in $r_b$}\\
                &\textbf{(8)}\quad&& &&\inst{ADDI}\quad &&r_{flps}\quad &&size_{fls}\quad &&\text{; Index to end of free lists}\\
                &\textbf{(9)}\quad&& &&\inst{SUBI}\quad &&r_{flps}\quad && 1\quad &&\text{; Index to last element of free lists}\\
                &\textbf{(10)}\quad&& &&\inst{EXCH}\quad &&rb\quad &&r_{flps}\quad &&\text{; Store address of first block in last element of free lists}\\
                &\textbf{(11)}\quad&& &&\inst{ADDI}\quad &&r_{flps}\quad && 1\quad &&\text{; Index to end of free lists}\\
                &\textbf{(12)}\quad&& &&\inst{SUBI}\quad &&r_{flps}\quad &&s\quad &&\text{; Index to beginning of free lists}   \\
                &\textbf{(13)}\quad&& &&\inst{xor}\quad &&r_{sp}\quad &&r_{hp} &&\text{; Initialize stack pointer}\\
                &\textbf{(14)}\quad&& &&\inst{addi}\quad &&r_{sp}\quad &&offset_{stack}\ &&\text{; Initialize stack pointer}\\
                &\textbf{(15)}\quad&& &&\inst{subi}\quad &&r_{sp}\quad &&size_m\qquad &&\text{; Allocate space for main object}\\
                &\textbf{(16)}\quad&& &&\inst{xor}\quad &&r_m\quad &&r_{sp}\qquad &&\text{; Store address of main object in $r_m$}\\
                &\textbf{(17)}\quad&& &&\inst{xori}\quad &&r_v\quad &&label_{vt}\qquad &&\text{; Store address of vtable in $r_v$}\\
                &\textbf{(18)}\quad&& &&\inst{push}\quad &&r_v\quad && &&\text{; Push address of vtable onto stack}\\
                &\textbf{(19)}\quad&& &&\inst{push}\quad &&r_m\quad && &&\text{; Push '\textit{this}' onto stack}\\
                &\textbf{(20)}\quad&& &&\inst{bra}\quad &&label_m \span\omit\span \qquad&&\text{; Call main procedure}\\
                &\textbf{(21)}\quad&& &&\inst{pop}\quad &&r_m\quad && &&\text{; Pop '\textit{this}' from stack}\\
                &\textbf{(22)}\quad&& &&\inst{pop}\quad &&r_v\quad && &&\text{; Pop vtable address into $r_v$}\\
                &\textbf{(23)}\quad&& &&\inst{xori}\quad &&r_v\quad &&label_{vt}\qquad &&\text{; Clear $r_v$}\\
                &\textbf{(24)}\quad&& &&\inst{xor}\quad &&r_m\quad &&r_{sp}\qquad &&\text{; Clear $r_m$}\\
                &\textbf{(25)}\quad&& &&\inst{addi}\quad &&r_{sp}\quad &&size_m\qquad &&\text{; Deallocate space of main object}\\
                &\textbf{(26)}\quad&& &&\inst{subi}\quad &&r_{sp}\quad &&offset_{stack} &&\text{; Clear stack pointer}\\
                &\textbf{(27)}\quad&& &&\inst{xor}\quad &&r_{sp}\quad &&r_{hp} &&\text{; Clear stack pointer}\\
                % &\textbf{(28)}\quad&& &&\inst{subi}\quad &&r_{hp}\quad &&size_{fls} &&\text{; Clear heap pointer}\\
                % &\textbf{(29)}\quad&& &&\inst{xor}\quad &&r_{hp}\quad &&r_{flsp} &&\text{; Clear heap pointer}\\
                % &\textbf{(30)}\quad&& &&\inst{subi}\quad &&r_{flps}\quad &&p &&\text{; Clear free lists pointer}\\
                &\textbf{(28)}\quad&&finish\ \texttt{:}\quad&&\inst{finish}\quad && && &&\text{; Program exit point}
            \end{alignat*}
        \end{minipage}
    }
    \caption{Overall layout of a translated \rooplpp program}
    \label{fig:pisa-program-layout}
\end{figure}

The following \textsc{Pisa} code block initializes the free lists pointer, the heap pointer, the stack pointer, allocates the main object on the stack, calls the main method, deallocates the main object and finally clears the free lists, heap and stack pointers.

The free lists pointer is initialized by adding the base address, which varies with the size of the translated program, to the register $r_{flps}$. In figure~\ref{fig:pisa-program-layout} the base address is denoted by $p$.

The heap pointer is initialized directly after the free lists pointer by adding the size of the free lists. One free list is the size of one word and the full size of the free lists is configured in the source code (defaulted to 10, as described earlier).

Once the heap pointer and free lists pointer is initialized, the initial block of free memory is placed in the largest free lists by indexing to said list, by adding the length of the list of free lists, subtracting 1, writing the address of the first block (which is the same address as the heap pointer, which points to the beginning of the heap) to the last free list and then resetting the free lists pointer to point to the first list again, afterwards.

The stack pointer is initialized simply by adding the stack offset to the heap pointer register $r_{hp}$. The stack offset is configured in the source code and defaults to $1024$, as described earlier in this chapter. As such, the heap and the stack each have $1024$ words of space to utilize. Once the stack pointer has been initialized, the main object is allocated on the stack and the main method called, analogously to the \textsc{Roopl} program structure.

When the program terminates and the main method returns, the main object is popped from the stack and deallocated and the stack pointer is cleared. The heap and free list pointer not intentionally not cleared to simulate future program simulation using these pointers. The contents of the free lists and whatever is left on the heap is untouched at this point. It is the programmers responsibility to free dynamically allocated objects in their \rooplpp program. Furthermore, depending on the deallocation order, we might not end up with exactly one fully merged block in the end and as such, we do not invert the steps taken to initialize this initial free memory block.
Analogously to \textsc{Roopl}, the values of the main object are left in the stack section of memory.

\section{Buddy Memory Translation}
\label{sec:buddy-memory-translation}
As briefly mentioned in section~\ref{sec:rooplpp-memory-layout}, the Buddy Memory layout was selected as the memory manager layout as it addressed a number of problems related to fragmentation and initialization. The Buddy Memory layout could be converted to a reversible section with only a few restrictions and side effects, which will be described in this section. Firstly, we present the algorithm translated to \textsc{Pisa}. As the algorithm is quite lengthy, it will be broken down into smaller chunks. The full translation is shown in appendix~\ref{app:pisa-translated-buddy-memory}.

The Buddy Memory algorithm consists of three \textsc{Janus} procedures; the entry point \textbf{malloc}, the recursion body \textbf{malloc1} and a helper function \textbf{double}. The entry point is omitted for now, as it differs depending on which type of memory object we are allocating and will be presented in sections~\ref{sec:object-allocation-deallocation} and~\ref{subsec:construction-destruction}. The helper function can be implemented using a single instruction in \textsc{Pisa} for our specific case of doubling number in the power-of-two, which we will show later. 

\begin{figure}[ht]
    \centering
    \resizebox{.8\linewidth}{!}{
        \begin{minipage}{\linewidth}
            \begin{alignat*}{7}
                &\textbf{(1)}\quad&&malloc1_{top}\ \texttt{:}\quad  &&\inst{bra}\quad &&malloc1_{bot} \span\omit\span\quad \span\omit\span\quad &&\text{; Receive jump}\\ 
                &\textbf{(2)}\quad&& &&\inst{pop}\quad&&r_{ro}&& && &&\text{; Pop return offset from the stack}\\
                &\textbf{(3)}\quad&& &&\cdots\cdots && && && &&\text{; Inverse of \textbf{(7)}}\\
                &\textbf{(4)}\quad&&malloc1_{entry}\ \texttt{:}\quad&&\inst{swapbr}\quad &&r_{ro} && && &&\text{; Malloc1 entry and exit point}\\
                &\textbf{(5)}\quad&& &&\inst{neg}\quad &&r_{ro} && && &&\text{; Negate return offset}\\        
                &\textbf{(6)}\quad&& &&\inst{push}\quad &&r_{ro} && && &&\text{; Store return offset on stack}\\
                &\textbf{(7-63)}\quad&& &&\cdots\cdots && && && &&\text{; Allocation code}\\ 
                &\textbf{(64)}\quad&&malloc1_{bot}\ \texttt{:}\quad  &&\inst{bra}\quad &&malloc1_{top} \span\omit\span\quad \span\omit\span\quad &&\text{; Jump}\\
            \end{alignat*}
        \end{minipage}
    }    
    \caption{Dynamic dispatch approach for entering the allocation subroutine}
    \label{fig:buddy-allocation-entry}
\end{figure}

Before we go into depth with the translation of the algorithm, we consider the mechanism for triggering the allocation subroutine. Naively, we could generate the entire block of code required for allocation for every \textbf{new} or \textbf{delete} statement in the target program. This approach would severely limit the amount of objects we could allocate as the register pressure of the Buddy Memory implementation is quite high, as we be shown in this section. Instead, we can utilize the  dynamic dispatching technique, which also is used for method invocations. This way, we only generate the allocation instructions once, and then simply jump to the entry point from different locations in the program. Figure~\ref{fig:buddy-allocation-entry} outlines the structure for this approach. By using the \inst{swapbr} instruction we can jump from multiple points of origin in the compiled program and internally for the recursive needs of the algorithm itself.

\begin{figure}[ht]
    \centering
    \begin{subfigure}{.3\textwidth}
        \resizebox{\linewidth}{!}{
\begin{lstlisting}[language=janus, style=basic, frame=none]
if (csize < osize) then
    // outer if-then
else
    if freelists[counter] != 0 then
        // inner if-then
    else
        // inner else
    fi (freelists[counter] = 0 || 
        p - csize != freelists[counter])
fi csize < osize 
\end{lstlisting} 
        }
    \end{subfigure}%
    \begin{subfigure}{.7\textwidth}
        \centering
        \resizebox{0.9\linewidth}{!}{
        \begin{minipage}{\linewidth}
            \begin{alignat*}{7}
                &\textbf{(7)}\quad&& &&\cdots\cdots && && && &&\text{; Code for $r_{fl}\ \leftarrow\ addr(fl[c])$}\\
                &\textbf{(8)}\quad&& &&\cdots\cdots && && && &&\text{; Code for $r_{block}\ \leftarrow\ \llbracket fl[c] \rrbracket$}\\
                &\textbf{(9)}\quad&& &&\cdots\cdots && && && &&\text{; Code for $r_{e1_o}\ \leftarrow\ \llbracket c_{size} < object_{size} \rrbracket$}\\
                &\textbf{(10)}\quad&& &&\inst{xor}\quad &&r_t && r_{e1_o} && &&\text{; Copy value of $c_{size} < object_{size}$ into $r_t$}\\        
                &\textbf{(11)}\quad&& &&\cdots\cdots && && && &&\text{; Inverse of \textbf{(9)}}\\ 
                &\textbf{(12)}\quad&&o_{test}\ \texttt{:}\quad &&\inst{beq} &&r_t && r_0 && o_{test_f} && \text{; Receive jump}\\
                &\textbf{(13)}\quad&& &&\inst{xori} &&r_t && 1 && && \text{; Clear $r_t$}\\
                &\textbf{(14-21)}\quad&& &&\cdots\cdots && && && &&\text{; Code for \textbf{outer if-then} statement}\\
                &\textbf{(22)}\quad&& &&\inst{xori} &&r_t && 1 && && \text{; Set $r_t = 1$}\\
                &\textbf{(23)}\quad&&o_{assert_t}\ \texttt{:}\quad &&\inst{bra} &&o_{assert} \span\omit\span\quad \span\omit\span\quad && \text{; Jump}\\
                &\textbf{(24)}\quad&&o_{test_f}\ \texttt{:}\quad &&\inst{bra} &&o_{test} \span\omit\span\quad \span\omit\span\quad && \text{; Receive jump}\\
                &\textbf{(25)}\quad&& &&\cdots\cdots && && && &&\text{; Code for $r_{e1_i}\ \leftarrow\ \llbracket addr(fl[c]) \neq 0 \rrbracket$}\\
                &\textbf{(26)}\quad&& &&\inst{xor}\quad &&r_{t2} && r_{e1_i} && &&\text{; Copy value of $r_{e1_i}$ into $r_{t2}$}\\        
                &\textbf{(27)}\quad&& &&\cdots\cdots && && && &&\text{; Inverse of \textbf{(25)}}\\
                &\textbf{(28)}\quad&&i_{test}\ \texttt{:}\quad &&\inst{beq} &&r_{t2} && r_0 && i_{test_f} && \text{; Receive jump}\\
                &\textbf{(29)}\quad&& &&\inst{xori} &&r_{t2} && 1 && && \text{; Clear $r_{t2}$}\\
                &\textbf{(30-34)}\quad&& &&\cdots\cdots && && && &&\text{; Code for \textbf{inner if-then} statement}\\
                &\textbf{(35)}\quad&& &&\inst{xori} &&r_{t2} && 1 && && \text{; Set $r_{t2} = 1$}\\
                &\textbf{(36)}\quad&&i_{assert_t}\ \texttt{:}\quad &&\inst{bra} &&i_{assert} \span\omit\span\quad \span\omit\span\quad && \text{; Jump}\\
                &\textbf{(37)}\quad&&i_{test_f}\ \texttt{:}\quad &&\inst{bra} &&i_{test} \span\omit\span\quad \span\omit\span\quad && \text{; Receive jump}\\
                &\textbf{(38-47)}\quad&& &&\cdots\cdots && && && &&\text{; Code for \textbf{inner else} statement}\\
                &\textbf{(48)}\quad&&i_{assert}\ \texttt{:}\quad &&\inst{bne} &&r_{t2} && r_0 && i_{assert_t} && \text{; Receive jump}\\
                &\textbf{(49)}\quad&& &&\inst{exch} &&r_{tmp} && r_{fl} && && \text{; Load address of head of current free list}\\
                &\textbf{(50)}\quad&& &&\inst{sub} &&r_{p} && r_{cs} && && \text{; Set p to previous block address}\\
                &\textbf{(51)}\quad&& &&\cdots\cdots && && && &&\text{; $r_{e2_{i1}}\ \leftarrow\ \llbracket p - c_{size} \neq addr(fl[c])\rrbracket$}\\
                &\textbf{(52)}\quad&& &&\cdots\cdots && && && &&\text{; $r_{e2_{i2}}\ \leftarrow\ \llbracket addr(fl[c]) = 0 \rrbracket$}\\
                &\textbf{(53)}\quad&& &&\cdots\cdots && && && &&\text{; $r_{e2_{i3}}\ \leftarrow\ \llbracket (p - c_{size} \neq addr(fl[c])) \vee (addr(fl[c]) = 0) \rrbracket$}\\
                &\textbf{(54)}\quad&& &&\inst{xor} &&r_{r2} && r_{e2_{i3}} && && \text{; Copy value of $r_{e2_{i3}}$ into $r_{t2}$}\\
                &\textbf{(55)}\quad&& &&\cdots\cdots && && && &&\text{; Inverse of \textbf{(53)}}\\
                &\textbf{(56)}\quad&& &&\cdots\cdots && && && &&\text{; Inverse of \textbf{(52)}}\\
                &\textbf{(57)}\quad&& &&\cdots\cdots && && && &&\text{; Inverse of \textbf{(51)}}\\
                &\textbf{(58)}\quad&& &&\inst{add} &&r_{p} && r_{cs} && && \text{; Inverse of \textbf{(50)}}\\
                &\textbf{(59)}\quad&& &&\inst{exch} &&r_{tmp} && r_{fl} && && \text{; Inverse of \textbf{(49)}}\\
                &\textbf{(60)}\quad&&o_{assert}\ \texttt{:}\quad &&\inst{bne} &&r_{t} && r_0 && o_{assert_t} && \text{; Receive jump}\\
                &\textbf{(61)}\quad&& &&\cdots\cdots && && && &&\text{; Code for $r_{e2_o}\ \leftarrow\ \llbracket c_{size} < object_{size} \rrbracket$}\\
                &\textbf{(62)}\quad&& &&\inst{xor}\quad &&r_t && r_{e2_o} && &&\text{; Copy value of $c_{size} < object_{size}$ into $r_t$}\\        
                &\textbf{(63)}\quad&& &&\cdots\cdots && && && &&\text{; Inverse of \textbf{(61)}}\\ 
            \end{alignat*}
        \end{minipage}
    }
    \end{subfigure}
    \caption{\textsc{Pisa} translation of the nested conditionals in the Buddy Memory algorithm}
    \label{fig:pisa-buddy-conditionals}
\end{figure}

The main recursion body of the algorithm, \textbf{malloc1} from listing~\ref{lst:buddy-memory}, page~\pageref{lst:buddy-memory} consists of two conditionals, in which one is nested in the else branch of the outer conditional. Figure~\ref{fig:pisa-buddy-conditionals} shows the translation structure of the nested conditional pair, using the translation techniques for conditionals presented in~\cite{ha:translation}.

The nested conditionals contain large amounts of boilerplate code for evaluating the various expressions of the conditionals. As these conditionals requires comparisons with contents of the free lists, we must be careful with extracting and storing the values in the free list. 

We have three statements to translate from here. The outer \textbf{if-then} statement, the inner \textbf{if-then} statement and the inner \textbf{else} statement. 

\begin{figure}[ht!]
    \centering
    \begin{subfigure}{.4\textwidth}
\begin{lstlisting}[language=janus, style=basic, frame=none]
counter += 1
call double(csize)
call malloc1(p, osize, freelists, 
             counter, csize) 
uncall double(csize)
counter -= 1
\end{lstlisting}  
    \end{subfigure}
    \begin{subfigure}{.4\textwidth}
        \centering
        \resizebox{0.99\linewidth}{!}{
        \begin{minipage}{\linewidth}
            \begin{alignat*}{7}
                &\textbf{(14)}\quad&& &&\inst{addi} &&r_{c} && 1 && && \text{; $Counter\texttt{++}$}\\
                &\textbf{(15)}\quad&& &&\inst{rl} &&r_{sc}\ && 1 && && \text{; Call $double(c_{size}$)}\\
                &\textbf{(16)}\quad&& &&\cdots\cdots && && && &&\text{; Inverse of \textbf{(7)}}\\
                &\textbf{(17)}\quad&& &&\cdots\cdots && && && &&\text{; Code for pushing temp reg values to stack}\\
                &\textbf{(18)}\quad&& &&\inst{bra}\quad &&malloc1_{entry} \span\omit\span\quad \span\omit\span\quad && \text{; Call $malloc1()$)}\\
                &\textbf{(19)}\quad&& &&\cdots\cdots && && && &&\text{; Inverse of \textbf{(17)}}\\
                &\textbf{(20)}\quad&& &&\inst{rr} &&r_{sc}\ && 1 && && \text{; Inverse of \textbf{(15)}}\\
                &\textbf{(21)}\quad&& &&\inst{subi} &&r_{c} && 1 && && \text{; Inverse of \textbf{(14)}}\\
            \end{alignat*}
        \end{minipage}
        }
    \end{subfigure}
    \caption{\textsc{Pisa} translation of the outer \textbf{if-then} statement for the Buddy Memory algorithm}
    \label{fig:pisa-buddy-outer-if}
\end{figure}

Figure~\ref{fig:pisa-buddy-outer-if} shows the translation of the outer \textbf{if-then} statement. As briefly mentioned, we can utilize the right bit shift instruction of \textsc{Pisa}, \inst{RL}, in place of the \textbf{double} helper procedure from the \textsc{Janus} implementation. By using a simple bit shift, we are able to maintain reversibility elegantly when doubling or halving numbers in the power-of-two. This statement also contains one of the careful storage operations of the free list values, in instruction \textbf{(16)}. Before we recursively branch to the entry point, we must place the previously extracted address of the head of the free list back into the free list. This is also the reason for instruction \textbf{(3)} in figure~\ref{fig:buddy-allocation-entry}. Furthermore, we must push all temporary evaluated expression values to the stack, so they can be popped when we return.   

\begin{figure}[ht]
    \centering
    \begin{subfigure}{.4\textwidth}
\begin{lstlisting}[language=janus, style=basic, frame=none]
p += freelists[counter]
freelists[counter] -= p

// Swap head of free list 
// with p's next block
freelists[counter] ^= M(p)
M(p) ^= freelists[counter]
freelists[counter] ^= M(p)
\end{lstlisting}  
    \end{subfigure}
    \begin{subfigure}{.5\textwidth}
        \centering
        \resizebox{\linewidth}{!}{
        \begin{minipage}{\linewidth}
            \begin{alignat*}{7}
                &\textbf{(30)}\quad&& &&\inst{add}\quad &&r_{p} && r_{block} && && \text{; Copy address of the current block to p}\\
                &\textbf{(31)}\quad&& &&\inst{sub}\quad &&r_{block}\ && r_{p} && && \text{; Clear $r_{block}$}\\
                &\textbf{(32)}\quad&& &&\inst{exch}\quad &&r_{tmp} && r_{p} && && \text{; Load address of next block}\\
                &\textbf{(33)}\quad&& &&\inst{exch}\quad &&r_{tmp} && r_{fl} && && \text{; Set address of next block as new head of free list}\\
                &\textbf{(34)}\quad&& &&\inst{xor}\quad &&r_{tmp} && r_{p} && && \text{; Clear address of next block}\\
            \end{alignat*}
        \end{minipage}
        }
    \end{subfigure}
    \caption{\textsc{Pisa} translation of the inner \textbf{if-then} statement for the Buddy Memory algorithm}
    \label{fig:pisa-buddy-inner-if}
\end{figure}

Figure~\ref{fig:pisa-buddy-inner-if} shows the translation of the inner \textbf{if-then} statement. This statement translates easily using the \inst{exch} instructions to swap with memory locations as simulated in the \textsc{Janus} code. 

\begin{figure}[ht]
    \centering
    \begin{subfigure}{.4\textwidth}
\begin{lstlisting}[language=janus, style=basic, frame=none]
counter += 1
call double(csize)
call malloc1(p, osize, freelists, 
             counter, csize)
uncall double(csize)
counter -= 1
freelists[counter] += p
p += csize
\end{lstlisting}  
    \end{subfigure}%
    \begin{subfigure}{.5\textwidth}
        \centering
        \resizebox{\linewidth}{!}{
        \begin{minipage}{\linewidth}
            \begin{alignat*}{7}
                &\textbf{(38)}\quad&& &&\inst{addi} &&r_{c} && 1 && && \text{; $Counter\texttt{++}$}\\
                &\textbf{(39)}\quad&& &&\inst{rl} &&r_{sc}\ && 1 && && \text{; Call $double(c_{size}$)}\\
                &\textbf{(40)}\quad&& &&\cdots\cdots && && && &&\text{; Push temp reg values to stack}\\
                &\textbf{(41)}\quad&& &&\inst{bra}\quad &&malloc1_{entry} \span\omit\span\quad \span\omit\span\quad && \text{; Call $malloc1()$)}\\
                &\textbf{(42)}\quad&& &&\cdots\cdots && && && &&\text{; Inverse of \textbf{(40)}}\\
                &\textbf{(43)}\quad&& &&\inst{rr} &&r_{sc}\ && 1 && && \text{; Inverse of \textbf{(39)}}\\
                &\textbf{(44)}\quad&& &&\inst{subi} &&r_{c} && 1 && && \text{; Inverse of \textbf{(38)}}\\
                &\textbf{(45)}\quad&& &&\inst{xor} &&r_{tmp} && r_p && && \text{; Copy current address of $p$}\\
                &\textbf{(46)}\quad&& &&\inst{exch} &&r_{tmp} && r_{fl} && && \text{; Store address of $p$ in free list}\\
                &\textbf{(47)}\quad&& &&\inst{add} &&r_{p} && r_{cs} && && \text{; Split block by $p$ = other half of block}\\
            \end{alignat*}
        \end{minipage}
        }
    \end{subfigure}
    \caption{\textsc{Pisa} translation of the inner \textbf{else} statement for the Buddy Memory algorithm}
    \label{fig:pisa-buddy-inner-else}
\end{figure}

The last statement translation is the inner \textbf{else} statement shown in figure~\ref{fig:pisa-buddy-inner-else}. This statement is almost identical to the outer \textbf{if-then} with the addition of the block splitting code. The block splitting is done in three instructions. First, the current block we are examining is set as the new head of the current free list. Afterwards the current free list block size is added to out pointer $p$, resulting in an effectively split block.

During the design of the reversible Buddy Memory algorithm limitations on the merging and splitting conditions were required to ensure reversibility. Since a split block is always added as two blocks to an empty free list, we can only merge adjacent blocks if they are the only two blocks in a free list. In the irreversible Buddy Memory algorithm block merging can occur in any place of the free list, but in the reversible version, we can only merge blocks at the start of the free list to maintain reversibility. The effect of this limitation prevents us from returning to one final block of free memory, if the deallocation order is not exactly opposite of the allocation order.

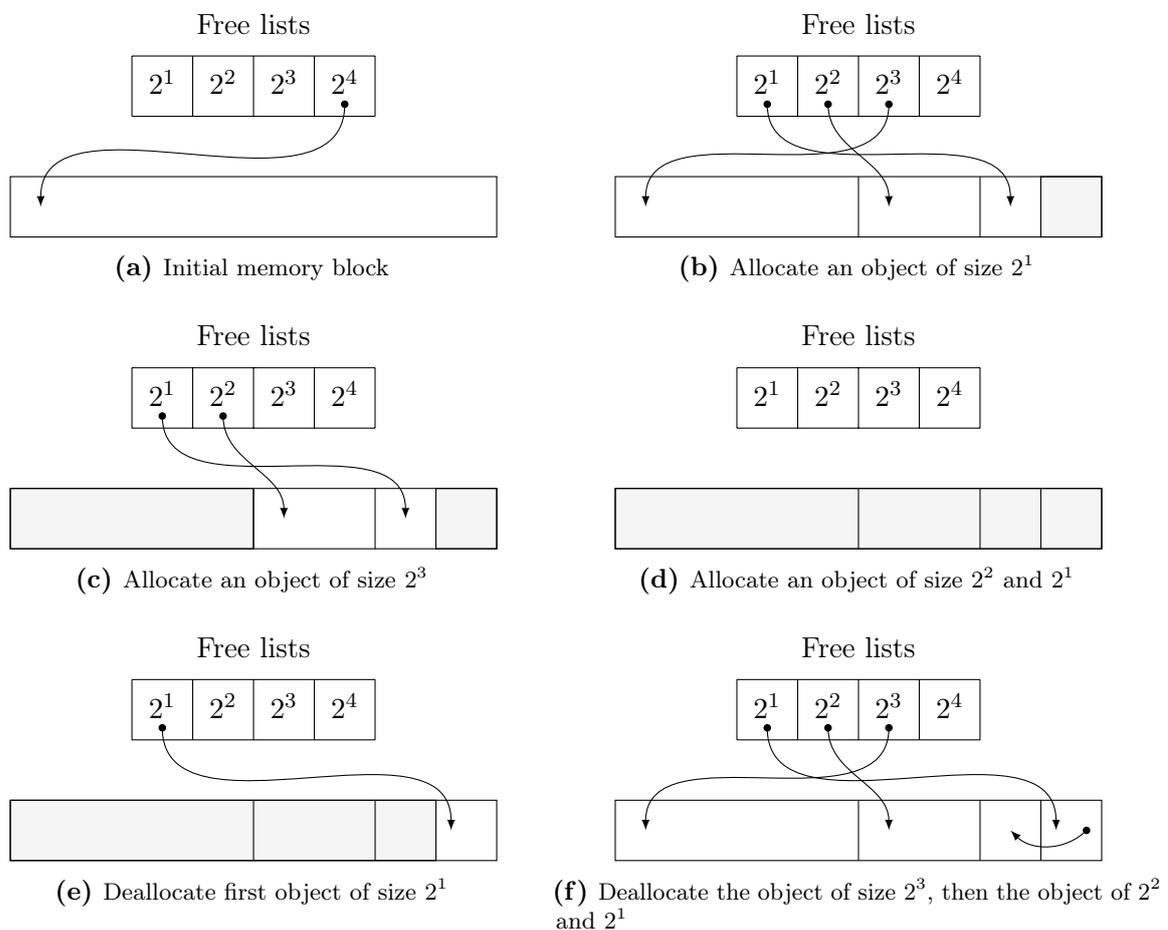
\begin{figure}[ht]
    \centering
    \begin{subfigure}[t]{.5\textwidth}
        \centering
        \begin{tikzpicture}[scale=0.8]
            % Boxes
            \draw[step=1] (2,2) grid (6,3);
            \draw (0,0) rectangle (8, 1);
            
            % Labels
            \node[above] at (4, 3.2) {Free lists}; 
            \node[above] at (2.5, 2.2) {$2^1$}; 
            \node[above] at (3.5, 2.2) {$2^2$}; 
            \node[above] at (4.5, 2.2) {$2^3$}; 
            \node[above] at (5.5, 2.2) {$2^4$}; 
            
            % Arrows
            \node[circle,fill,inner sep=1pt] at (5.5, 2.2) {};
            \draw[-latex] (5.5, 2.2) to[out=270, in=90] (0.5, 0.5);
        \end{tikzpicture}
        \caption{\footnotesize Initial memory block}
    \end{subfigure}%
    \begin{subfigure}[t]{.5\textwidth}
        \centering
        \begin{tikzpicture}[scale=0.8]
            % Fills
            \draw[fill=grey] (7,0) rectangle (8,1);
            
            % Boxes
            \draw[step=1] (2,2) grid (6,3);
            \draw (0,0) rectangle (8, 1);
            
            % Labels
            \node[above] at (4, 3.2) {Free lists}; 
            \node[above] at (2.5, 2.2) {$2^1$}; 
            \node[above] at (3.5, 2.2) {$2^2$}; 
            \node[above] at (4.5, 2.2) {$2^3$}; 
            \node[above] at (5.5, 2.2) {$2^4$}; 
            
            % Lines
            \draw (4,0) -- (4, 1);
            \draw (6,0) -- (6, 1);
            \draw (7,0) -- (7, 1);
            
            % Arrows
            \node[circle,fill,inner sep=1pt] at (4.5, 2.2) {};
            \draw[-latex] (4.5, 2.2) to[out=270, in=90] (0.5, 0.5);
            
            \node[circle,fill,inner sep=1pt] at (3.5, 2.2) {};
            \draw[-latex] (3.5, 2.2) to[out=270, in=90] (4.5, 0.5);
            
            \node[circle,fill,inner sep=1pt] at (2.5, 2.2) {};
            \draw[-latex] (2.5, 2.2) to[out=270, in=90] (6.5, 0.5);
        \end{tikzpicture}
        \caption{\footnotesize Allocate an object of size $2^1$}
    \end{subfigure}%
    \vskip 1em
    \begin{subfigure}[t]{.5\textwidth}
        \centering
        \begin{tikzpicture}[scale=0.8]
            % Fills
            \draw[fill=grey] (7,0) rectangle (8,1);
            \draw[fill=grey] (0,0) rectangle (4,1);
            
            % Boxes
            \draw[step=1] (2,2) grid (6,3);
            \draw (0,0) rectangle (8, 1);
            
            % Labels
            \node[above] at (4, 3.2) {Free lists}; 
            \node[above] at (2.5, 2.2) {$2^1$}; 
            \node[above] at (3.5, 2.2) {$2^2$}; 
            \node[above] at (4.5, 2.2) {$2^3$}; 
            \node[above] at (5.5, 2.2) {$2^4$}; 
            
            % Lines
            \draw (4,0) -- (4, 1);
            \draw (6,0) -- (6, 1);
            \draw (7,0) -- (7, 1);
            
            % Arrows
            \node[circle,fill,inner sep=1pt] at (3.5, 2.2) {};
            \draw[-latex] (3.5, 2.2) to[out=270, in=90] (4.5, 0.5);
            
            \node[circle,fill,inner sep=1pt] at (2.5, 2.2) {};
            \draw[-latex] (2.5, 2.2) to[out=270, in=90] (6.5, 0.5);
        \end{tikzpicture}
        \caption{\footnotesize Allocate an object of size $2^3$}
    \end{subfigure}%
    \begin{subfigure}[t]{.5\textwidth}
        \centering
        \begin{tikzpicture}[scale=0.8]
            % Fills
            \draw[fill=grey] (0,0) rectangle (8,1);
            % \draw[fill=grey] (0,0) rectangle (4,1);
            % \draw[fill=grey] (4,0) rectangle (6,1);
  
            % Boxes
            \draw[step=1] (2,2) grid (6,3);
            \draw (0,0) rectangle (8, 1);
  
            % Labels
            \node[above] at (4, 3.2) {Free lists}; 
            \node[above] at (2.5, 2.2) {$2^1$}; 
            \node[above] at (3.5, 2.2) {$2^2$}; 
            \node[above] at (4.5, 2.2) {$2^3$}; 
            \node[above] at (5.5, 2.2) {$2^4$}; 
  
            % Lines
            \draw (4,0) -- (4, 1);
            \draw (6,0) -- (6, 1);
            \draw (7,0) -- (7, 1);
        \end{tikzpicture}
        \caption{\footnotesize Allocate an object of size $2^2$ and $2^1$}
    \end{subfigure}%
    \vskip 1em
    \begin{subfigure}[t]{.5\textwidth}
        \centering
        \begin{tikzpicture}[scale=0.8]
            % Fills
            \draw[fill=grey] (0,0) rectangle (7,1);
  
            % Boxes
            \draw[step=1] (2,2) grid (6,3);
            \draw (0,0) rectangle (8, 1);
  
            % Labels
            \node[above] at (4, 3.2) {Free lists}; 
            \node[above] at (2.5, 2.2) {$2^1$}; 
            \node[above] at (3.5, 2.2) {$2^2$}; 
            \node[above] at (4.5, 2.2) {$2^3$}; 
            \node[above] at (5.5, 2.2) {$2^4$}; 
  
            % Lines
            \draw (4,0) -- (4, 1);
            \draw (6,0) -- (6, 1);
            \draw (7,0) -- (7, 1);
  
            % Arrows            
            \node[circle,fill,inner sep=1pt] at (2.5, 2.2) {};
            \draw[-latex] (2.5, 2.2) to[out=270, in=90] (7.25, 0.5);

            % \node[circle,fill,inner sep=1pt] at (7.75, .5) {};
            % \draw[-latex] (7.75, .5) to[out=-135, in=-45] (6.5, .5);
        \end{tikzpicture}
        \caption{\footnotesize Deallocate first object of size $2^1$}
    \end{subfigure}%
    \begin{subfigure}[t]{.5\textwidth}
        \centering
        \begin{tikzpicture}[scale=0.8]           
            % Boxes
            \draw[step=1] (2,2) grid (6,3);
            \draw (0,0) rectangle (8, 1);
  
            % Labels
            \node[above] at (4, 3.2) {Free lists}; 
            \node[above] at (2.5, 2.2) {$2^1$}; 
            \node[above] at (3.5, 2.2) {$2^2$}; 
            \node[above] at (4.5, 2.2) {$2^3$}; 
            \node[above] at (5.5, 2.2) {$2^4$}; 
  
            % Lines
            \draw (4,0) -- (4, 1);
            \draw (6,0) -- (6, 1);
            \draw (7,0) -- (7, 1);
  
            % Arrows            
            \node[circle,fill,inner sep=1pt] at (2.5, 2.2) {};
            \draw[-latex] (2.5, 2.2) to[out=270, in=90] (7.25, 0.5);

            \node[circle,fill,inner sep=1pt] at (7.75, .5) {};
            \draw[-latex] (7.75, .5) to[out=-135, in=-45] (6.5, .5);

            \node[circle,fill,inner sep=1pt] at (3.5, 2.2) {};
            \draw[-latex] (3.5, 2.2) to[out=270, in=90] (4.5, 0.5);

            \node[circle,fill,inner sep=1pt] at (4.5, 2.2) {};
            \draw[-latex] (4.5, 2.2) to[out=270, in=90] (0.5, 0.5);
        \end{tikzpicture}
        \caption{\footnotesize Deallocate the object of size $2^3$, then the object of $2^2$ and $2^1$}
    \end{subfigure}%
    \caption{Non-opposite deallocation results in a different free list after termination}
    \label{fig:deallocation-order-free-list}
\end{figure}

Figure~\ref{fig:deallocation-order-free-list} shows how alternative deallocation orders results in different free lists, compared to the original given to some function. However, as discussed in section~\ref{sec:memory-garbage}, we can consider every collection of Buddy Memory free lists equivalent, as a later computation can take another set of free lists and still execute its function, as long as the free lists have the required blocks available.

\newpage
\section{Object Allocation and Deallocation}
\label{sec:object-allocation-deallocation}
Now that we have the main allocation mechanism in place and a method of accessing it through a label and a \inst{swapbr} instruction, we can continue translating the \textbf{malloc} procedure entry point from listing~\ref{lst:buddy-memory} on page~\pageref{lst:buddy-memory}.

\begin{figure}[ht!]
    \centering
    \begin{subfigure}{.4\textwidth}
\begin{lstlisting}[language=janus, style=basic, frame=none]
procedure malloc(int p, int osize, 
                 int freelists[])
    local int counter = 0
    local int csize = 2
    call malloc1(p, osize, freelists, 
                 counter, csize)
    delocal int csize = 2
    delocal int counter = 0
\end{lstlisting}  
    \end{subfigure}%
    \begin{subfigure}{.5\textwidth}
        \centering
        \resizebox{\linewidth}{!}{
        \begin{minipage}{\linewidth}
            \begin{alignat*}{6}
                &\textbf{(1)}\quad&&l_{malloc\_top}\ \texttt{:}\quad &&\inst{bra}\quad && l_{malloc\_bot}\quad && \quad&& \text{; Receive jump}\\
                &\textbf{(2)}\quad&&l_{malloc}\ \texttt{:}\quad &&\inst{swapbr}\quad &&r_o && && \text{; Entry and exit point}\\
                &\textbf{(3)}\quad&& &&\inst{neg} &&r_o && && \text{; Negate return offset}\\
                &\textbf{(4)}\quad&& &&\inst{addi} &&c_{size} && 2\quad && \text{; Init $c_{size}$}\\
                &\textbf{(5)}\quad&& &&\inst{xor} &&r_{counter} && r_0 && \text{; Init counter}\\
                &\textbf{(6)}\quad&& &&\cdots && && && \text{; Pop $r_p$ and $object_{size}$ from stack}\\
                &\textbf{(7)}\quad&& &&\inst{push} &&r_{0} && && \text{; Push $r_o$}\\
                &\textbf{(8)}\quad&& &&\inst{bra} &&l_{malloc1} && && \text{; call \textbf{malloc1()}}\\
                &\textbf{(9)}\quad&& &&\inst{pop} &&r_{0} && && \text{; Inverse of \textbf{(7)}}\\
                &\textbf{(10)}\quad&& &&\cdots &&r_{0} && && \text{; Inverse of \textbf{(6)}}\\
                &\textbf{(11)}\quad&& &&\inst{xor} &&r_{counter} && r_0 && \text{; Inverse of \textbf{(5)}}\\
                &\textbf{(12)}\quad&& &&\inst{subi} &&c_{size} && 2 && \text{; Inverse of \textbf{(4)}}\\
                &\textbf{(13)}\quad&&l_{malloc\_bot}\ \texttt{:}\quad &&\inst{bra} &&l_{malloc\_top} && && \text{; Jump}\\
            \end{alignat*}
        \end{minipage}
        }
    \end{subfigure}
    \caption{\textsc{Pisa} translation of the \textbf{malloc} procedure entry point of Buddy Memory algorithm}
    \label{fig:pisa-buddy-malloc-entry}
\end{figure}

Figure~\ref{fig:pisa-buddy-malloc-entry} shows the translated \textbf{malloc} procedure. In addition to the original procedure, we also push the current return offset register value to the stack before we branch to the \textbf{malloc1} implementation, to ensure we have a zero-cleared register before starting the allocation process. The translated procedure assumes that the pointer to the object we are allocating and its size are on top of the stack before entering the block. This translated procedure serves as the entry point for the allocation subroutine as it is also only generated once. Each \textbf{new} and \textbf{delete} statement branches to the $l_{malloc}$ label to begin an allocation or a deallocation.  

\begin{figure}[H]
    \centering
    \begin{subfigure}[t]{0.495\linewidth}
        \vskip 0pt
        \centering
        \begin{equation*} 
            \textbf{new}\ c\ x
        \end{equation*}
        \resizebox{.9\linewidth}{!}{
            \begin{minipage}{\linewidth}
                \begin{alignat*}{5}
                    &\textbf{(1)}\quad&&\cdots \quad &&\quad && \quad &&\text{; Push registers }\\
                    &\textbf{(2)}\quad&&\cdots\quad &&\quad && &&\text{; Code for $r_t \leftarrow x_{size}$}\\
                    &\textbf{(3)}\quad&&\inst{push}\quad && r_t \quad && &&\text{; Push $r_t$}\\
                    &\textbf{(4)}\quad&&\inst{push}\quad && r_p \quad && &&\text{; Push $r_p$}\\
                    &\textbf{(5)}\quad&&\inst{bra}\quad && l_{malloc} \quad && &&\text{; Allocate}\\
                    &\textbf{(6)}\quad&&\inst{pop}\quad && r_p \quad && &&\text{; Inverse of \textbf{(4)}}\\
                    &\textbf{(7)}\quad&&\inst{pop}\quad && r_t \quad && &&\text{; Inverse of \textbf{(3)}}\\
                    &\textbf{(8)}\quad&&\cdots\quad && \quad && &&\text{; Inverse of \textbf{(2)}}\\
                    &\textbf{(9)}\quad&&\cdots\quad && \quad && &&\text{; Inverse of \textbf{(1)}}\\
                    &\textbf{(10)}\quad&&\cdots\quad &&\quad && &&\text{; Code for $r_v \leftarrow \llbracket addr(x) \rrbracket$}\\
                    &\textbf{(11)}\quad&&\inst{xori}\quad && r_t\quad && label_{vt} &&\text{; Store address of vtable in $r_t$}\\
                    &\textbf{(12)}\quad&&\inst{exch}\quad && r_t\quad && r_p\quad &&\text{; Store vtable in new object}\\
                    &\textbf{(13)}\quad&&\inst{addi}\quad && r_p\quad && offset_{ref}\quad &&\text{; Index to ref count pos}\\
                    &\textbf{(14)}\quad&&\inst{xori}\quad && r_t\quad && 1 &&\text{; Init ref count}\\
                    &\textbf{(15)}\quad&&\inst{exch}\quad && r_t\quad && r_p\quad &&\text{; Store ref count}\\
                    &\textbf{(16)}\quad&&\inst{subi}\quad && r_p\quad && offset_{ref}\quad &&\text{; Inverse of \textbf{(13)}}\\
                    &\textbf{(17)}\quad&&\inst{exch}\quad && r_p\quad && r_v\quad &&\text{; Store address in variable}\\
                    &\textbf{(18)}\quad&&\cdots\quad &&\quad && &&\text{; Inverse of \textbf{(10)}}\\
                \end{alignat*}
            \end{minipage}
        }
    \end{subfigure}
    \begin{subfigure}[t]{0.495\linewidth}
        \vskip 0pt
        \centering
        \begin{equation*}
            \textbf{delete}\ c\ x
        \end{equation*}
        \resizebox{.9\linewidth}{!}{
            \begin{minipage}{\linewidth}
                \begin{alignat*}{5}
                    &\textbf{(1)}\quad&&\cdots\quad &&\quad && &&\text{; Code for $r_p \leftarrow \llbracket addr(x) \rrbracket$}\\
                    &\textbf{(2)}\quad&&\inst{exch}\quad && r_t\quad && r_p\quad &&\text{; extract vtable from object}\\
                    &\textbf{(3)}\quad&&\inst{xori}\quad && r_t\quad && label_{vt} &&\text{; clear address of vtable in $r_t$}\\
                    &\textbf{(4)}\quad&&\inst{addi}\quad && r_p\quad && offset_{ref}\quad &&\text{; Index to ref count pos}\\
                    &\textbf{(5)}\quad&&\inst{exch}\quad && r_t\quad && r_p\quad &&\text{; Extract ref count}\\
                    &\textbf{(6)}\quad&&\inst{xori}\quad && r_t\quad && 1 &&\text{; Clear ref count}\\
                    &\textbf{(7)}\quad&&\inst{subi}\quad && r_p\quad && offset_{ref}\quad &&\text{; Inverse of \textbf{(4)}}\\
                    &\textbf{(8)}\quad&&\cdots \quad &&\quad && \quad &&\text{; Push registers except $r_p$, $r_t$ }\\
                    &\textbf{(9)}\quad&&\cdots\quad &&\quad && &&\text{; Code for $r_t \leftarrow x_{size}$}\\
                    &\textbf{(10)}\quad&&\inst{push}\quad && r_t \quad && &&\text{; Push $r_t$}\\
                    &\textbf{(11)}\quad&&\inst{push}\quad && r_p \quad && &&\text{; Push $r_p$}\\
                    &\textbf{(12)}\quad&&\inst{rbra}\quad && l_{malloc} \quad && &&\text{; Deallocate}\\
                    &\textbf{(13)}\quad&&\inst{pop}\quad && r_p \quad && &&\text{; Inverse of \textbf{(11)}}\\
                    &\textbf{(14)}\quad&&\inst{pop}\quad && r_t \quad && &&\text{; Inverse of \textbf{(10)}}\\
                    &\textbf{(15)}\quad&&\cdots\quad &&\quad && &&\text{; Inverse of \textbf{(9)}}\\
                    &\textbf{(16)}\quad&&\cdots\quad && \quad && &&\text{; Inverse of \textbf{(8)}}\\
                    &\textbf{(17)}\quad&&\cdots\quad && \quad && &&\text{; Inverse of \textbf{(1)}}\\
                \end{alignat*}
            \end{minipage}
        }
    \end{subfigure}
    \caption{PISA translation of heap allocation and deallocation for objects}
    \label{fig:pisa-allocation-deallocation}
\end{figure}
\clearpage

Figure~\ref{fig:pisa-allocation-deallocation} shows how each \textbf{new} and \textbf{delete} statement for objects are translated during compilation. They are simply inverse of each other. For allocation, the object pointer and its size are pushed to the stack and then a jump to the malloc entry point is executed. After allocation, the virtual table and reference count are stored in the first two words of the allocated memory. Note how deallocation jumps and flips the direction of execution using the \inst{rbra} instruction, which then runs the allocation process in reverse. In the figure $x_{size}$ denotes the computed size of objects with class $c$, plus two, to account for the virtual table pointer and reference count space, rounded up to nearest power-of-two.

\begin{figure}[ht]
    \centering
    \begin{equation*}
        \textbf{construct}\ c\ x\quad s\quad\textbf{destruct}\ x
    \end{equation*}
    \resizebox{.8\linewidth}{!}{
    \begin{minipage}{\linewidth}
    \begin{alignat*}{6}
    &\textbf{(1)}\quad&&\inst{xor}\quad &&r_x\quad &&r_{sp}\qquad &&\text{; Store address of new object $x$ in $r_x$}\\
    &\textbf{(2)}\quad&&\inst{push}\quad &&r_x\quad && &&\text{; Push $r_x$ to the stack}\\
    &\textbf{(3)}\quad&&\cdots\cdots && && &&\text{; Code for \textbf{new} $c$ $x$}\\
    &\textbf{(4)}\quad&&\cdots\cdots && && &&\text{; Code for statement $s$}\\
    &\textbf{(5)}\quad&&\cdots\cdots && && &&\text{; Code for \textbf{delete} $c$ $x$}\\
    &\textbf{(6)}\quad&&\inst{pop}\quad &&r_x\quad && &&\text{; Pop $r_x$ from the stack}\\
    &\textbf{(7)}\quad&&\inst{xor}\quad &&r_x\quad &&r_{sp}\qquad &&\text{; Clear $r_x$}
    \end{alignat*}
    \end{minipage}
    }
    \caption{PISA translation of a \rooplpp object block}
    \label{fig:pisa-object-block}
\end{figure}

Figure~\ref{fig:pisa-object-block} shows the updated translation technique for object blocks. In \textsc{Roopl}, the object blocks allocated their objects on the stack, but in \rooplpp, we can now allocate them on the heap. to facilitate this, we simply execute the exact same instructions as in \textbf{new} and \textbf{delete} statements, with body statement execution code in between. As described in section~\ref{sec:local-blocks}, the \textbf{construct}/\textbf{desctruct} block can be considered syntactic sugar, and its usage in a real world example would probably be limited.

\section{Referencing}
\label{sec:referencing-compilation}
As mentioned, one of the main strengths of \rooplpp in terms of increased expressiveness is allowance of multiple references to objects and arrays. When an object or array is constructed we allocate enough space to hold an additional \textit{reference counter} which is initialized to $1$. For each reference copied using the \textbf{copy}-statement, we incrementally increase the reference counter by one. When we \textbf{uncopy} a reference, the reference counter is decreased by one. The object or array cannot be deconstructed until its reference counter has been returned to $1$ as we would have a reference pointer to cleared memory in the heap. Such references are known as dangling pointers.

Figure~\ref{fig:ref-counting-object-layout} shows the object layout of \rooplpp objects with the added space for the reference counting from the original \textsc{Roopl} model in figure~\ref{fig:roopl-object-layout} on page~\pageref{fig:roopl-object-layout}.
\newpage

\begin{figure}[ht!]
    \centering
    \begin{subfigure}[t]{.32\textwidth}
        \vspace{0pt}
        \vskip 0pt
        \centering
        \begin{tikzpicture}
            \draw[dashed] (0, 1.5) -- (0, 2);
            \draw[dashed] (3, 1.5) -- (3, 2);
            \filldraw[fill = grey, draw = black] (0, 1) rectangle (3, 1.5) node[midway] {addr(vtable)};
            \filldraw[fill = darkgrey, draw = black] (0, .5) rectangle (3, 1) node[midway] {reference count};
            \filldraw[fill = grey, draw = black] (0, 0) rectangle (3, .5) node[midway] {x};
            \filldraw[fill = darkgrey, draw = black] (0, -.5) rectangle (3, 0) node[midway] {y};
            \draw[dashed] (0, -.5) -- (0, -1);
            \draw[dashed] (3, -.5) -- (3, -1);
    
            \node at (-.3, 1.25) {\texttt{+}$0$};
            \node at (-.3, .75) {\texttt{+}$1$};
            \node at (-.3, .25) {\texttt{+}$2$};
            \node at (-.3, -.25) {\texttt{+}$3$};
            \draw[->] (3.5, 1.25) -- (3.1, 1.25);
            \node[rotate = 270] at (3.7, 1.25) {$r_{shape}$};
            
            \node at (1.5, 2.5) {\textbf{Shape}};
        \end{tikzpicture}
    \end{subfigure}
    \begin{subfigure}[t]{.32\textwidth}
        \vspace{0pt}
        \vskip 0pt
        \centering
        \begin{tikzpicture}
            \draw[dashed] (0, 1.5) -- (0, 2);
            \draw[dashed] (3, 1.5) -- (3, 2);
            \filldraw[fill = grey, draw = black] (0, 1) rectangle (3, 1.5) node[midway] {addr(vtable)};
            \filldraw[fill = darkgrey, draw = black] (0, .5) rectangle (3, 1) node[midway] {reference count};
            \filldraw[fill = grey, draw = black] (0, 0) rectangle (3, .5) node[midway] {x};
            \filldraw[fill = darkgrey, draw = black] (0, -.5) rectangle (3, 0) node[midway] {y};
            \filldraw[fill = grey, draw = black] (0, -1) rectangle (3, -.5) node[midway] {radius};
            \draw[dashed] (0, -1) -- (0, -1.5);
            \draw[dashed] (3, -1) -- (3, -1.5);
    
            \node at (-.3, 1.25) {\texttt{+}$0$};
            \node at (-.3, .75) {\texttt{+}$1$};
            \node at (-.3, .25) {\texttt{+}$2$};
            \node at (-.3, -.25) {\texttt{+}$3$};
            \node at (-.3, -.75) {\texttt{+}$4$};
            \draw[->] (3.5, 1.25) -- (3.1, 1.25);
            \node[rotate = 270] at (3.7, 1.25) {$r_{circ}$};
            
            \node at (1.5, 2.5) {\textbf{Circle}};
        \end{tikzpicture}
    \end{subfigure}
    \begin{subfigure}[t]{.32\textwidth}
        \vspace{0pt}
        \vskip 0pt
        \centering
        \begin{tikzpicture}
            \draw[dashed] (0, 1.5) -- (0, 2);
            \draw[dashed] (3, 1.5) -- (3, 2);
            \filldraw[fill = grey, draw = black] (0, 1) rectangle (3, 1.5) node[midway] {addr(vtable)};
            \filldraw[fill = darkgrey, draw = black] (0, .5) rectangle (3, 1) node[midway] {reference count};
            \filldraw[fill = grey, draw = black] (0, 0) rectangle (3, .5) node[midway] {x};
            \filldraw[fill = darkgrey, draw = black] (0, -.5) rectangle (3, 0) node[midway] {y};
            \filldraw[fill = grey, draw = black] (0, -1) rectangle (3, -.5) node[midway] {a};
            \filldraw[fill = darkgrey, draw = black] (0, -1.5) rectangle (3, -1) node[midway] {b};
            \draw[dashed] (0, -1.5) -- (0, -2);
            \draw[dashed] (3, -1.5) -- (3, -2);
    
            \node at (-.3, 1.25) {\texttt{+}$0$};
            \node at (-.3, .75) {\texttt{+}$1$};
            \node at (-.3, .25) {\texttt{+}$2$};
            \node at (-.3, -.25) {\texttt{+}$3$};
            \node at (-.3, -.75) {\texttt{+}$4$};
            \node at (-.3, -1.25) {\texttt{+}$5$};
            \draw[->] (3.5, 1.25) -- (3.1, 1.25);
            \node[rotate = 270] at (3.7, 1.25) {$r_{rect}$};
            
            \node at (1.5, 2.5) {\textbf{Rectangle}};
        \end{tikzpicture}
    \end{subfigure}
    \caption[Illustration of object memory layout]{Illustration of prefixing in the memory layout of three \rooplpp objects}
    \label{fig:ref-counting-object-layout}
\end{figure}
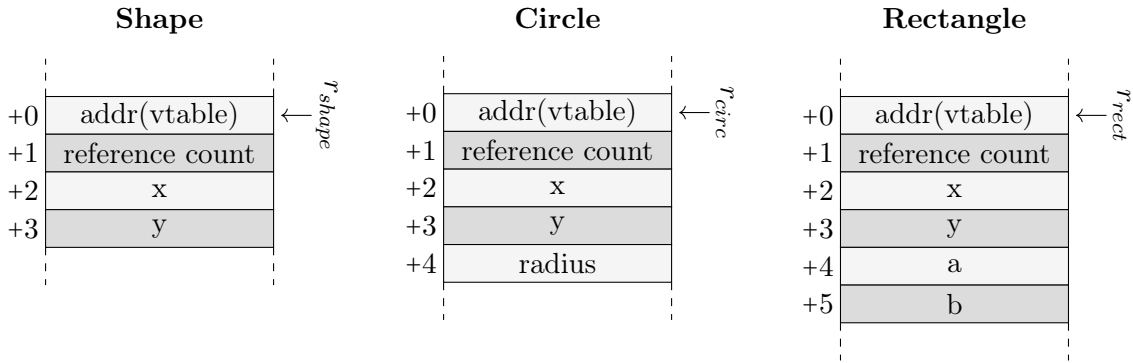

Figure~\ref{fig:pisa-referencing} shows the translated \textsc{Pisa} code for the \textbf{copy} and \textbf{uncopy} statements. As shown, they are both very simple and each others inverse. For copying, the address of the passed variable $x$ is simply copied into the zero-cleared value of $x'$ and the reference count incremented by one. For deletion, the address is cleared and the reference count decremented. Copying and clearing is done through the \inst{xor} instruction. These translations features no error handling, but a solution is discussed in section~\ref{sec:error-handling}.

\begin{figure}[ht]
    \centering
    
    \begin{subfigure}[t]{0.495\linewidth}
    \vskip 0pt
    \centering
    \begin{equation*} 
        \textbf{copy}\ c\ x\ x'
    \end{equation*}
    
    \resizebox{.8\linewidth}{!}{
        \begin{minipage}{1.025\linewidth}
            \begin{alignat*}{5}
                &\textbf{(1)}\quad&&\dots\quad && && &&\text{; Code for $r_{p}\ \leftarrow\ addr(x)$}\\
                &\textbf{(2)}\quad&&\dots\quad && && &&\text{; Code for $r_{cp}\ \leftarrow\ value(x')$}\\
                &\textbf{(3)}\quad&&\inst{xor}\quad &&r_{cp}\quad && r_p \quad &&\text{; Copy address of $x$ into $x'$}\\
                &\textbf{(4)}\quad&&\inst{addi}\quad &&r_{p}\quad && offset_{ref} &&\text{; Index to reference count address}\\
                &\textbf{(5)}\quad&&\inst{exch}\quad &&r_{t}\quad && r_p &&\text{; Extract reference count}\\
                &\textbf{(6)}\quad&&\inst{addi}\quad &&r_{t}\quad && 1 &&\text{; Increment reference count}\\
                &\textbf{(7)}\quad&&\inst{exch}\quad &&r_{t}\quad && r_p &&\text{; Store updated reference count}\\
                &\textbf{(8)}\quad&&\inst{subi}\quad &&r_{p}\quad && offset_{ref} &&\text{; Inverse of \textbf{(3)}}\\
                &\textbf{(9)}\quad&&\dots\quad && && &&\text{; Inverse of \textbf{(2)}}\\
                &\textbf{(10)}\quad&&\dots\quad && && &&\text{; Inverse of \textbf{(1)}}
            \end{alignat*}
        \end{minipage}
    }
    \end{subfigure}
    \begin{subfigure}[t]{0.495\linewidth}
    \vskip 0pt
    \centering
    \begin{equation*}
        \textbf{uncopy}\ c\ x\ x'
    \end{equation*}
    
    \resizebox{.8\linewidth}{!}{
        \begin{minipage}{1.025\linewidth}
            \begin{alignat*}{5}
                &\textbf{(1)}\quad&&\dots\quad && && &&\text{; Code for $r_{p}\ \leftarrow\ addr(x)$}\\
                &\textbf{(2)}\quad&&\dots\quad && && &&\text{; Code for $r_{cp}\ \leftarrow\ value(x')$}\\
                &\textbf{(3)}\quad&&\inst{xor}\quad &&r_{cp}\quad && r_p \quad &&\text{; Clear address of $x$ from $x'$}\\
                &\textbf{(4)}\quad&&\inst{addi}\quad &&r_{p}\quad && offset_{ref} &&\text{; Index to reference count address}\\
                &\textbf{(5)}\quad&&\inst{exch}\quad &&r_{t}\quad && r_p &&\text{; Extract reference count}\\
                &\textbf{(6)}\quad&&\inst{subi}\quad &&r_{t}\quad && 1 &&\text{; Decrement reference count}\\
                &\textbf{(7)}\quad&&\inst{exch}\quad &&r_{t}\quad && r_p &&\text{; Store updated reference count}\\
                &\textbf{(8)}\quad&&\inst{subi}\quad &&r_{p}\quad && offset_{ref} &&\text{; Inverse of \textbf{(3)}}\\
                &\textbf{(9)}\quad&&\dots\quad && && &&\text{; Inverse of \textbf{(2)}}\\
                &\textbf{(10)}\quad&&\dots\quad && && &&\text{; Inverse of \textbf{(1)}}
            \end{alignat*}
        \end{minipage}
    }
    \end{subfigure}
    
    \caption{\textsc{Pisa} translation of the reference copying and deletion statements}
    \label{fig:pisa-referencing}
\end{figure}

\section{Arrays}
\label{sec:arrays}
The fixed-sized arrays in \rooplpp are also heap allocated to allow dynamic lifetime. The array memory layout is presented in figure~\ref{fig:array-layout}. As shown, the arrays feature two additional fields to store the size of the array and the reference count. Additionally, integer arrays store their values directly in the array, while object arrays are a simple pointer stores.

As the size of a \rooplpp array is determined by a passed expression evaluation, it is unknown at compile time. This also means that out-of-bounds checking cannot be conducted during compilation. A possible solution for this is presented in section~\ref{sec:error-handling}.  

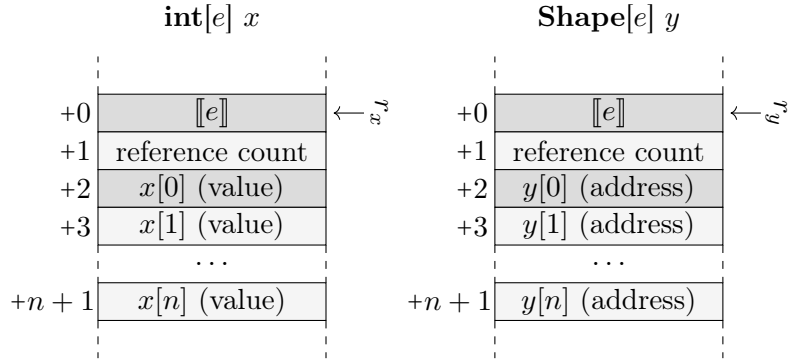
\begin{figure}[ht]
    \centering
    \begin{subfigure}[ht]{.32\textwidth}
        \vskip 0pt
        \centering
        \begin{tikzpicture}
            \draw[dashed] (0, 2) -- (0, 2.5);
            \draw[dashed] (3, 2) -- (3, 2.5);
            \filldraw[fill = darkgrey, draw = black] (0, 1.5) rectangle (3, 2) node[midway] {$\llbracket e \rrbracket$};
            \filldraw[fill = grey, draw = black] (0, 1) rectangle (3, 1.5) node[midway] {reference count};
            \filldraw[fill = darkgrey, draw = black] (0, .5) rectangle (3, 1) node[midway] {$x[0]$ (value)};
            \filldraw[fill = grey, draw = black] (0, 0) rectangle (3, .5) node[midway] {$x[1]$ (value)};
            \draw[dashed] (0, 0) -- (0, -.5);
            \node at (1.5, -.25) {$\dots$};
            \draw[dashed] (3, 0) -- (3, -.5);
            \filldraw[fill = grey, draw = black] (0, -.5) rectangle (3, -1) node[midway] {$x[n]$ (value)};
            \draw[dashed] (0, -1) -- (0, -1.5);
            \draw[dashed] (3, -1) -- (3, -1.5);
    
            \node at (-.3, 1.75) {\texttt{+}$0$};
            \node at (-.3, 1.25) {\texttt{+}$1$};
            \node at (-.3, .75) {\texttt{+}$2$};
            \node at (-.3, .25) {\texttt{+}$3$};
            \node at (-.6, -.75) {\texttt{+}$n+1$};
            \draw[->] (3.5, 1.75) -- (3.1, 1.75);
            \node[rotate = 270] at (3.7, 1.75) {$r_{x}$};
            
            \node at (1.5, 3) {\textbf{int}[$e$] $x$};
        \end{tikzpicture}
    \end{subfigure}
    \begin{subfigure}[ht]{.32\textwidth}
        \vskip 0pt
        \centering
        \begin{tikzpicture}
            \draw[dashed] (0, 2) -- (0, 2.5);
            \draw[dashed] (3, 2) -- (3, 2.5);
            \filldraw[fill = darkgrey, draw = black] (0, 1.5) rectangle (3, 2) node[midway] {$\llbracket e \rrbracket$};
            \filldraw[fill = grey, draw = black] (0, 1) rectangle (3, 1.5) node[midway] {reference count};
            \filldraw[fill = darkgrey, draw = black] (0, .5) rectangle (3, 1) node[midway] {$y[0]$ (address)};
            \filldraw[fill = grey, draw = black] (0, 0) rectangle (3, .5) node[midway] {$y[1]$ (address)};
            \draw[dashed] (0, 0) -- (0, -.5);
            \node at (1.5, -.25) {$\dots$};
            \draw[dashed] (3, 0) -- (3, -.5);
            \filldraw[fill = grey, draw = black] (0, -.5) rectangle (3, -1) node[midway] {$y[n]$ (address)};
            \draw[dashed] (0, -1) -- (0, -1.5);
            \draw[dashed] (3, -1) -- (3, -1.5);
    
            \node at (-.3, 1.75) {\texttt{+}$0$};
            \node at (-.3, 1.25) {\texttt{+}$1$};
            \node at (-.3, .75) {\texttt{+}$2$};
            \node at (-.3, .25) {\texttt{+}$3$};
            \node at (-.6, -.75) {\texttt{+}$n+1$};
            \draw[->] (3.5, 1.75) -- (3.1, 1.75);
            \node[rotate = 270] at (3.7, 1.75) {$r_{y}$};
            
            \node at (1.5, 3) {\textbf{Shape}[$e$] $y$};
        \end{tikzpicture}
    \end{subfigure}
    \caption{Illustration of prefixing in the memory layout of two \rooplpp arrays}
    \label{fig:array-layout}
\end{figure}
\newpage
\subsection{Construction and Destruction}
\label{subsec:construction-destruction}
As \rooplpp arrays also are heap allocated, the buddy allocation implementation is also used for allocating arrays. The only difference between object and array allocation is that no virtual table is stored in the allocated space while the offsets for the reference counter are shared for both types. Due to this fact, \textbf{copy} and \textbf{uncopy} \textsc{Pisa} blocks generated during compile time are exactly the same for arrays and objects, as shown in the previous section.

\begin{figure}[ht]
    \centering
    \begin{subfigure}[t]{0.495\linewidth}
        \vskip 0pt
        \centering
        \begin{equation*} 
            \textbf{new}\ a[e]\ x
        \end{equation*}
        \resizebox{.9\linewidth}{!}{
            \begin{minipage}{\linewidth}
                \begin{alignat*}{5}
                    &\textbf{(1)}\quad&&\cdots \quad &&\quad && \quad &&\text{; Push registers }\\
                    &\textbf{(2)}\quad&&\cdots\quad &&\quad && &&\text{; Code for $r_t \leftarrow \llbracket e \rrbracket + 2$}\\
                    &\textbf{(3)}\quad&&\inst{push}\quad && r_t \quad && &&\text{; Push $r_t$}\\
                    &\textbf{(4)}\quad&&\inst{push}\quad && r_p \quad && &&\text{; Push $r_p$}\\
                    &\textbf{(5)}\quad&&\inst{bra}\quad && l_{malloc} \quad && &&\text{; Allocate array}\\
                    &\textbf{(6)}\quad&&\inst{pop}\quad && r_p \quad && &&\text{; Inverse of \textbf{(4)}}\\
                    &\textbf{(7)}\quad&&\inst{pop}\quad && r_t \quad && &&\text{; Inverse of \textbf{(3)}}\\
                    &\textbf{(9)}\quad&&\cdots\quad && \quad && &&\text{; Inverse of \textbf{(1)}}\\
                    &\textbf{(10)}\quad&&\cdots\quad &&\quad && &&\text{; Code for $r_v \leftarrow \llbracket addr(x) \rrbracket$}\\
                    &\textbf{(11)}\quad&&\inst{subi}\quad && r_t\quad && 2 &&\text{; $r_t \leftarrow \llbracket e \rrbracket$}\\
                    &\textbf{(12)}\quad&&\inst{exch}\quad && r_t\quad && r_p\quad &&\text{; Store size in new array}\\
                    &\textbf{(13)}\quad&&\inst{addi}\quad && r_p\quad && offset_{ref}\quad &&\text{; Index to ref count pos}\\
                    &\textbf{(14)}\quad&&\inst{xori}\quad && r_t\quad && 1 &&\text{; Init ref count}\\
                    &\textbf{(15)}\quad&&\inst{exch}\quad && r_t\quad && r_p\quad &&\text{; Store ref count}\\
                    &\textbf{(16)}\quad&&\inst{subi}\quad && r_p\quad && offset_{ref}\quad &&\text{; Inverse of \textbf{(13)}}\\
                    &\textbf{(17)}\quad&&\inst{exch}\quad && r_p\quad && r_v\quad &&\text{; Store address in variable}\\
                    &\textbf{(18)}\quad&&\cdots\quad &&\quad && &&\text{; Inverse of \textbf{(10)}}\\
                \end{alignat*}
            \end{minipage}
        }
    \end{subfigure}
    \begin{subfigure}[t]{0.495\linewidth}
        \vskip 0pt
        \centering
        \begin{equation*}
            \textbf{delete}\ a[e]\ x
        \end{equation*}
        \resizebox{.9\linewidth}{!}{
            \begin{minipage}{\linewidth}
                \begin{alignat*}{5}
                    &\textbf{(1)}\quad&&\cdots\quad &&\quad && &&\text{; Code for $r_p \leftarrow \llbracket addr(x) \rrbracket$}\\
                    &\textbf{(2)}\quad&&\cdots\quad && \quad && \quad &&\text{; Code for $r_v \leftarrow \llbracket e \rrbracket$}\\
                    &\textbf{(3)}\quad&&\inst{addi}\quad && r_p\quad && offset_{ref}\quad &&\text{; Index to ref count pos}\\
                    &\textbf{(4)}\quad&&\inst{exch}\quad && r_t\quad && r_p\quad &&\text{; Extract ref count}\\
                    &\textbf{(5)}\quad&&\inst{xori}\quad && r_t\quad && 1 &&\text{; Clear ref count}\\
                    &\textbf{(6)}\quad&&\inst{subi}\quad && r_p\quad && offset_{ref}\quad &&\text{; Inverse of \textbf{(3)}}\\
                    &\textbf{(7)}\quad&&\inst{exch}\quad && r_t\quad && r_p\quad &&\text{; extract size from object}\\
                    &\textbf{(8)}\quad&&\inst{xori}\quad && r_t\quad && r_v &&\text{; clear size in $r_t$}\\
                    &\textbf{(9)}\quad&&\cdots \quad &&\quad && \quad &&\text{; Push registers except $r_p$, $r_v$ }\\
                    &\textbf{(10)}\quad&&\inst{addi}\quad && r_v \quad && 2 &&\text{; Actual size of array}\\
                    &\textbf{(11)}\quad&&\inst{push}\quad && r_v \quad && &&\text{; Push $r_v$}\\
                    &\textbf{(12)}\quad&&\inst{push}\quad && r_p \quad && &&\text{; Push $r_p$}\\
                    &\textbf{(13)}\quad&&\inst{rbra}\quad && l_{malloc} \quad && &&\text{; Deallocate array}\\
                    &\textbf{(14)}\quad&&\inst{pop}\quad && r_p \quad && &&\text{; Inverse of \textbf{(12)}}\\
                    &\textbf{(15)}\quad&&\inst{pop}\quad && r_v \quad && &&\text{; Inverse of \textbf{(11)}}\\
                    &\textbf{(16)}\quad&&\inst{subi}\quad && r_v \quad && 2 &&\text{; Inverse of \textbf{(10)}}\\
                    &\textbf{(17)}\quad&&\cdots\quad &&\quad && &&\text{; Inverse of \textbf{(9)}}\\
                    &\textbf{(18)}\quad&&\cdots\quad && \quad && &&\text{; Inverse of \textbf{(2)}}\\
                    &\textbf{(19)}\quad&&\cdots\quad && \quad && &&\text{; Inverse of \textbf{(1)}}\\
                \end{alignat*}
            \end{minipage}
        }
    \end{subfigure}
    \caption{\textsc{Pisa} translations of array allocation and deallocation statements}
    \label{fig:array-allocation-deallocation}
\end{figure}

Figure~\ref{fig:array-allocation-deallocation} shows the translation schemes used for array allocation and deallocation. As said, these are almost identical to the object allocation and deallocation schemes presented in figure~\ref{fig:pisa-allocation-deallocation} on page~\pageref{fig:pisa-allocation-deallocation}. Classes are analyzed during a compilation phase and their allocation size, the object size + 2 (for virtual table and reference counter) rounded up to nearest power-of-two. The size of arrays cannot be determined during compilation, as that would require evaluating the expression passed to the initialization call, and as such, we add the overhead needed directly in the allocation and deallocation instructions. While the two blocks are code are not exact opposites they are functionally inverse of each other. An extra \inst{xori} instruction on line \textbf{(8)} in the deallocation block has been included to clear the stored array size using the value of the passed expression and further use this size for the inverse \textbf{malloc} subroutine.  

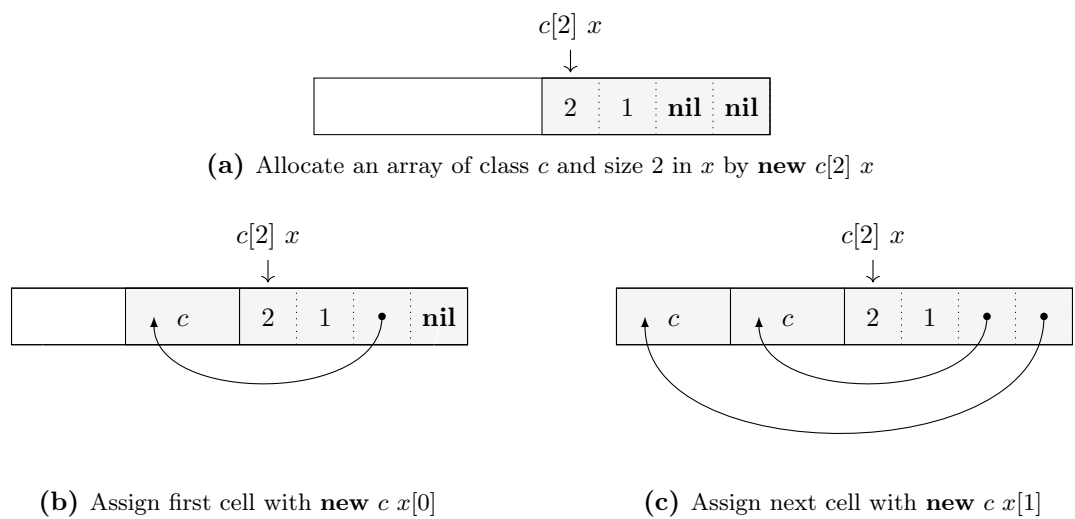
\begin{figure}[ht]
    \centering
        \begin{subfigure}{.75\textwidth}
            \centering
            \begin{tikzpicture}[scale=0.75]
                % Fills
                \draw[fill=grey] (4,0) rectangle (8,1);
                
                % Boxes
                \draw (0,0) rectangle (8, 1);

                % lines
                \draw[dotted] (5, 0) to (5, 1);
                \draw[dotted] (6, 0) to (6, 1);
                \draw[dotted] (7, 0) to (7, 1);
                \draw[dotted] (8, 0) to (8, 1);

                % pointer
                \draw[->] (4.5, 1.5) -- (4.5, 1.1);
                \node[above] at (4.5, 1.5) {\small $c[2]\ x$};

                % labels
                \node[above] at (4.5, 0.15) {\small $2$}; 
                \node[above] at (5.5, 0.15) {\small $1$};
                \node[above] at (6.5, 0.15) {\small \textbf{nil}}; 
                \node[above] at (7.5, 0.15) {\small \textbf{nil}}; 
            \end{tikzpicture}
            \caption{\footnotesize Allocate an array of class $c$ and size $2$ in $x$ by $\textbf{new}\ c[2]\ x$}
        \end{subfigure}%
        \vskip 1em
        \begin{subfigure}[t]{.5\textwidth}
            \vspace{0pt}
            \vskip 0pt
            \centering
            \begin{tikzpicture}[scale=0.75]
                % Fills
                \draw[fill=grey] (4,0) rectangle (8,1);
                \draw[fill=grey] (2,0) rectangle (4,1);
                
                % Boxes
                \draw (0,0) rectangle (8, 1);

                % lines
                \draw[dotted] (5, 0) to (5, 1);
                \draw[dotted] (6, 0) to (6, 1);
                \draw[dotted] (7, 0) to (7, 1);
                \draw[dotted] (8, 0) to (8, 1);

                % pointer
                \draw[->] (4.5, 1.5) -- (4.5, 1.1);
                \node[above] at (4.5, 1.5) {\small $c[2]\ x$};

                % labels
                \node[above] at (3, 0.15) {\small $c$};
                \node[above] at (4.5, 0.15) {\small $2$}; 
                \node[above] at (5.5, 0.15) {\small $1$};
                \node[above] at (7.5, 0.15) {\small \textbf{nil}};                

                % Arrows
                \node[circle,fill,inner sep=1pt] at (6.5, 0.5) {};
                \draw[-latex] (6.5, 0.5) to[out=-90, in=-90] (2.5, 0.5);
                \draw[-latex, white] (7.5, 0.5) to[out=-90, in=-90] (0.5, 0.5);
                \draw (0,0) rectangle (8, 1);
            \end{tikzpicture}
            \caption{\footnotesize Assign first cell with $\textbf{new}\ c\ x[0]$}
        \end{subfigure}%
        \begin{subfigure}[t]{.5\textwidth}
            \vspace{0pt}
            \vskip 0pt
            \centering
            \begin{tikzpicture}[scale=0.75]
                % Fills
                \draw[fill=grey] (4,0) rectangle (8,1);
                \draw[fill=grey] (2,0) rectangle (4,1);
                \draw[fill=grey] (0,0) rectangle (2,1);
                
                % Boxes
                \draw (0,0) rectangle (8, 1);

                % lines
                \draw[dotted] (5, 0) to (5, 1);
                \draw[dotted] (6, 0) to (6, 1);
                \draw[dotted] (7, 0) to (7, 1);
                \draw[dotted] (8, 0) to (8, 1);

                % pointer
                \draw[->] (4.5, 1.5) -- (4.5, 1.1);
                \node[above] at (4.5, 1.5) {\small $c[2]\ x$};

                % labels
                \node[above] at (1, 0.15) {\small $c$};
                \node[above] at (3, 0.15) {\small $c$};
                \node[above] at (4.5, 0.15) {\small $2$}; 
                \node[above] at (5.5, 0.15) {\small $1$};               

                % Arrows
                \node[circle,fill,inner sep=1pt] at (6.5, 0.5) {};
                \draw[-latex] (6.5, 0.5) to[out=-90, in=-90] (2.5, 0.5);

                \node[circle,fill,inner sep=1pt] at (7.5, 0.5) {};
                \draw[-latex] (7.5, 0.5) to[out=-90, in=-90] (0.5, 0.5);
            \end{tikzpicture}
            \caption{\footnotesize Assign next cell with $\textbf{new}\ c\ x[1]$}
        \end{subfigure}%
    \caption{Illustration of array memory storage layout}
    \label{fig:array-memory-storage}
\end{figure}

Figure~\ref{fig:array-memory-storage} shows how object arrays simply contain pointers to allocated objects. For integer arrays, the cell values would stored directly in the allocated array space instead. 

\subsection{Array Element Access}
\label{subsec:array-element-access}
Array elements are simply passed as any other variable to methods or statements. Based on the variable type, compilation of various statements individually determines whether the address or the value of the passed variable should be used for the compiling the statement. For arrays, this is no different. If an integer array element is passed, it is treated just liked a regular integer variable. For an object array element, it is treated just like a regular object variable.

\section{Error Handling}
\label{sec:error-handling}
While a program written in \rooplpp might be syntactically valid and well-typed, this is not a guarantee that it executes successfully. A number of conditions exist, which cannot be determined at compile time, which in turn results in erroneous executed code. \citeauthor{th:roopl} describes the following conditions:

\begin{itemize}
    \item If the entry expression of a conditional is \textbf{true}, then the exit assertion should also be \textbf{true} after executing the then-branch.
    \item If the entry expression of a conditional is \textbf{false}, then the exit assertion should also be \textbf{false} after executing the else-branch.
    \item The entry expression of a loop should initially be \textbf{true}.
    \item If the exit assertion of a loop is \textbf{false}, then the entry expression should also be \textbf{false} after executing the loop-statement.
    \item All instance variables should be zero-cleared within an object block before the object is deallocated.
    \item The value of a local variable should always match the value of the delocal-expression after the block statement has executed~\cite{th:roopl}.
\end{itemize}

The extensions made to \textsc{Roopl} in \rooplpp brings forth a number of additional conditions:

\begin{itemize}
    \item All fields of an object instance should be zero-cleared before the object is deallocated using the \textbf{delete} statement.
    \item All cells of an instance should be zero-cleared before the array is deallocated using the \textbf{delete} statement.
    \item Local object blocks should have their fields zero-cleared after the execution of the block statement.
    \item Local array blocks should have their cells zero-cleared after the execution of the block statement.
    \item If the value of a local object variable is exchanged during its block statement and the new value is an object reference, this object must have its fields zero-cleared after the execution of the block statement.
    \item If the value of a local array variable is exchanged during its block statement and the new value is an array reference, this array must have its cell zero-cleared after the execution of the block statement.
    \item The variable in the \textbf{new} statement must be zero-cleared beforehand.
    \item The variable in the \textbf{copy} statement must be zero-cleared beforehand.
    \item An object variable must be initialized using \textbf{new} or \textbf{copy} before its methods can be called.
    \item An array variable must be initialized using \textbf{new} or \textbf{copy} before its fields can be accessed.
    \item Array cell indices must be within bounds defined in the expression passed during initialization.
    \item Only one reference to an object or an array must exist when executing the \textbf{delete} statement.
    \item Swapping cell values between a subtype $A$ variable and parent-type $B$ array is only allowed if the value stored in the variable is also $A$ afterwards.
\end{itemize}

It is the programmer's responsibility to meet these conditions. As these conditions, in general, cannot be determined at compile time, undefined program behaviour will occur as the termination will continue silently, resulting in erroneous program state. We can insert run time error checks in the generated instructions such that the program is terminated if one of the conditions does not hold. The run time error checks can be added as dynamic error checks using error routines defined at labels, such as $label_{uninitialized\_object}$ which the program can jump to, if such a condition is unmet. \citeauthor{th:roopl} presented an example for dynamic error checking for local blocks in~\cite{th:roopl}. \textsc{Pisa} and its simulator PendVM is, however, limited and does not support exit codes natively. To fully support dynamic error checking, PendVM could be extended to read from a value from a designated register to supply a more meaningful message for the programmer in the case of a run time exit.
 
\section{Implementation}
\label{sec:implementation}
The \rooplpp compiler (\textsc{Rooplppc}) was implemented using techniques and translation schemes presented in this chapter, expanding upon the work of the original \textsc{Roopl} compiler (\textsc{Rooplc}). The compiler serves as a proof-of-concept and simply performs one-to-one translations of \rooplpp code to \textsc{Pisa} code without any optimizations along the way. The compiler is written in \textsc{Haskell} 7.10 and the translated output was tested on the Pendulum simulator, PendVM~\cite{cr:pendvm}. 

As with the \textsc{Roopl} compiler, the \rooplpp compiler is structured around the same six separate compilation phases.
\begin{enumerate}
    \item \textbf{Parsing} consists of constructing an abstract syntax tree from the input program text using parser combinators from the \textsc{Parsec} library in \textsc{Haskell}.
    \item \textbf{Class Analysis} verifies inheritance cycles, duplicated method names or fields and base classes. In this phase, we also compute the allocation size of each class
    \item \textbf{Scope Analysis} constructs the virtual and symbol tables and maps every identifier to a unique variable or method.
    \item \textbf{Type Checking} verifies that the parsed program is well-typed.
    \item \textbf{Code Generation} translates the abstract syntax tree to blocks of \textsc{Pisa} code in a recursive descent.
    \item \textbf{Macro Expansion} expands macros left by the code generator for i.e. configuration variables, etc.
\end{enumerate}

Compiled \textsc{Roopl} programs have a size increase by a factor of 10 to 15 in terms of the lines of code. For \rooplpp the size increase is much larger, partially due to the increase of static code included in form of the memory manager using the buddy layout described in this chapter and partially because heap allocations are more costly than stack allocations in terms of lines of code.

The \textsc{Roopl} compiler was implemented in $1403$ lines of \textsc{Haskell} and the \rooplpp compiler was extended to $2046$ lines of \textsc{Haskell}.

The entire compiler source code as well as example programs and their compiled versions are provided in the appendices and in the supplied ZIP archive. It is also hosted on Github as open source software under the MIT license at \url{https://github.com/cservenka/ROOPLPPC}.

Building and usage of the compiler is supplied in the README.md file found in the ZIP archive and in appendix~\ref{app:rooplc-source}.

\section{Evaluation}
\label{sec:evaluation}
For evaluating the results of the implemented compiler, it was tested against example code provided throughout this thesis. Tests programs utilizing the linked list, doubly-linked list and binary tree data structures and the RTM implementation are found in appendix~\ref{app:example-output}.

\begin{figure}[ht]
    \centering
    \begin{tabular}{ c | c | c | c}
        \textbf{Program} & \textbf{\rooplpp LOC} & \textbf{\textsc{Pisa} LOC} & \textbf{Number of executed instructions}  \\ \hline 
        Linked List & 61 & 1280 & 18015 \\
        Doubly-Linked List & 66 & 1339 & 21825 \\
        Binary Tree & 86 & 2056 & 6065\\ 
        RTM Simulation & 211 & 6716 & 64922\\ 
      \end{tabular}
    \caption{Lines of code comparison between target and compiled \rooplpp programs}
    \label{fig:lines-of-code}
\end{figure}

The linked list test programs simply instantiates ten cells and links them in their respective lists. The binary tree test program instantiates three nodes and adds them to the tree structure, which afterwards is traversed to determine the sum of the nodes and finally mirroring the tree. The Reversible Turing Machine implementing incrementation of a non-negative $n$-bit binary number by 1 originally described in~\cite{ty:ejanus} has been implemented in \rooplpp and successfully converts its initial tape value in little endian form of $1101$ to $0011$ after termination. It should be noted that these test programs require additional stack space during their lengthy computations and as such has been compiled with twice the length between the stack and heap to allow further stack growth.

As discussed, the compiler is considered proof-of-concept and no noteworthy optimizations has been implemented. However, for the sake of giving the reader an idea of the size blowup of a compiled \rooplpp program, figure~\ref{fig:lines-of-code} details this difference. The lines of translated \textsc{Pisa} instructions includes the $204$ instructions needed for the \textbf{malloc} and \textbf{malloc1} \textsc{Pisa}-equivalent mechanisms. The last row of the table shows how many instructions are execution during simulation using PendVM.

\newpage

\chapter{Conclusions}
\label{chp:conclusions}
We formally presented a dynamic memory management extension for the reversible object-oriented programing language, \textsc{Roopl}, in the form of the superset language \rooplpp. The extension expands upon the previously presented static typing system defining well-typedness. The language successfully extends the expressiveness of its predecessor by allowing more flexibility within the domain of reversible object-oriented programming. With \rooplpp we, as reversible programmers, can now define and model non-trivial dynamic data structures in a reversible setting, such as lists, trees and graphs. We illustrated this by example programs such as a new reversible Turing machine simulator along with implementations for linked lists, doubly-linked lists and binary trees as well as techniques for traversing these. Besides expanding the expressiveness of \textsc{Roopl}, we have also shown that complex dynamic data structures are not only feasible, but furthermore do not contradict the reversible computing paradigm.

We presented various dynamic memory management layouts and how each would translate into the reversible allocation algorithms. Weighing the advantages and disadvantages of each, the Buddy Memory layout was found to translate into reversible code very naturally with few side effects and addressed a number of disadvantages found in other considered layouts. With dynamically lifetimed objects the allocation and deallocation order is important in terms of a entirely garbage-free computation. In most cases with \rooplpp, we only obtain partially garbage-free computations, as our free lists might not be restored to their original form, without an effective garbage collector design for the memory manager.

Techniques for clean translations of extended parts of the language, such as the memory manager, the new fixed-sized array type and reference counting have been demonstrated and implemented in a proof-of-concept compiler for validation.

With the dynamic memory manager for reversible object-oriented programming languages allowing dynamic object-scopes and multiple references, exemplified by \rooplpp, we have successfully taking an additional step in the direction towards high-level abstractions reversible computations.  

\section{Future Work}
\label{sec:future-work}
Naturally with the discovery of feasibility of non-trivial, reversible data structures with the introduction of \rooplpp, further study of design and implementation of reversible algorithms working with these data structures are an obvious contender for future research. Data structures such as lists, graphs and trees could potentially provide very interesting future reversible programs.

In terms of the future of reversible object-oriented languages, additional work could be made to extend the fixed-sized array type with a fully dynamic array supporting multiple dimensions. This addition could further help the discovery and research of reversible data structures such as trees and graphs. Such an extension could perhaps be added via a \textbf{put} and \textbf{take} statement pair, being each others inverse. After a dynamic array has been declared, it could automatically reallocate or upscale its internal space when putting new data outside of its current bounds. In reverse, the space could shrink or reallocate when removing the largest indexed value. The current memory management layout will still suffice for this extension.

Finally, more research could be conducted into reversible heap managers. We provided a simple manager which translated to our problem domain naturally. To obtain completely garbage free computations, a garbage collector could be designed to work with the reversible Buddy Memory memory manager. A reversible garbage collector for non-mutable objects has been designed and shown feasible for the reversible functional language \textsc{Rcfun} in~\cite{tm:garbage}. Additionally, experimentation with implementing the Buddy Memory layout into other reversible languages with dynamic allocation and deallocation such as \textsc{R-While} and \textsc{R-Core} provides an interesting opportunity~\cite{rg:rwhile, rg:rcore}.

\newpage

\printreferences 
\newpage 
 
\begin{appendices}
    \chapter{Pisa Translated Buddy Memory}
\label{app:pisa-translated-buddy-memory}
\allowdisplaybreaks
{\tiny
\begin{alignat*}{7}
    &\textbf{(1)}\quad&&malloc1_{top}\ \texttt{:}\quad  &&\inst{bra}\quad &&malloc1_{bot} \span\omit\span\quad \span\omit\span\quad &&\text{; Receive jump}\\ 
    &\textbf{(2)}\quad&& &&\inst{pop}\quad&&r_{ro}&& && &&\text{; Pop return offset from the stack}\\
    &\textbf{(3)}\quad&& &&\cdots\cdots && && && &&\text{; Inverse of \textbf{(7)}}\\
    &\textbf{(4)}\quad&&malloc1_{entry}\ \texttt{:}\quad&&\inst{swapbr}\quad &&r_{ro} && && &&\text{; Malloc1 entry and exit point}\\
    &\textbf{(5)}\quad&& &&\inst{neg}\quad &&r_{ro} && && &&\text{; Negate return offset}\\        
    &\textbf{(6)}\quad&& &&\inst{push}\quad &&r_{ro} && && &&\text{; Store return offset on stack}\\  
    &\textbf{(7)}\quad&& &&\cdots\cdots && && && &&\text{; Code for $r_{fl}\ \leftarrow\ addr(freelists[counter])$}\\
    &\textbf{(8)}\quad&& &&\cdots\cdots && && && &&\text{; Code for $r_{block}\ \leftarrow\ \llbracket freelists[counter] \rrbracket$}\\
    &\textbf{(9)}\quad&& &&\cdots\cdots && && && &&\text{; Code for $r_{e1_o}\ \leftarrow\ \llbracket c_{size} < object_{size} \rrbracket$}\\
    &\textbf{(10)}\quad&& &&\inst{xor}\quad &&r_t && r_{e1_o} && &&\text{; Copy value of $c_{size} < object_{size}$ into $r_t$}\\        
    &\textbf{(11)}\quad&& &&\cdots\cdots && && && &&\text{; Inverse of \textbf{(9)}}\\ 
    &\textbf{(12)}\quad&&o_{test}\ \texttt{:}\quad &&\inst{beq} &&r_t && r_0 && o_{test_f} && \text{; Receive jump}\\
    &\textbf{(13)}\quad&& &&\inst{xori} &&r_t && 1 && && \text{; Clear $r_t$}\\
    &\textbf{(14)}\quad&& &&\inst{addi} &&r_{c} && 1 && && \text{; $Counter\texttt{++}$}\\
    &\textbf{(15)}\quad&& &&\inst{rl} &&r_{sc}\ && 1 && && \text{; Call $double(c_{size}$)}\\
    &\textbf{(16)}\quad&& &&\cdots\cdots && && && &&\text{; Inverse of \textbf{(7)}}\\
    &\textbf{(17)}\quad&& &&\cdots\cdots && && && &&\text{; Code for pushing temp reg values to stack}\\
    &\textbf{(18)}\quad&& &&\inst{bra}\quad &&malloc1_{entry} \span\omit\span\quad \span\omit\span\quad && \text{; Call $malloc1()$)}\\
    &\textbf{(19)}\quad&& &&\cdots\cdots && && && &&\text{; Inverse of \textbf{(17)}}\\
    &\textbf{(20)}\quad&& &&\inst{rr} &&r_{sc}\ && 1 && && \text{; Inverse of \textbf{(15)}}\\
    &\textbf{(21)}\quad&& &&\inst{subi} &&r_{c} && 1 && && \text{; Inverse of \textbf{(14)}}\\
    &\textbf{(22)}\quad&& &&\inst{xori} &&r_t && 1 && && \text{; Set $r_t = 1$}\\
    &\textbf{(23)}\quad&&o_{assert_t}\ \texttt{:}\quad &&\inst{bra} &&o_{assert} \span\omit\span\quad \span\omit\span\quad && \text{; Jump}\\
    &\textbf{(24)}\quad&&o_{test_f}\ \texttt{:}\quad &&\inst{bra} &&o_{test} \span\omit\span\quad \span\omit\span\quad && \text{; Receive jump}\\
    &\textbf{(25)}\quad&& &&\cdots\cdots && && && &&\text{; Code for $r_{e1_i}\ \leftarrow\ \llbracket addr(freelists[counter]) \neq 0 \rrbracket$}\\
    &\textbf{(26)}\quad&& &&\inst{xor}\quad &&r_{t2} && r_{e1_i} && &&\text{; Copy value of $r_{e1_i}$ into $r_{t2}$}\\        
    &\textbf{(27)}\quad&& &&\cdots\cdots && && && &&\text{; Inverse of \textbf{(25)}}\\
    &\textbf{(28)}\quad&&i_{test}\ \texttt{:}\quad &&\inst{beq} &&r_{t2} && r_0 && i_{test_f} && \text{; Receive jump}\\
    &\textbf{(29)}\quad&& &&\inst{xori} &&r_{t2} && 1 && && \text{; Clear $r_{t2}$}\\
    &\textbf{(30)}\quad&& &&\inst{add} &&r_{p} && r_{block} && && \text{; Copy address of the current block to p}\\
    &\textbf{(31)}\quad&& &&\inst{sub} &&r_{block}\ && r_{p} && && \text{; Clear $r_{block}$}\\
    &\textbf{(32)}\quad&& &&\inst{exch} &&r_{tmp} && r_{p} && && \text{; Load address of next block}\\
    &\textbf{(33)}\quad&& &&\inst{exch} &&r_{tmp} && r_{fl} && && \text{; Set address of next block as new head of free list}\\
    &\textbf{(34)}\quad&& &&\inst{xor} &&r_{tmp} && r_{p} && && \text{; Clear address of next block}\\
    &\textbf{(35)}\quad&& &&\inst{xori} &&r_{t2} && 1 && && \text{; Set $r_{t2} = 1$}\\
    &\textbf{(36)}\quad&&i_{assert_t}\ \texttt{:}\quad &&\inst{bra} &&i_{assert} \span\omit\span\quad \span\omit\span\quad && \text{; Jump}\\
    &\textbf{(37)}\quad&&i_{test_f}\ \texttt{:}\quad &&\inst{bra} &&i_{test} \span\omit\span\quad \span\omit\span\quad && \text{; Receive jump}\\
    &\textbf{(38)}\quad&& &&\inst{addi} &&r_{c} && 1 && && \text{; $Counter\texttt{++}$}\\
    &\textbf{(39)}\quad&& &&\inst{rl} &&r_{sc}\ && 1 && && \text{; Call $double(c_{size}$)}\\
    &\textbf{(40)}\quad&& &&\cdots\cdots && && && &&\text{; Code for pushing temp reg values to stack}\\
    &\textbf{(41)}\quad&& &&\inst{bra}\quad &&malloc1_{entry} \span\omit\span\quad \span\omit\span\quad && \text{; Call $malloc1()$)}\\
    &\textbf{(42)}\quad&& &&\cdots\cdots && && && &&\text{; Inverse of \textbf{(40)}}\\
    &\textbf{(43)}\quad&& &&\inst{rr} &&r_{sc}\ && 1 && && \text{; Inverse of \textbf{(39)}}\\
    &\textbf{(44)}\quad&& &&\inst{subi} &&r_{c} && 1 && && \text{; Inverse of \textbf{(38)}}\\
    &\textbf{(45)}\quad&& &&\inst{xor} &&r_{tmp} && r_p && && \text{; Copy current address of p}\\
    &\textbf{(46)}\quad&& &&\inst{exch} &&r_{tmp} && r_{fl} && && \text{; Store current address of p in current free list}\\
    &\textbf{(47)}\quad&& &&\inst{add} &&r_{p} && r_{cs} && && \text{; Split block by setting p to second half of current block}\\
    &\textbf{(48)}\quad&&i_{assert}\ \texttt{:}\quad &&\inst{bne} &&r_{t2} && r_0 && i_{assert_t} && \text{; Receive jump}\\
    &\textbf{(49)}\quad&& &&\inst{exch} &&r_{tmp} && r_{fl} && && \text{; Load address of head of current free list}\\
    &\textbf{(50)}\quad&& &&\inst{sub} &&r_{p} && r_{cs} && && \text{; Set p to previous block address}\\
    &\textbf{(51)}\quad&& &&\cdots\cdots && && && &&\text{; Code for $r_{e2_{i1}}\ \leftarrow\ \llbracket p - c_{size} \neq addr(freelists[counter])\rrbracket$}\\
    &\textbf{(52)}\quad&& &&\cdots\cdots && && && &&\text{; Code for $r_{e2_{i2}}\ \leftarrow\ \llbracket addr(freelists[counter]) = 0 \rrbracket$}\\
    &\textbf{(53)}\quad&& &&\cdots\cdots && && && &&\text{; Code for $r_{e2_{i3}}\ \leftarrow\ \llbracket (p - c_{size} \neq addr(freelists[counter])) \vee (addr(freelists[counter]) = 0) \rrbracket$}\\
    &\textbf{(54)}\quad&& &&\inst{xor} &&r_{r2} && r_{e2_{i3}} && && \text{; Copy value of $r_{e2_{i3}}$ into $r_{t2}$}\\
    &\textbf{(55)}\quad&& &&\cdots\cdots && && && &&\text{; Inverse of \textbf{(53)}}\\
    &\textbf{(56)}\quad&& &&\cdots\cdots && && && &&\text{; Inverse of \textbf{(52)}}\\
    &\textbf{(57)}\quad&& &&\cdots\cdots && && && &&\text{; Inverse of \textbf{(51)}}\\
    &\textbf{(58)}\quad&& &&\inst{add} &&r_{p} && r_{cs} && && \text{; Inverse of \textbf{(50)}}\\
    &\textbf{(59)}\quad&& &&\inst{exch} &&r_{tmp} && r_{fl} && && \text{; Inverse of \textbf{(49)}}\\
    &\textbf{(60)}\quad&&o_{assert}\ \texttt{:}\quad &&\inst{bne} &&r_{t} && r_0 && o_{assert_t} && \text{; Receive jump}\\
    &\textbf{(61)}\quad&& &&\cdots\cdots && && && &&\text{; Code for $r_{e2_o}\ \leftarrow\ \llbracket c_{size} < object_{size} \rrbracket$}\\
    &\textbf{(62)}\quad&& &&\inst{xor}\quad &&r_t && r_{e2_o} && &&\text{; Copy value of $c_{size} < object_{size}$ into $r_t$}\\        
    &\textbf{(63)}\quad&& &&\cdots\cdots && && && &&\text{; Inverse of \textbf{(61)}}\\ 
    &\textbf{(64)}\quad&&malloc1_{bot}\ \texttt{:}\quad  &&\inst{bra}\quad &&malloc1_{top} \span\omit\span\quad \span\omit\span\quad &&\text{; Jump}\\
\end{alignat*}

}%

    \newpage
    \newcommand{\source}[1]{
    \lstinputlisting[style = basic, language = haskell]{#1}
}
\lstset{
    inputpath = {./},
    mathescape=false, 
    texcl=false,
    inputencoding=latin1,
    breaklines=true
}

\chapter{\textsc{Rooplppc} Source Code}
\label{app:rooplc-source}

\section*{README.md}
\begin{lstlisting}[style = basic]
# ROOPLPPC

*ROOPLPP* is a compiler translating source code written in **Reversible Object Oriented Programming Language++** (ROOPL++) to the reversible assembly language Pendulum ISA (PISA).

The compiler is to be considered proof-of-concept in connection with my Master's Thesis on the ROOPL++ language.

## Requirements
ROOPLPPC uses [Stack](https://docs.haskellstack.org/en/stable/README/) to manage all dependencies and requirements.

## Building
Simply invoke
```
stack build
```
which compiles an executable into the `.stack-work` folder

## Usage
To compile a ROOPL++ program simply run
```
stack exec ROOPLPPC input.rplpp
```
which compiles the input program into Pisa and stores the compiled file as `input.pal` in the current directory.

To specify an output file name, simply provide it as an additional argument
```
stack exec ROOPLPP input.rplpp output.pal
```

## Examples
To see usage examples, please refer to `test/` for example programs. 

## Running compiled programs
The PendVM simulator executes compiled Pisa code and is hosted on Github [here](https://github.com/TueHaulund/PendVM).
\end{lstlisting}

\newpage
\section*{AST.hs}
\source{code/ROOPLPPC/src/AST.hs}

\newpage
\section*{PISA.hs}
\source{code/ROOPLPPC/src/PISA.hs}
 
\newpage
\section*{Parser.hs}
\source{code/ROOPLPPC/src/Parser.hs} 

\newpage
\section*{ClassAnalyzer.hs}
\source{code/ROOPLPPC/src/ClassAnalyzer.hs}

\newpage
\section*{ScopeAnalyzer.hs}
\source{code/ROOPLPPC/src/ScopeAnalyzer.hs}

\newpage
\section*{TypeChecker.hs}
\source{code/ROOPLPPC/src/TypeChecker.hs}

\newpage
\section*{CodeGenerator.hs}
\source{code/ROOPLPPC/src/CodeGenerator.hs}

\newpage
\section*{MacroExpander.hs}
\source{code/ROOPLPPC/src/MacroExpander.hs}

\newpage
\section*{ROOPLPPC.hs} 
\source{code/ROOPLPPC/src/ROOPLPPC.hs}
    \newpage
    \chapter{Example Ouput}
\label{app:example-output}

\section*{LinkedList.rplpp}
\begin{lstlisting}[
    style = basic,
    frame = leftline,
    language = roopl]
class Cell
    Cell next
    int data

    method constructor(int value)
        data ^= value     

    method append(Cell cell)
        if next = nil & cell != nil then
            next <=> cell                   // Store as next cell if current cell is end of list
        else skip
        fi next != nil & cell = nil

        if next != nil then
            call next::append(cell)         // Recusively search until we reach end of list
        else skip
        fi next != nil
               
class LinkedList
    Cell head
    int listLength

    method insertHead(Cell cell)
        if head = nil & cell != nil then
            head <=> cell               // Set cell as head of list if list is empty
        else skip
        fi head != nil & cell = nil

    method appendCell(Cell cell)
        call insertHead(cell)           // Insert as head if empty list

        if head != nil then
            call head::append(cell)     // Iterate list until we reach the end, then insert the node
        else skip
        fi head != nil

        listLength += 1                 // Increment lenght

    method prependCell(Cell cell)
        call insertHead(cell)           // Insert as head if empty list

        if cell != nil & head != nil then
            call cell::append(head)     // Set cell.next = head. head = nil after execution
        else skip
        fi cell != nil & head = nil

        if cell != nil & head = nil then
            cell <=> head               // Set head = cell. Cell is nil after execution
        else skip
        fi cell = nil & head != nil

        listLength += 1                 // Increment length

    method length(int result)
        result ^= listLength    

class Program
    LinkedList linkedList
    int sumResult
    int listLength

    method main()
        new LinkedList linkedList           // Init new linked linkedList
        listLength += 10

        local int x = 0
        from x = 0 do
            skip
        loop
            local Cell cell = nil
            new Cell cell                       // Instantiate new cell
            call cell::constructor(x)           // Set value of cell
            call linkedList::appendCell(cell)   // Append it to the linkedList
            delocal Cell cell = nil
            x += 1                      
        until x = listLength
        delocal int x = listLength
          
\end{lstlisting}

\newpage
\section*{LinkedList.pal}
\lstinputlisting[
    style = basic,
    frame = leftline,
    multicols=2,
    language = pisa]{code/ROOPLPPC/test/LinkedList.pal}

\newpage
\section*{BinaryTree.rplpp}
\begin{lstlisting}[
    style = basic,
    frame = leftline,
    language = roopl]
class Node
    Node left
    Node right
    int value

    method setValue(int newValue)
        value ^= newValue 

    method insertNode(Node node, int nodeValue)
        if nodeValue < value then                       // Determine if we insert left or right
            if left = nil & node != nil then
                left <=> node                           // If open left node, store here
            else skip
            fi left != nil & node = nil

            if left != nil then
                call left::insertNode(node, nodeValue)  // If current node has left, continue iterating
            else skip
            fi left != nil
        else
            if right = nil & node != nil then
                right <=> node                          // If open right node spot, store here
            else skip
            fi right != nil & node = nil

            if right != nil then
                call right::insertNode(node, nodeValue) // If current node has, continue searching
            else skip
            fi right != nil
        fi nodeValue < value

    method getSum(int result)
        result += value                  // Add the value of this node to the sum   

        if left != nil then
            call left::getSum(result)   // If we have a left child, follow that path
        else skip                       // Else, skip
        fi left != nil

        if right != nil then
            call right::getSum(result)  // If we have a right child, follow that path
        else skip                       // Else, skip
         fi right != nil

    method mirror()
        left <=> right                  // Swap left and right children

        if left = nil then skip
        else call left::mirror()        // Recursively swap children if left != nil
        fi left = nil

        if right = nil then skip
        else call right::mirror()       // Recursively swap children if right != nil
        fi right = nil                 
    
class Tree
    Node root
    
    method insertNode(Node node, int value)
        if root = nil & node != nil then
            root <=> node
        else skip
        fi root != nil & node = nil

        if root != nil then
            call root::insertNode(node, value)
        else skip
        fi root != nil

    method sum(int result)
        if root != nil then
            call root::getSum(result)
        else skip
        fi root != nil

    method mirror()
        if root != nil then
            call root::mirror()
        else skip
        fi root != nil

class Program
    int sumResult
    Tree tree
    int nodeCount

    method main()
        new Tree tree 
        nodeCount += 3
        
        local int x = 0
        from x = 0 do
            skip
        loop
            local Node node = nil
            new Node node                   // Init new node
            call node::setValue(x)          // Set node value
            call tree::insertNode(node, x)  // Insert node in tree
            delocal Node node = nil
            x += 1
        until x = nodeCount
        delocal int x = nodeCount

        call tree::sum(sumResult)
        call tree::mirror()
\end{lstlisting}

\newpage
\section*{BinaryTree.pal}
\lstinputlisting[
    style = basic,
    frame = leftline,
    multicols=2,
    language = pisa]{code/ROOPLPPC/test/BinaryTree.pal}

\newpage
\section*{DoublyLinkedList.rplpp}
\begin{lstlisting}[
    style = basic,
    frame = leftline,
    language = roopl]
class Cell
    int data
    int index
    Cell left
    Cell right
    Cell self

    method setData(int value)
        data ^= value

    method setIndex(int i)
        index ^= i    

    method setLeft(Cell cell)
        left <=> cell

    method setRight(Cell cell)
        right <=> cell

    method setSelf(Cell cell)
        self <=> cell

    method append(Cell cell)
        if right = nil & cell != nil then   // If current cell does not have a right neighbour
            right <=> cell                  // Set new cell as right neighbour of current cell
        
            local Cell selfCopy = nil      
            copy Cell self selfCopy         // Copy reference to current cell
            call right::setLeft(selfCopy)   // Set current cell as left neighbour of newly added right neighbour
            delocal Cell selfCopy = nil

            local int cellIndex = index + 1
            call right::setIndex(cellIndex) // Set cell index in newly added right neightbour of current cell
            delocal int cellIndex = index + 1
        else skip
        fi right != nil & cell = nil

        if right != nil then
            call right::append(cell)        // Keep searching for empty right neighbour
        else skip
        fi right != nil
        
class DoublyLinkedList
    Cell head
    int length

    method appendCell(Cell cell)
        if head = nil & cell != nil then
            head <=> cell
        else skip
        fi head != nil & cell = nil

        if head != nil then 
            call head::append(cell)
        else skip
        fi head != nil

        length += 1

class Program
    DoublyLinkedList list
    int listLength

    method main()
        new DoublyLinkedList list
        listLength += 10

        local int x = 0
        from x = 0 do skip
        loop
            local Cell cell = nil
                new Cell cell

                local Cell cellCopy = nil
                copy Cell cell cellCopy
                call cell::setSelf(cellCopy)
                delocal Cell cellCopy = nil

                call cell::setData(x)
                call list::appendCell(cell)
            delocal Cell cell = nil
            x += 1
        until x = listLength
        delocal int x = listLength 
\end{lstlisting}

\newpage
\section*{DoublyLinkedList.pal}
\lstinputlisting[
    style = basic,
    frame = leftline,
    multicols=2,
    language = pisa]{code/ROOPLPPC/test/DoublyLinkedList.pal}   

\newpage
\section*{RTM.rplpp}
\begin{lstlisting}[
    style = basic,
    frame = leftline,
    language = roopl]
class Cell
    Cell self
    Cell right
    Cell left
    int data
    
    method getLeft(Cell cell)
        right <=> cell

    method getRight(Cell cell)
        left <=> cell
    
    method getSelf(Cell cell)
        self <=> cell

    method getSymbol(int symbol)
        symbol <=> data

class RTM
    Cell tapeHead
    int[] q1
    int[] q2
    int[] s1
    int[] s2
    int SLASH
    int LEFT
    int RIGHT
    int BLANK
    int state
    int Qs
    int Qf
    int symbol
    int PC_MAX
    int pc

    method initLiterals()
        // Initialize string literals
        SLASH += 9999
        LEFT  += 9998
        RIGHT += 9997
        BLANK += 9996

        // Set max program counter
        PC_MAX += 7

    method initRules()
        // Initialize transition rule arrays
        new int[8] q1
        new int[8] q2
        new int[8] s1
        new int[8] s2

        // Define transition rules for binary number incrementation
        q1[0] += 1
        s1[0] += BLANK
        s2[0] += BLANK
        q2[0] += 2

        q1[1] += 2
        s1[1] += SLASH
        s2[1] += RIGHT
        q2[1] += 3

        q1[2] += 3
        s1[2] += 0
        s2[2] += 1
        q2[2] += 4

        q1[3] += 3
        s1[3] += 1
        s2[3] += 0
        q2[3] += 2

        q1[4] += 3
        s1[4] += BLANK
        s2[4] += BLANK
        q2[4] += 4

        q1[5] += 4
        s1[5] += SLASH
        s2[5] += LEFT
        q2[5] += 5

        q1[6] += 5
        s1[6] += 0
        s2[6] += 0
        q2[6] += 4

        q1[7] += 5
        s1[7] += BLANK
        s2[7] += BLANK
        q2[7] += 6

    method initTape()
        local Cell cell0 = nil
        local Cell cell1 = nil
        local Cell cell2 = nil
        local Cell cell3 = nil
        local Cell cell4 = nil

        // Init cells
        new Cell cell0
        new Cell cell1
        new Cell cell2
        new Cell cell3
        new Cell cell4

        // Write 1 1 0 1 on tape
        symbol += BLANK
        uncall cell0::getSymbol(symbol)
        symbol += 1
        uncall cell1::getSymbol(symbol)
        symbol += 1
        uncall cell2::getSymbol(symbol)
        symbol += 1
        uncall cell4::getSymbol(symbol)

        // Set tape head
        tapeHead <=> cell0

        // Set self pointers
        copy Cell tapeHead cell0
        uncall tapeHead::getSelf(cell0)
        copy Cell cell1 cell0
        uncall cell1::getSelf(cell0)
        copy Cell cell2 cell0
        uncall cell2::getSelf(cell0)
        copy Cell cell3 cell0
        uncall cell3::getSelf(cell0)
        copy Cell cell4 cell0
        uncall cell4::getSelf(cell0)

        // Link cell 3 and 4
        copy Cell cell3 cell0
        uncall cell4::getLeft(cell0)
        uncall cell3::getRight(cell4)

        // Link cell 2 and 3
        copy Cell cell2 cell0
        uncall cell3::getLeft(cell0)
        uncall cell2::getRight(cell3)

        // Link cell1 and cell 2
        copy Cell cell1 cell0
        uncall cell2::getLeft(cell0)
        uncall cell1::getRight(cell2)

        // Link tapeHead and cell 1
        copy Cell tapeHead cell0
        uncall cell1::getLeft(cell0)
        uncall tapeHead::getRight(cell1)

        delocal Cell cell4 = nil
        delocal Cell cell3 = nil
        delocal Cell cell2 = nil
        delocal Cell cell1 = nil
        delocal Cell cell0 = nil

    method init()
        // Prepare for simulation
        call initLiterals()
        call initRules()
        call initTape()

        // Init pc, start and finishing state
        state += 1
        Qs += 1
        Qf += 6

        // Start simulation
        call simulate()

    method simulate() 
        from state = Qs do
            call tapeHead::getSymbol(symbol)    // Fetch current symbol 
            call inst()
            uncall tapeHead::getSymbol(symbol)  // Zero-clear symbol
            pc += 1                             // Increment pc

            if pc = PC_MAX then                 // Reset pc 
                pc ^= PC_MAX
            else skip
            fi pc = 0
        loop skip
        until state = Qf
 
    method inst()
        if state = q1[pc] & symbol = s1[pc] then    // Symbol rule:
            state += q2[pc]-q1[pc]                  // set state to q2[pc]
            symbol += s2[pc]-s1[pc]                 // set symbol to s2[pc]
        else skip
        fi state = q2[pc] & symbol = s2[pc] 
        if state = q1[pc] & s1[pc] = SLASH then     // Move rule:
            state += q2[pc]-q1[pc]                  // set state to q2[pc] 
            if s2[pc] = RIGHT then  
                call moveRight()                    // Move tape head right
            else skip 
            fi s2[pc] = RIGHT   
            if s2[pc] = LEFT then   
                uncall moveRight()                  // Move tape head left
            else skip
            fi s2[pc] = LEFT
        else skip
        fi state = q2[pc] & s1[pc] = SLASH

    method moveRight()
        local Cell right = nil
        local Cell tmp = nil
        uncall tapeHead::getSymbol(symbol)   // Put symbol back in current cell
        call tapeHead::getRight(right)       // Get right neighbour

        if right = nil & symbol = BLANK then
            symbol ^= BLANK                  // Zero clear symbol
            new Cell right                   // Init new neighbour
            copy Cell right tmp              // Copy reference to self
            uncall right::getSelf(tmp)       // Store self reference
            uncall right::getLeft(tapeHead)  // Set tape head as left of new cell
            right <=> tapeHead 
        else          
            call right::getLeft(tmp)         // Get copy of tape head reference
            uncopy Cell tmp tapeHead         // Clear reference to tape head

            if tapeHead = nil & symbol = BLANK then
                call tmp::getSelf(tapeHead)  // rev: set self pointer
                uncopy Cell tmp tapeHead     // rev: new self pointer
                delete Cell tmp              // rev: new left neighbour
                symbol ^= BLANK
            else skip                        // In reverse:
            fi tmp = nil                     // Allocate new left if current is nil
            
            uncall right::getLeft(tmp)       // Put tape head reference back
            tapeHead <=> right
            call tapeHead::getRight(right)   // Get right of new tape head
            call tapeHead::getSymbol(symbol) // Get symbol of new tape head
        fi right = nil

        uncall tapeHead::getRight(right)     // Set right neighbour
        delocal Cell right = nil
        delocal Cell tmp = nil

class Program
    RTM bni 

    method main()
        // This program contains a RTM implementing 
        // incrementation of a non-negative n-bit binary number by 1 (modulo 2n). 
        // The tape is initialized with | b | 1 | 1 | 0 | 0 | and after execution,
        // the tape is left with | b | 0 | 0 | 1 | 1 |
        new RTM bni
        call bni::init()
\end{lstlisting}

\newpage
\section*{RTM.pal}
\lstinputlisting[
    style = basic,
    frame = leftline,
    multicols=2,
    language = pisa]{code/ROOPLPPC/test/RTM.pal} 

\end{appendices} 

\end{document}